\newif\iflong
\newif\ifshort

\iflong
\else
\shorttrue
\fi
\documentclass[a4paper,UKenglish]{article}

\usepackage{microtype}

\usepackage[textsize=footnotesize,disable]{todonotes}
\usepackage{fullpage}
\usepackage{authblk}

\usepackage[pagebackref]{hyperref}

\usepackage{amsthm}

\usepackage{amsthm}
\usepackage{amsmath}
\usepackage{amssymb}
\usepackage{comment}
\usepackage{amsfonts}
\usepackage{mathtools}

\usepackage[capitalize,sort&compress,nameinlink,noabbrev]{cleveref}
\usepackage{tabto}
\NumTabs{10}

\usepackage[ruled,vlined,linesnumbered]{algorithm2e}

\usepackage{tikz}
\usetikzlibrary{positioning,backgrounds,patterns,calc,shapes, decorations.pathmorphing,plotmarks}
\tikzstyle{nnode} = [draw, circle, inner sep=1pt, fill=black]
\tikzstyle{bnode} = [draw, circle, inner sep=5pt]
\tikzstyle{deletede} = [thick, dash pattern=on 2.5pt off 1pt, decorate, decoration={snake,amplitude=.4mm,segment length=1.6mm}, draw=red, yshift=-5ex]
\tikzstyle{added} = [thick, decorate, decoration={snake,amplitude=.3mm,segment length=1mm}, draw=darkgreen]
\tikzstyle{deletedv} = [inner sep=1.5pt, draw=darkgreen, fill=darkgreen]

\usepackage[absolute]{textpos}  % For ERC acknowledgment with flags.

\theoremstyle{plain}
\newtheorem{theorem}{Theorem}{\bfseries}{\normalfont}
\newtheorem{lemma}{Lemma}{\bfseries}{\normalfont}
\newtheorem{obs}{Observation}{\bfseries}{\normalfont}
{\bfseries}{\normalfont}
\newtheorem{prop}{Proposition}{\bfseries}{\normalfont}
\newtheorem{lem}{Lemma}{\bfseries}{\normalfont}

\theoremstyle{definition}
\newtheorem{rrule}{Reduction Rule}{\bfseries}{\normalfont}
\newtheorem{brule}{Branching Rule}{\bfseries}{\normalfont}
\newtheorem*{kernel}{Modified Kernelization Algorithm \boldmath$K$}{\bfseries}{\normalfont}
\newtheorem*{grule}{Greedy Rule}{\bfseries}{\normalfont}
\newtheorem*{crule}{Clean-up Rule}{\bfseries}{\normalfont}
\newtheorem*{rulezero}{Rule~0}{\bfseries}{\normalfont}

\theoremstyle{remark}
\newtheorem{remark}{Remark}{\sffamily}{\normalfont}
{\sffamily}{\normalfont}

\crefname{claim}{Claim}{Claims}

\theoremstyle{definition}
\newtheorem{defi}{Definition}{\bfseries}{\normalfont}

\newcommand{\decprob}[3]{%
  \begin{center}%
    \begin{minipage}{0.9\linewidth}%
      \textsc{#1}\\
      \textbf{Input:} #2\\
      \textbf{Question:} #3
    \end{minipage}%
  \end{center}%
}

\newenvironment{lemenum}{\begin{compactenum}[(i)]}{\end{compactenum}}

\newcommand{\G}{\ensuremath{\mathcal{G}}}

\usepackage[numbers,sort]{natbib}

\makeatletter
\def\NAT@spacechar{~}
\makeatother

\usepackage{paralist}

\usepackage{enumerate}

\usepackage{xcolor}
\definecolor{lightgray}{rgb}{0.8, 0.8, 0.8}
\definecolor{darkgreen}{rgb}{0.01,0.6,0.1}

\newcommand{\proofparagraph}[1]{\par\smallskip\textit{#1}}
\newcommand{\proofparagraphm}[1]{\textit{#1}}

\newcommand{\MLCE}{\textsc{MLCE}}
\newcommand{\TCE}{\textsc{TCE}}
\newcommand{\MLCELong}{\textsc{Multi-Layer Cluster Editing}}
\newcommand{\TCELong}{\textsc{Temporal Cluster Editing}}

\newcommand{\MIS}{\textsc{Multicolored Independent Set}\xspace}

\newcommand{\CE}{\textsc{Cluster Editing}}

\newcommand{\coNP}{\ensuremath{\textsf{coNP}}\xspace}
\newcommand{\NP}{\ensuremath{\textsf{NP}}}
\newcommand{\XP}{\ensuremath{\textsf{XP}}\xspace}
\newcommand{\FPT}{\ensuremath{\textsf{FPT}}}
\newcommand{\W}[1]{\ensuremath{\textsf{W[#1]}}\xspace}
\newcommand{\paraNP}{\ensuremath{\textsf{para-NP}}}

\crefname{rrule}{Rule}{Rules}

\DeclareRobustCommand{\NoKernelAssume}{$\NP\subseteq \textsf{\coNP/poly}$}

\newcommand{\N}{\mathbb{N}\xspace}

\newcommand{\M}{\mathcal{M}}

\graphicspath{{images/}}

\usepackage{etoolbox}

\newcommand{\appsymb}{$\star$}
\newcommand{\appref}[1]{%
}

\newcommand{\oldstuff}[1]{%
}

\newcommand{\appendixproof}[2]{%
#2
}

\newcommand{\appendixcorrectnessproof}[2]{%
#2
}

\newcommand{\appendixsection}[1]{%
}

\makeatletter \newcommand{\gettikzxy}[3]{%
  \tikz@scan@one@point\pgfutil@firstofone#1\relax \edef#2{\the\pgf@x}%
  \edef#3{\the\pgf@y}%
} \makeatother

\usepackage{etoolbox}

\begin{document}

\title{Cluster Editing for Multi-Layer and Temporal~Graphs\footnote{OS
    supported by grant 17-20065S of the Czech Science Foundation. JC
    and MS supported by the People Programme (Marie Curie Actions) of
    the European Union's Seventh Framework Programme (FP7/2007-2013)
    under REA grant agreement number {631163.11}, the Israel Science
    Foundation (grant no. 551145/14), and by the European Research
    Council (ERC) under the European Union’s Horizon 2020 research and
    innovation programme under grant agreement numbers~677651 (JC)
    and~714704 (MS). Main work of JC and MS done while with Dept.\
    Industrial Engineering and Management, Ben-Gurion University of
    the Negev, Beer Sheva, Israel. HM supported by the DFG, project
    MATE (NI 369/17). HM supported by the DFG, project MATE (NI
    369/17). This work was initiated at the research retreat of the TU
    Berlin Algorithmics and Computational Complexity group held in
    Boiensdorf (Baltic Sea), in April~2017. An extended abstract of
    this paper is accepted to appear in the proceedings of the 29th
    International Symposium on Algorithms and Computation
    (ISAAC~'18)~\cite{chen2018parameterized}}}

\author[1]{Jiehua~Chen}

% \affil[1]{Department of Industrial Engineering and Management, Ben-Gurion University of the Negev, Beer Sheva, Israel,
% \texttt{jiehua.chen2@gmail.com, sorge@post.bgu.ac.il}}
\affil[1]{Faculty of Mathematics, Informatics and Mechanics, University of Warsaw, Warsaw, Poland, \texttt{jiehua.chen2@gmail.com, manuel.sorge@mimuw.edu.pl}}

\author[2]{Hendrik~Molter}

\affil[2]{Institut f\"ur Softwaretechnik und Theoretische Informatik,
 TU Berlin, Germany,
 \texttt{h.molter@tu-berlin.de}}

\author[1]{Manuel~Sorge}

\author[3]{Ond\v rej~Such\'{y}}
\affil[3]{Faculty of Information Technology, Czech Technical University in Prague, Prague, Czech~Republic, \texttt{ondrej.suchy@fit.cvut.cz}}
% \begin{bibunit}

\maketitle

\begin{textblock}{20}(.1, 12.5)
\includegraphics[width=50px]{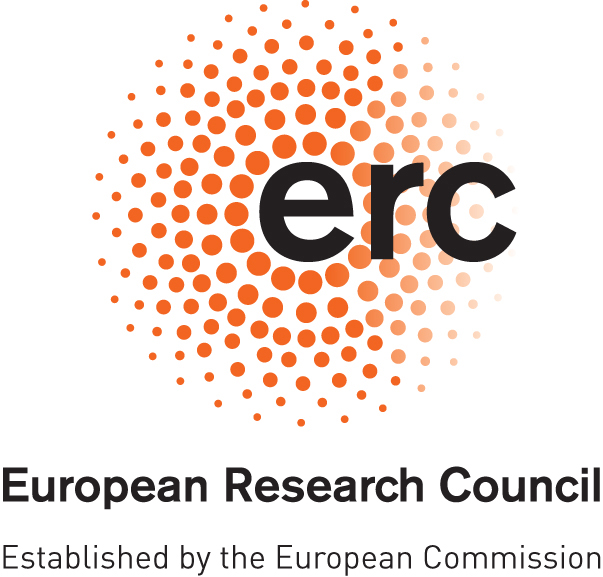}% 
\end{textblock}
\begin{textblock}{20}(.1, 13.5)
\includegraphics[width=50px]{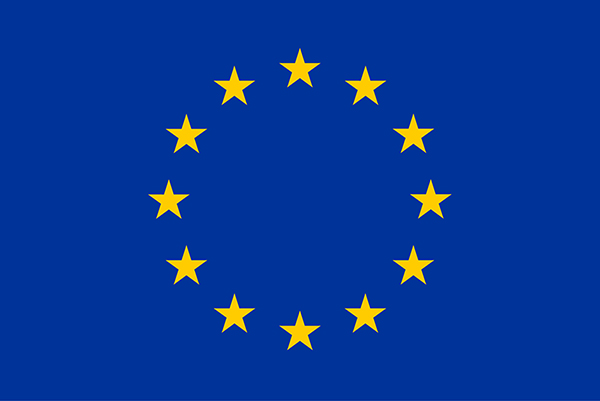}% 
\end{textblock}

\begin{abstract}
%   \looseness=-1 
  Motivated by the recent rapid growth of research for
  algorithms to cluster multi-layer and temporal graphs, we study
  extensions of the classical \textsc{Cluster Editing} problem. 
  In \textsc{Multi-Layer Cluster Editing} we receive a set of graphs
  on the same vertex set, called \emph{layers} and aim to transform
  all layers into cluster graphs (disjoint unions of cliques) that
  differ only slightly. More specifically,
%we aim to obtain the same cluster graph by the following editing:
%Mark a subset~$V'$ of at most $d$~vertices, modify at most $k$ (possibly different) edges in each layer, and then delete the vertices from $V'$. 
  we want to mark at most $d$ vertices and to transform each
  layer % of the multi-layer graph
  into a cluster graph using at most $k$ edge additions or deletions
  per layer so that, if we remove the marked vertices, we obtain the
  same cluster graph in all layers. In \TCELong\ we receive a
  \emph{sequence} of layers and we want to transform each layer into a
  cluster graph so that consecutive layers differ only slightly. That
  is, we want to transform each layer into a cluster graph with at
  most~$k$ edge additions or deletions and to mark a distinct set
  of~$d$ vertices in each layer so that each two consecutive layers
  are the same after removing the vertices marked in the first of
  the two layers. We study the combinatorial structure of the two problems
  via their parameterized complexity with respect to the
  parameters~$d$ and~$k$, among others. Despite the similar
  definition, the two problems behave quite differently: In
  particular, \MLCELong\ is fixed-parameter tractable with running
  time
%  in $f(k,d)\cdot |I|^{O(1)}$~time for any input instance~$I$,
 % whereas an algorithm with such running time for \TCE\ is unlikely to exists, even if $d=1$.
  $k^{O(k + d)} s^{O(1)}$ for inputs of size~$s$, whereas
  \TCELong\ is \W{1}-hard with respect to~$k$ even if~$d = 3$.
\end{abstract}
% \looseness=-1
\section{Introduction}\label{sec:intro}
\CE\ and its weighted form \textsc{Correlation Clustering} are two
important and well-studied models of graph
clustering~\cite{bansal2004correlation,bocker_cluster_2013,cao_cluster_2012,fomin_tight_2014,komusiewicz_cluster_2012}.
In the former, we are given a graph and we aim to edit (that is, add
or delete) the fewest number of edges in order to obtain a
\emph{cluster graph}, a graph in which each connected component is a
clique. \CE\ has attracted a lot of attention from a
parameterized-algorithms point of view
(e.g.~\cite{bocker_cluster_2013,cao_cluster_2012,fomin_tight_2014,komusiewicz_cluster_2012,LMNN18})
and the resulting contributions have found their way back into
practice~\cite[Section 6]{bocker_cluster_2013}.

% \looseness=-1 
Meanwhile, additional information is now available and used in clustering
methods. In particular, research on clustering so-called multi-layer
and temporal graphs grows rapidly
(e.g.~\cite{kim_community_2015,tagarelli_ensemble-based_2017,tang_clustering_2009,tantipathananandh_finding_2011,tantipathananandh_framework_2007}).
A \emph{multi-layer graph} is a set of
graphs, called
\emph{layers}, on the same vertex set~\cite{CY14,kim_community_2015,Kiv+14}. In social
networks, a layer can represent social interactions, geographic
closeness, common interests or activities~\cite{kim_community_2015}.\footnote{When considering the activity in
different communities, we typically obtain a large number of layers~\cite{nicosia_measuring_2015}.} A \emph{temporal
  graph} is a multi-layer graph in which the layers are ordered linearly~\cite{holme2015,holme2012temporal,latapy2017stream,michail2016introduction,tantipathananandh_finding_2011,tantipathananandh_framework_2007}. 
Temporal graphs naturally model the evolution of relationships of
individuals over time or their set of time-stamped interactions%  between
% them
.

The goals in clustering multi-layer and temporal graphs are,
respectively, to find a
%the parameterized complexity of clustering
%problems in which the input is multi-layered and we ought to find a
clustering that is consistent with all
layers~\cite{kim_community_2015,Kiv+14,tagarelli_ensemble-based_2017,tang_clustering_2009} 
or a clustering that slowly evolves over time consistently with the
graph% evolves over the sequence of layers, closely matching each
% individual layer
~\cite{tantipathananandh_finding_2011,tantipathananandh_framework_2007}.
% Our work adds to the nascent field of theory and algorithms for
% multi-layer and temporal graphs.
The methods used herein are often heuristic and beyond observing
NP-hardness, to the best of our knowledge, there is no deeper analysis
of the complexity of the general underlying computational problems
that are attacked in this way. Hence, there is also a lack of
knowledge about the possible avenues for algorithmic tractability. We
initiate this research here.

% \looseness=-1 % We investigate generalized variants of the
% \textsc{Cluster Editing} problem in which the input graph is
% multi-layered or may evolve over time.
We analyze the combinatorial structure behind cluster editing for
multi-layer and temporal graphs, defined formally below, via studying
their parameterized complexity with respect to the most basic
parameters, such as the number of edits. That is, we aim to find
\emph{fixed-parameter algorithms} (\FPT), which have running time
$f(p)\cdot \ell^{O(1)}$ where $p$ is the parameter and $\ell$ the
input length, or to show \W{1}-hardness, which indicates that there
cannot be such algorithms.\todo{hm: we do not introduce \XP, \paraNP, or kernels so far..}

As we will see, both problems offer rich interactions between the
layers on top of the structure inherited from \CE. Our main
contributions are an intricate fixed-parameter algorithm for
multi-layer cluster editing, whose underlying techniques should be
applicable to a broader range of multi-layer problems, and a hardness
result for temporal cluster editing, which shows that certain
non-local structures harbor algorithmic intractability.

\subparagraph{\TCELong\ (\TCE).} Berger-Wolf and
Tantipathananandh~\cite{tantipathananandh_finding_2011} were motivated
by cluster detection problems from practice to study the following
problem. Given a temporal graph, edit each layer into a cluster graph,
that is, add or remove edges such that the layer becomes a disjoint
union of cliques. Below we will also call the connected components of
a cluster graph \emph{clusters}. The goal is to minimize the sum of
the total number of edits and the number of vertices moving between
different clusters in two consecutive layers. \TCE\ is a variant of
this problem where we instead minimize the layer-wise maxima of the
number of edits and moving vertices, respectively. The problem can be
formalized as follows.

Let $\G = (G_i)_{i \in [\ell]}$ be a temporal graph with vertex
set~$V$, that is, $G_i$ is the $i$th layer, and let $E_i$ be the edge
set of~$G_i$. (By $[\ell]$ we denote the set~$\{1, \ldots, \ell\}$ for
$\ell \in \mathbb{N}$.) An \emph{edge modification} or \emph{edit} for
a graph~$G = (V, E)$ is an unordered pair of vertices from~$V$. Let
$M$ be a set of edits for~$G$. If the graph $G' = (V, E \oplus M)$ is
a cluster graph, then we say that $M$ is a \emph{cluster editing set}
for~$G$. Herein, $\oplus$ denotes the symmetric difference:
$A \oplus B = (A\setminus B)\cup (B\setminus A)$. A \emph{clustering} for~$\G$ is a sequence
$\M = (M_i)_{i \in [\ell]}$ of edge modification sets such that $M_i$
is a cluster editing set for layer~$G_i$. Intuitively, sets $M_i$
contain the data that we need to disregard in order to cluster our
input and hence we want to minimize their
sizes~\cite{tantipathananandh_framework_2007,tantipathananandh_finding_2011}.
For that, we say that $\M$ is \emph{$k$-bounded} for some
integer~$k \in \mathbb{N}$ if $|M_i| \leq k$ for each~$i \in [\ell]$.

A fundamental property of clusterings of temporal graphs is their
evolution over time. In practice, these clusterings evolve only slowly
as measured by the number of vertices switching between clusters from
one layer to
another~\cite{tantipathananandh_framework_2007,tantipathananandh_finding_2011}.
This requirement can be formalized as follows. Let $d \in \mathbb{N}$. Clustering~$\M$ for $\G$ (as above) is \emph{temporally $d$-consistent} if there exists a sequence
$(D_i)_{i \in [\ell - 1]}$ of vertex sets, each of size at most~$d$, such that each pair of consecutive layers is \emph{consistent}, that is
$G_i'[V \setminus D_i] = G_{i + 1}'[V \setminus D_i]$ for each
$i \in [\ell - 1]$. Hence, the sets~$D_i$ contain the vertices
changing clusters. We arrive at the following.

\decprob{\TCELong\ (\TCE)}{A temporal graph $\G$ and two integers $k, d$.}{Is there a temporally $d$-consistent $k$-bounded clustering for~$\G$?}%
\noindent We also say that the corresponding sets $D_i \subseteq V$ and
$M_i \subseteq \binom{V}{2}$ as above form a \emph{solution} and the vertices
in $D_i$ are \emph{marked}. An example is shown in \cref{fig:problem-example}. 

\begin{figure}[h!]
  \centering
  \begin{tikzpicture}
      % Layer 1
    \def \xs {16ex}
    \def \xss {36ex}
    \def \xsss {54ex}
    \def \ys {28ex}
    \def \yss {-60ex}
    \def \xscc {.45}
    \def \yscc {.9}
    \def \scc {.6}
    \def \aa {3ex}
    \def \bb {6ex}
  \begin{scope}
    \foreach \n / \x / \y / \s in  {1/0/0/above, 2/1/1/above,3/1/-1/below,4/2/0/above,5/3/0/above} {%
      \node[nnode] at (\x*\xscc,\y*\yscc) (\n) {};
      \node[\s = 0pt of \n] {$\n$};
    }
    \foreach \s / \t in {1/2,1/3,2/3,3/4,1/4,2/4,4/5} {
      \draw (\s) edge (\t);
    }
    \gettikzxy{(2)}{\twox}{\twoy};
    \gettikzxy{(4)}{\fourx}{\foury};
    \node at ($(\twox/2+\fourx/2, \twoy+20)$) {Layer $1$};
  \end{scope}
  % Layer 2
  \begin{scope}[xshift=\xs]
    \foreach \n / \x / \y / \s in  {1/0/0/above, 2/1/1/above,3/1/-1/below,4/2/0/above,5/3/0/above} {
      \node[nnode] at (\x*\xscc,\y*\yscc) (\n) {};
      \node[\s = 0pt of \n] {$\n$};
    }
    \foreach \s / \t in {2/3,3/4,2/4,4/5} {
      \draw (\s) edge (\t);
    }
    \gettikzxy{(2)}{\twox}{\twoy};
    \gettikzxy{(4)}{\fourx}{\foury};
    \node at ($(\twox/2+\fourx/2, \twoy+20)$) {Layer $2$};
  \end{scope}
  % Layer 3
  \begin{scope}[xshift=\xs*2]
    \foreach \n / \x / \y / \s in  {1/0/0/above, 2/1/1/above,3/1/-1/below,4/2/0/above,5/3/0/above} {
      \node[nnode] at (\x*\xscc,\y*\yscc) (\n) {};
      \node[\s = 0pt of \n] {$\n$};
    }
    \foreach \s / \t in {2/3,3/4,2/4,2/5,3/5} {
      \draw (\s) edge (\t);
    }
    \gettikzxy{(2)}{\twox}{\twoy};
    \gettikzxy{(4)}{\fourx}{\foury};
    \node at ($(\twox/2+\fourx/2, \twoy+20)$) (o3) {Layer $3$};
  \end{scope}

  % Temporal
    \begin{scope}[xshift=\xs*3+5ex]
    \foreach \n / \x / \y / \s in  {1/0/0/above, 2/1/1/above,3/1/-1/below,4/2/0/above,5/3/0/above} {
      \node[nnode] at (\x*\xscc,\y*\yscc) (\n) {};
      \node[\s = 0pt of \n] {$\n$};
    }
    \foreach \s / \t in {1/2,1/3,1/4,2/3,3/4,2/4,4/5} {
      \draw (\s) edge (\t);
    }
    \gettikzxy{(2)}{\twox}{\twoy};
    \gettikzxy{(4)}{\fourx}{\foury};
    \node at ($(\twox/2+\fourx/2, \twoy+20)$) (t11) {Layer $1$};
    
    % Sols
    % Marked vertices D
    \foreach \s in {1} {
      \node[bnode, deletedv] at (\s) {};
    }  
    
    % Deleted edges 
    \foreach \s / \t in {4/5} {
      \draw  (\s) edge[white] (\t);
      \draw  (\s) edge[deletede] (\t);      
    }
  \end{scope}
  \begin{scope}[xshift=\xs*4+5ex]
    \foreach \n / \x / \y / \s in  {1/0/0/above, 2/1/1/above,3/1/-1/below,4/2/0/above,5/3/0/above} {
      \node[nnode] at (\x*\xscc,\y*\yscc) (\n) {};
      \node[\s = 0pt of \n] {$\n$};
    }
    \foreach \s / \t in {2/3,3/4,2/4,4/5} {
      \draw (\s) edge (\t);
    }
    \gettikzxy{(2)}{\twox}{\twoy};
    \gettikzxy{(4)}{\fourx}{\foury};
    \node at ($(\twox/2+\fourx/2, \twoy+20)$) {Layer $2$};
     
    % Sols
    % Marked vertices D
    \foreach \s in {5} {
      \node[bnode, deletedv] at (\s) {};
    }  
       
    % Deleted edges 
    \foreach \s / \t in {4/5} {
      \draw  (\s) edge[white] (\t);
      \draw  (\s) edge[deletede] (\t);      
    }
  \end{scope}
    \begin{scope}[xshift=\xs*5+5ex]
    \foreach \n / \x / \y / \s in  {1/0/0/above, 2/1/1/above,3/1/-1/below,4/2/0/above,5/3/0/above} {
      \node[nnode] at (\x*\xscc,\y*\yscc) (\n) {};
      \node[\s = 0pt of \n] {$\n$};
    }
    \foreach \s / \t in {2/3,3/4,2/4,2/5,3/5} {
      \draw (\s) edge (\t);
    }
    \gettikzxy{(2)}{\twox}{\twoy};
    \gettikzxy{(4)}{\fourx}{\foury};
    \node at ($(\twox/2+\fourx/2, \twoy+20)$) {Layer $3$};   
    %Sol
    % Added edges 
    \foreach \s / \t in {4/5} {
      \draw  (\s) edge[added] (\t);
    }
  \end{scope}

  % ML-1
  % Layer 1
    \begin{scope}[yshift=-\ys]
    \foreach \n / \x / \y / \s in  {1/0/0/above, 2/1/1/above,3/1/-1/below,4/2/0/above,5/3/0/above} {
      \node[nnode] at (\x*\xscc,\y*\yscc) (\n) {};
      \node[\s = 0pt of \n] {$\n$};
    }
    \foreach \s / \t in {1/2,1/3,1/4,2/3,3/4,2/4,4/5} {
      \draw (\s) edge (\t);
    }
    \gettikzxy{(2)}{\twox}{\twoy};
    \gettikzxy{(4)}{\fourx}{\foury};
    \node at ($(\twox/2+\fourx/2, \twoy+20)$) (ml11) {Layer $1$};
    
    % Sols
    % Marked vertices D
    \foreach \s in {1} {
      \node[bnode, deletedv] at (\s) {};
    }  
    
    % Added edges 
    \foreach \s / \t in {2/5,3/5} {
      \draw  (\s) edge[added] (\t);
    }
    \draw (1) edge[added, bend right] (5);
  \end{scope}
  \begin{scope}[xshift=\xs,yshift=-\ys]
    \foreach \n / \x / \y / \s in  {1/0/0/above, 2/1/1/above,3/1/-1/below,4/2/0/above,5/3/0/above} {
      \node[nnode] at (\x*\xscc,\y*\yscc) (\n) {};
      \node[\s = 0pt of \n] {$\n$};
    }
    \foreach \s / \t in {2/3,3/4,2/4,4/5} {
      \draw (\s) edge (\t);
    }
    \gettikzxy{(2)}{\twox}{\twoy};
    \gettikzxy{(4)}{\fourx}{\foury};
    \node at ($(\twox/2+\fourx/2, \twoy+20)$) {Layer $2$};
     
    % Sols
    % Marked vertices D
    \foreach \s in {1} {
      \node[bnode, deletedv] at (\s) {};
    }  
    % Added edges 
    \foreach \s / \t in {2/5,3/5} {
      \draw  (\s) edge[added] (\t);
    }
  \end{scope}
    \begin{scope}[xshift=\xs*2,yshift=-\ys]
    \foreach \n / \x / \y / \s in  {1/0/0/above, 2/1/1/above,3/1/-1/below,4/2/0/above,5/3/0/above} {
      \node[nnode] at (\x*\xscc,\y*\yscc) (\n) {};
      \node[\s = 0pt of \n] {$\n$};
    }
    \foreach \s / \t in {2/3,3/4,2/4,2/5,3/5} {
      \draw (\s) edge (\t);
    }
    \gettikzxy{(2)}{\twox}{\twoy};
    \gettikzxy{(4)}{\fourx}{\foury};
    \node at ($(\twox/2+\fourx/2, \twoy+20)$) (ml13) {Layer $3$};   
    %Sol
    \foreach \s in {1} {
      \node[bnode, deletedv] at (\s) {};
    }  
    % Added edges 
    \foreach \s / \t in {4/5} {
      \draw  (\s) edge[added] (\t);
    }
   \end{scope}

  % ML-2
  % Layer 1
    \begin{scope}[xshift=\xs*3+5ex,yshift=-\ys]
    \foreach \n / \x / \y / \s in  {1/0/0/above, 2/1/1/above,3/1/-1/below,4/2/0/above,5/3/0/above} {
      \node[nnode] at (\x*\xscc,\y*\yscc) (\n) {};
      \node[\s = 0pt of \n] {$\n$};
    }
    \foreach \s / \t in {1/2,1/3,1/4,2/3,3/4,2/4,4/5} {
      \draw (\s) edge (\t);
    }
    \gettikzxy{(2)}{\twox}{\twoy};
    \gettikzxy{(4)}{\fourx}{\foury};
    \node at ($(\twox/2+\fourx/2, \twoy+20)$) (ml21) {Layer $1$};
    
    % Sols
    % Marked vertices D
    \foreach \s in {1,5} {
      \node[bnode, deletedv] at (\s) {};
    }  
    
    % Deleted edges 
    \foreach \s / \t in {4/5} {
      \draw  (\s) edge[white] (\t);
      \draw  (\s) edge[deletede] (\t);     
    }
  \end{scope}

  \begin{scope}[xshift=\xs*4+5ex,yshift=-\ys]
    \foreach \n / \x / \y / \s in  {1/0/0/above, 2/1/1/above,3/1/-1/below,4/2/0/above,5/3/0/above} {
      \node[nnode] at (\x*\xscc,\y*\yscc) (\n) {};
      \node[\s = 0pt of \n] {$\n$};
    }
    \foreach \s / \t in {2/3,3/4,2/4,4/5} {
      \draw (\s) edge (\t);
    }
    \gettikzxy{(2)}{\twox}{\twoy};
    \gettikzxy{(4)}{\fourx}{\foury};
    \node at ($(\twox/2+\fourx/2, \twoy+20)$) {Layer $2$};
     
    % Sols
    % Marked vertices D
    \foreach \s in {1,5} {
      \node[bnode, deletedv] at (\s) {};
    }  
       
    % Deleted edges 
    \foreach \s / \t in {4/5} {    
      \draw  (\s) edge[white] (\t);
      \draw  (\s) edge[deletede] (\t);      
    }
  \end{scope}
    \begin{scope}[xshift=\xs*5+5ex,yshift=-\ys]
    \foreach \n / \x / \y / \s in  {1/0/0/above, 2/1/1/above,3/1/-1/below,4/2/0/above,5/3/0/above} {
      \node[nnode] at (\x*\xscc,\y*\yscc) (\n) {};
      \node[\s = 0pt of \n] {$\n$};
    }
    \foreach \s / \t in {2/3,3/4,2/4,2/5,3/5} {
      \draw (\s) edge (\t);
    }
    \gettikzxy{(2)}{\twox}{\twoy};
    \gettikzxy{(4)}{\fourx}{\foury};
    \node at ($(\twox/2+\fourx/2, \twoy+20)$) (ml23) {Layer $3$};   
    %Sol
    \foreach \s in {1,5} {
      \node[bnode, deletedv] at (\s) {};
    }  
    % Added edges 
    \foreach \s / \t in {4/5} {
      \draw  (\s) edge[added] (\t);
    }
  \end{scope}

   \draw ($(ml11.north west)+(0,.5)$) edge ($(ml23.north east)+(0,.5)$);
   \gettikzxy{(o3)}{\otx}{\oty};
   \gettikzxy{(t11)}{\tonex}{\toney};
   \gettikzxy{(ml13)}{\mlonex}{\mloney};
   \gettikzxy{(ml21)}{\mltwox}{\mltwoy};

   \draw ($(\otx/2+\tonex/2,\oty)+(0,0)$) edge ($(\mlonex/2+\mltwox/2,\mloney)+(0,-3)$);
  \end{tikzpicture}
  \caption{Examples for \TCE{} and \MLCE. Upper left: An instance with three layers. Upper right: A solution for \TCE\ with $k=1$ and $d=1$. Lower left: A solution for \MLCE{} with $k=3$ and $d=1$. Lower right: A solution for \MLCE{} with $k=1$ and $d=2$. 
We use red dashed edges to indicate edge deletion and green solid edges to indicate edge addition. Marked vertices are colored in green. Observe that there is no solution for \TCE{} when $k=0$ or $d=0$ and there is no solution for \MLCE{} when $k=0$ and $d \le 1$ or when $k\le 2$ and $d=0$.}
\label{fig:problem-example}
\end{figure}
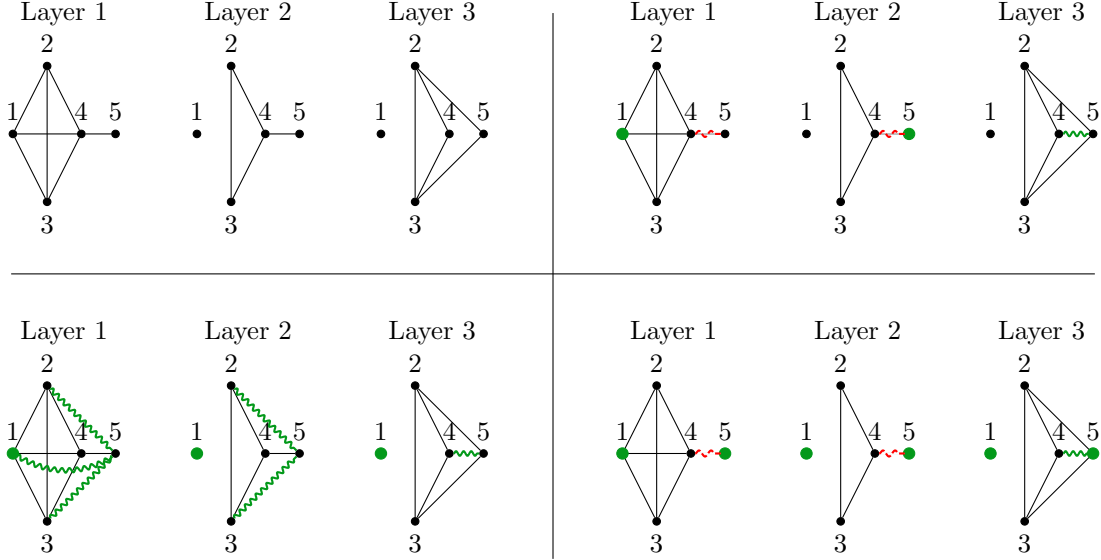

% \looseness=-1 
The most natural parameters are the ``number~$k$ of edge
modifications per layer'', the ``number~$d$ of marked vertices'', the
``number~$\ell$ of layers'', and the ``number~$n=|V|$ of vertices''.
An overview on our results is shown in \autoref{fig:resultsdiagram}.
(Note that, within these parameters, we have~$d\le n$ and
$k \le n^2$.) A straightforward reduction yields that \TCE\ is
\NP-complete even if both~$d = 0$
and~$\ell = 1$. %~(\appref{sec:hardness})\footnote{The proofs of results
%   marked by (\appsymb{}) and proofs of correctness and safeness of
%   reduction rules and branching rules marked by
%   (\appsymb{}) are deferred to an appendix.}. 
   On the positive side,
\TCE\ %remains polynomial-time solvable if $k=0$ and $\ell\le 2$~(\appref{sec:hardness}) and
allows for an algorithm with running time~$n^{O(k)}\ell$: The basic
idea is to check whether any two possible cluster editing sets for two
consecutive layers allow for a small number of marked vertices by
matching techniques. As it turns out, even for $d = 3$, we cannot
obtain an improved running time on the order of $(n\ell)^{o(k)}$
unless the Exponential Time Hypothesis (ETH) fails. The reason is an
obstruction represented by small clusters which may have to be joined or
split throughout many layers, to be able to form clusters in some
later layer. Finally, we give a polynomial kernel with respect to
the parameter combination $(d,k,\ell)$ and show that the problem does
not admit a polynomial kernel for parameter ``number $n$ of vertices''
unless \NoKernelAssume{}.

\newcommand{\tworows}[2]{\begin{tabular}{c}{#1}\\{#2}\end{tabular}}
\newcommand{\threerows}[3]{\begin{tabular}{c}{#1}\\{#2}\\{#3}\end{tabular}}

\tikzstyle{param}=[draw,rectangle, rounded corners=0.5mm, text height= 2ex, align=center] %, text width=12ex]
% \todo[inline]{update thm references}
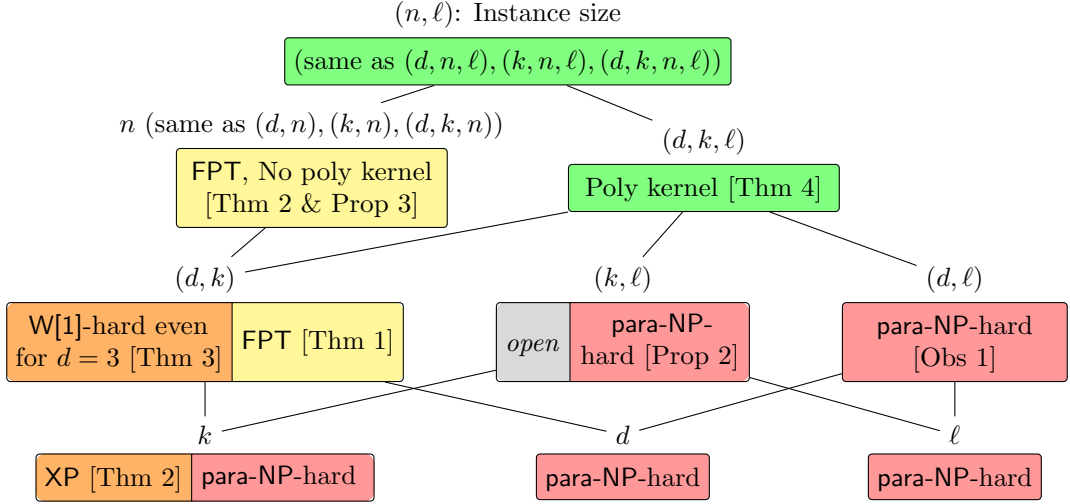
\begin{figure}[t]
\centering
    \begin{tikzpicture}[%yscale=0.7,xscale=1.1
      box/.style args = {#1/#2}{
        rectangle split, rectangle split horizontal, rectangle split parts=2,
        rounded corners=0.5mm,
        align=center,
        draw, 
        text height = 2ex,
        append after command={\pgfextra
            \fill[color=#1]
        (\tikzlastnode.south)
        [rounded corners] -| (\tikzlastnode.west) |- (\tikzlastnode.one north)
        [sharp corners]   -| (\tikzlastnode.one split) |- cycle;
            \fill[color=#2]
        (\tikzlastnode.two south)
        [rounded corners] -| (\tikzlastnode.east) |- (\tikzlastnode.north)
        [sharp corners]   -| (\tikzlastnode.one split) |- cycle;
                                        \endpgfextra}% end of the append after command
                            }% end of the box style definition
    ]

%    \node[param, fill=white!40!orange] (dk) at (-5, 2.3) {\tworows{$(d, k)$}{\W{1}-hard wrt.\ $k$ for constant $d$, \autoref{thm:tempw1hk}}};
    \node[] (dkname) at (0,0) {$(d,k)$};
  
    \node[box=white!40!orange/white!50!yellow, below = 0pt of dkname] (dk)  {
      \nodepart[text width = 18ex]{one} \W{1}-hard even for $d=3$~[Thm~\ref{thm:tempw1hk}]
      \nodepart{two}\FPT{}~[Thm~\ref{thm:fpt}]};

    \node[box=white!70!gray/white!60!red, right = 8ex of dk] (kl) {
      \nodepart[]{one}{\emph{open}}
      \nodepart[text width=14ex]{two} \paraNP{}-hard~[Prop~\ref{prop:hardness}]};

    \node[above = 0pt of kl] (klname) {$(k,\ell)$};
    
    \node[param, text width = 18ex, fill=white!60!red, right  = 8ex of kl] (dl) {\paraNP{}-hard \qquad\; [Obs~\ref{obs:hardness}]};
    
    \node[above = 0pt of dl] (dlname) {$(d,\ell)$};

    \gettikzxy{(dk.north east)}{\dkx}{\dky};
    \gettikzxy{(kl.north west)}{\klx}{\kly};
    \gettikzxy{(kl.north east)}{\klex}{\kley};
    \gettikzxy{(dl.north west)}{\dlx}{\dly};

    \node[param, text width = 22ex, fill=white!50!yellow,xshift=-12ex, yshift = 10ex] (n) at ($(\dkx/2+\klx/2,\dky)$) 
    % \node[param, fill=white!50!yellow] (n) at (-5, 4.2)
    {\FPT{}, No poly kernel [Thm~\ref{thm:tempxpk} \& Prop~\ref{thm:nopolykernel-n}]};

    \node[above = 0pt of n] (nname) {$n$ (same as $(d, n), (k, n), (d, k, n)$)};

    \node[param, text width = 22ex, fill=white!50!green,xshift=-8ex, yshift = 10ex] (dkl) at ($(\dlx/2+\klex/2,\dly)$)  {Poly kernel [Thm~\ref{thm:pk}]};

    \node[above = 0pt of dkl] (dklname) {$(d,k,\ell)$};

    \gettikzxy{(nname)}{\nx}{\ny};
    \gettikzxy{(dklname)}{\dklx}{\dkly};
    
    \node[param, fill=white!50!green,yshift=7ex] (nl) at ($(\nx/2+\dklx/2,\dkly)$) 
    {(same as $(d, n, \ell), (k, n, \ell), (d, k, n, \ell)$)};
 
    \node[above = 0pt of nl] (nlname) {$(n,\ell)$: Instance size};
%   \node[param, fill=white!50!green] (nl) at (-1, 6.1) {\tworows{$(n, \ell)$}{(same as $(d, n, \ell), (k, n, \ell), (d, k, n, \ell)$; ``instance size'')}};

    \node[below = 3ex of kl] (dname) {$d$};
    \node[param, fill=white!60!red, below = 0pt of dname] (d) {\paraNP{}-hard};
 
    \node[below = 3ex of dl] (lname) {$\ell$};
    \node[param, fill=white!60!red, below = 0pt of lname] (l) {\paraNP{}-hard};

%    \node[param, fill=white!40!orange] (k) at (-3.3, 0.4) {\tworows{$k$}{\XP{}, \autoref{thm:tempxpk}}};

    \node[below = 3ex of dk] (kname) {$k$};
    \node[box=white!40!orange/white!60!red, below=0pt of kname] (k) {
      \nodepart[]{one} \XP{} [Thm~\ref{thm:tempxpk}]
      \nodepart[]{two} \paraNP{}-hard %
    };
    
    \draw (dname) -- (dk);
    \draw (dname) -- (dl);
    \draw (kname) -- (dk);
    \draw (kname) -- (kl);
    \draw (lname) -- (dl);
    \draw (lname) -- (kl);
    
    \draw (dkname) -- (n);
    \draw (dkname) -- (dkl.south west);
    \draw (dlname) -- (dkl);
    \draw (klname) -- (dkl);
    
    \draw (nname) -- (nl);
    \draw (dklname) -- (nl);
    \end{tikzpicture}
    \caption{Our results for \TCE\ and \MLCE{} in a Hasse diagram of
      the upper-boundedness relation between the parameters the
      ``number~$k$ of edge modifications per layer'', the ``number~$d$
      of marked vertices'', the ``number~$\ell$ of layers'', and the
      ``number~$n=|V|$ of vertices'' and all of their combinations. A
      node is split into two parts if the complexity results differ;
      the left part shows the result for \TCE{}, the right part for
      \MLCE{}. Red entries mean that the corresponding parameterized
      problem % (the parameter is above the entry)
      is \paraNP{}-hard (NP-hard for constant parameter values).
      Orange entries mean that the corresponding
      parameterized problem is \W{1}-hard while contained in \XP{} (solvable in polynomial time for constant parameter values). It
      is in \FPT{} for all parameter combinations colored
      yellow or green and admits a polynomial kernel for all parameter
      combinations colored green. It does not admit a polynomial
      kernel for all parameter combinations that are colored yellow
      unless~\NoKernelAssume{}. A tight parameterized complexity
      classification for the gray colored parameter combination
      is
      open. % (it is clear that the problem is also in \XP{} for this combination).
    }
\label{fig:resultsdiagram}
\end{figure}

\subparagraph{\MLCELong\ (\MLCE).}
For clusterings of multi-layer graphs we typically have to consider the tradeoff between closely matching individual layers
and getting an overall sufficient fit~\cite{tagarelli_ensemble-based_2017,tang_clustering_2009} (see also the example in \autoref{fig:problem-example}).
A local upper bound on the number of allowed edits per layer and a global set of marked vertices allow us to study the influence of this tradeoff on
the complexity of multi-layer cluster editing. Formally, a clustering~$\M = (M_i)_{i \in [\ell]}$ for a multi-layer graph~$\{G_i \mid i \in [\ell]\}$ is defined in the same way as for temporal graphs. Clustering~$\M$ is \emph{totally $d$-consistent} if there is a single subset~$D$ of vertices such that~$G'_i[V\setminus D]=G'_j[V\setminus D]$ for all $i,j\in [\ell]$. Below we drop the qualifiers ``temporally'' and ``totally'' if they are clear from the context. The computational problem capturing the tradeoff between local and global fit mentioned above thus formalizes as follows.

\decprob{\MLCELong{} (\MLCE)}{A multi-layer graph~$\G$ and two integers~$k$ and~$d$.}{Is there a totally $d$-consistent $k$-bounded clustering for~$\G$?}
Again, we say that the vertices in the corresponding set~$D$ are \emph{marked} and that they together with sets $M_i$ of edge modifications constitute a \emph{solution}. Examples are shown in \cref{fig:problem-example}.

%  $G'_i\coloneqq (V,E_i\oplus M_i)$ is a cluster graph and such that % the consistency condition is fulfilled, 
% % namely
% for all $i,j\in [\ell]$ it holds that $G'_i[V\setminus D]=G'_j[V\setminus D]$.

% \decprob{Multi-Layer Cluster
%   Editing}{$\ell$~graphs $G_1=(V,E_1), \ldots, G_\ell = (V, E_\ell)$,
%   and two integers $k, d$.}{Is there a vertex subset $D\subseteq V$
%   with $|D|\le d$ and $\ell$~edge modification sets
%   $M_1, \ldots, M_\ell\subseteq \binom{V}{2}$ with $|M_i| \leq k$ such
%   that \begin{compactenum}
%   % \item for each $i \in [\ell]$ we have that $|M_i|\le k$,\label{en:budget}
%   \item for each $i \in [\ell]$ the graph $G_i'=(V, E_i\oplus M_i)$ is a
%   cluster graph, and\label{en:cluster}
%   \item for all $i, j \in [\ell]$ we have that $G_i'[V\setminus D] = G_j'[V\setminus D]$?\label{en:consistent}
%   \end{compactenum}}%

%   \looseness=-1 
  A brief summary of our results for \MLCE: While strong overall fit
  (small parameter~$d$) or closely matched layers (small
  parameter~$k$) alone do not lead to fixed-parameter tractability,
  jointly they do. Indeed, we obtain an
  $k^{O(k + d)}\cdot n^3 \cdot \ell$-time algorithm, in contrast to
  \TCE. At first glance, this is surprising because in the temporal
  case, we only need to satisfy the consistency condition ``locally''.
  This requires less interaction among layers and thus, %.
  % because, due to a smaller interaction between layers, Thus, it
  % seemed that the temporal problem should
  seemed to be easier to tackle than the multi-layer case. The
  algorithm uses a novel method that allows us make decisions over a
  large number of layers at once. It can be compared with greedy
  localization~\cite{hutchison_greedy_2004} in that some of the decisions are greedy and transient,
  meaning that they seem intuitively favorable and can be reversed in
  individual layers if they later turn out to be wrong. However, the
  application of this method is not straightforward, requires new techniques to deal with the interaction between layers and consequently
  intricately tuned branching and reduction rules.

%   \looseness=-1 
  We in fact completely classify \MLCE{} in terms of fixed-parameter
  tractability and existence of polynomial-size problem kernels with
  respect to the parameters~$k,d,\ell$, and
  $n$, % the ``number~$d$ of marked vertices'',
  % the ``number~$k$ of edge modifications per layer'', 
  % the ``number~$n$ of vertices'', 
  % and the ``number~$\ell$ of layers'',
  and all of their combinations, see \autoref{fig:resultsdiagram} for
  an overview. \MLCE\ is \paraNP{}-hard (NP-hard for constant parameter values) for all parameter combinations
  which are smaller or incomparable to $k + d$. Straightforward
  reductions yield \NP-completeness even if both $d = 0$ and
  $\ell = 1$ or both $k = 0$ and~$\ell = 3$; the problem is
  polynomial-time solvable if $k = 0$ and
  $\ell\le 2$. %~(\appref{sec:hardness}). 
  Finally, the kernelization
  results for \TCE{} also hold for \MLCE{}, that is, the problem
  admits a polynomial kernel with respect to % the parameter combination
  $(d,k,\ell)$ and does not admit a polynomial kernel for the % parameter
  ``number $n$ of vertices'' unless \NoKernelAssume{}.

  \subparagraph{Related Work.} % As mentioned, 
   Both multi-layer and temporal graphs harbor a range of
 important combinatorial problems, each with useful, nontrivial
 algorithmic theory. Such problems include
 multi-layer~\cite{bredereck2017assessing} and temporal (dense)
 subgraphs~\cite{himmel_adapting_2017,BHMMNS16}, temporal separators and
 paths~\cite{FMNZ18,ZFMN18,kempe2002connectivity}, covering problems~\cite{MMZ2019,AMSZ18}, and multi-layer
 connectivity~\cite{ALMS15,CY14,mertzios2013temporal}.

%   \looseness=-1 
  We are not aware of
  studies of the fundamental algorithmic properties of multilayer and
  temporal graph clustering. In terms of parameterized algorithms,
  only the indirect approach of aggregating clusterings into one has
  been
  studied for multilayer~\cite{betzler_average_2011,dornfelder_parameterized_2014} and temporal graphs~\cite{tantipathananandh_framework_2007}. These approaches are less accurate, however~\cite{barigozzi_identifying_2011, tantipathananandh_finding_2011% ,dong_clustering_2012,tang_clustering_2009
  }. The approximability of temporal versions of $k$-means clustering and its variants was studied by Dey et al.~\cite{dey_temporal_2017}.

\section{Basic Observations and Few Layers}
\label{sec:hardness}
% \appendixsection{sec:hardness}
%\subsection{Hardness Results}
We now give some basic observations on the complexity of \TCE\ and \MLCE\ on few layers. Note that the two problems coincide when the input multi-layer or temporal graph has only two layers. Moreover, we obtain a complexity dichotomy for \MLCE\ with $k = 0$ showing that for $\ell \leq 2$ the problem is polynomial-time solvable and for $\ell \geq 3$ it becomes \NP-hard. 

Both \TCE\ and \MLCE\ are contained in \NP{} since we can verify in polynomial time whether a given subset(s) of vertices and edge modification sets constitute a solution to the problem in question.
Thus, in all proofs for \NP-completeness, we omit the proof for \NP{} containment and only show the hardness part.

\textsc{Cluster Editing} is contained as a special case in both \TCE\ and \MLCE\ when $\ell = 1$, $d = 0$. Since \textsc{Cluster Editing} is \NP-complete~\cite{bansal2004correlation}, 
we immediately get \NP-hardness for \TCE{} and \MLCE{}.

\begin{obs}
\label{obs:hardness}
\TCE{} and \MLCE{} are both \NP-complete for $d=0$ and $\ell=1$. 
\end{obs}

We now consider the scenario where we are not allowed to edit any edges (\emph{i.e.}\ $k=0$).
We find that for two layers our problem is related to computing a maximum-weight matching in a bipartite graph, which is polynomial-time solvable.

\begin{prop}[See also Exercise 4.5 and its hint in Cygan et al.~\cite{CyganFKLMPPS15}]\label[prop]{prop:mlce-two-layers}
If $k=0$ and $\ell=2$, then \TCE{} and \MLCE\ can be solved in $O(n^2\log{n})$~time, where $n$ denotes the number of vertices.
\end{prop}
\begin{proof}
 Let $I=(G_1=(V, E_1), G_2=(V, E_2), k=0, d)$ be an input instance of \MLCE{}.
 We claim that the following procedure decides in $O(n^2\log{n})$~time whether $I$ is a yes-instance of \MLCE{}, \emph{i.e.}\
 whether there is a subset~$D\subseteq V$ of at most $d$ vertices such that $G_1[V\setminus D]=G_2[V\setminus D]$.

\begin{compactenum}
\item Check whether $G_1$ and $G_2$ are both cluster graphs. If at least one is not, answer NO.
\item Create a complete (edge-weighted) bipartite graph $H=(A\uplus B, E, w\colon E \to \{1,2,\ldots, n\})$ in the following way:
\begin{compactitem}
\item For each maximal clique $X$ in $G_1$ add a vertex $v_X$ to $A$.
\item For each maximal clique $X$ in $G_2$ add a vertex $v_X$ to $B$.
\item Add an edge between each two vertices~$v_X\in A$ and $v_Y\in B$ with edge weight $w(\{v_X,v_Y\})=|X\cap Y|$.
\end{compactitem}
\item Compute a maximum-weight matching for $H$. If the weight of the matching is at least~$|V|-d$, answer~YES, otherwise answer~NO.
\end{compactenum}
It is well-known that the first step reduces to checking whether there
is an induced $P_3$ in one of the graphs, which can be done in
$O(n + m)$ time, where $n$ is the number of vertices and $m$ is the maximum number of edges in a cluster graph.\footnote{The proof is folklore and proceeds roughly as follows. Find the connected components of the input graph. Next, determine whether there are two nonadjacent vertices~$u, v$ in a connected component. If so, then find an induced $P_3$ along a shortest path between $u$ and $v$. Otherwise, there is no induced~$P_3$. Nonadjacent vertices in a connected component can be checked for in $O(\deg(v))$ time summed over each vertex~$v$ in that component.} The second step can be performed in $O(n + m)$ time
as follows. 
Find all connected components in $G_1$ and label the vertices in $G_1$ according to the components that contain them.
Introduce to $A$ a \emph{cluster vertex}~$v_X$ for each label~$X$.
The vertices in $B$ are constructed analogously.
Now, to compute the edge weights in $H$, 
iterate over all vertices in~$V$ and 
add to $H$ an edge of weight one that is incident with the two corresponding cluster vertices or increase the edge
weight if the edge is added due to previous iteration. Note that $H$ contains at most~$n$ edges.
Finally, the third step can be carried out in $O(n^2\log{n})$ time using the
Hungarian algorithm,\todo{OS:should we have a reference here?}{} which also dominates the remaining running time.

\proofparagraph{Correctness.} Note that if one of $G_1$ and $G_2$ is not a cluster graph, then we clearly face a no-instance, which is correctly identified by the algorithm in the first step. So from now on, assume that both $G_1$ and $G_2$ are cluster graphs. 
To show the correctness of the last step, suppose that there is a vertex subset~$D\subseteq V$ of size at most $d$ such that $G_1[V\setminus D] = G_2[V\setminus D]$. 
Let $q_1,q_2,\ldots,q_x$ bet the maximal cliques remaining in $G_1[V\setminus D]$.
One can verify the following matching $M$
has weight $|V|-|D|$: For each clique~$q_i$, add to $M$ the edge $\{v_X,v_Y\}$ where $X$ and $Y$ are the two maximal cliques that contain $q_i$ in $G_1$ and $G_2$, respectively. 
Note that since $G_1$ and $G_2$ are two cluster graphs on the same vertex set,
no maximal clique remaining in $G_1[V\setminus D]$ belongs to two different maximal cliques in $G_1$ or~$G_2$.
Thus, $M$ is indeed a matching. 
It is straightforward to see that it has weight $|V| - |D|$.
%Hence, the weighted matching for $H$ that matches for each cluster in $G_1[V\setminus D]$ the cluster that contains it in $G_1$ with the cluster that contains it in $G_2$ has weight at least $|V|-|D|$. 

In the opposite direction, assume that $H$ admits a matching $M$ with weight at least $|V|-d$.
We consider the following subset~$V'$ of vertices:
For each edge~$\{v_X,v_Y\}$ in $M$,
% if there is a weighted matching $M$ in $H$ of weight at least $|V|-d$ then for each edge $\{v_X,v_Y\}\in M$ we 
add to $V'$ all vertices in $X\cap Y$; their number is exactly the weight of $\{v_x,v_y\}$ in $H$.
Since $V'$ only contains vertices which are in the intersection of two maximal cliques in $G_1$ and $G_2$, respectively, it follows that $G_1[V'] = G_2[V']$.
Thus, if we remove, by marking, all vertices in $V\setminus V'$,
then both cluster graphs become the same.
Since $M$ is a matching, it follows $|V'|=w(M)\ge |V|-d$.
Thus, at most $d$ vertices, namely those in $V\setminus V'$, are marked.
%The matching~$M$ has weight at least $|V|-d$ 
%to check that if the marked vertices are removed, both cluster graphs are the same, and at least $|V|-d$ vertices remain unmarked, hence at most $d$ vertices are marked.
\end{proof}

As soon as there are three layers, even when we are only allowed to mark vertices, \MLCE{} is \NP-hard.
We establish this by providing a polynomial-time reduction from an \NP-complete\ 3-SAT variant,
called (2,2)-3-SAT. Herein, each clause has two or three literals,
and each literal appears exactly twice~\cite[Lemma~1]{BulCheFalNieTal2015}.

% By a polynomial-time reduction from \textsc{Vertex Cover} 
% (given a simple graph~$G$ and an integer~$s$, 
% decide whether there is a size-at-most-$s$ \emph{vertex cover}, 
% i.e., a subset of at most $s$ vertices
% which are jointly incident to all edges)
% which is \NP-complete on graphs with maximum vertex degree three~\cite{GJ79},
% we obtain that \MLCE{} is \NP-hard for a constant number of layers even if no edge modifications are allowed. 
% This means that also the marking of vertices is a computationally hard task.
% \todo[inline][inline]{Hua: Do we need to introduce the problem Vertex Cover?}
% \begin{compactitem}
% \item NP-hard for $d=0$ and $t=1$ (reduction from the ordinary Cluster Editing).
% % \item NP-hard for $k=0$ and $t=4$ (reduction from Vertex Cover on max.-degree 3 graphs).
% % \item Polykernel in $(k, d, t)$.
% % \item Time $n^d\cdot f(k) \cdot\text{poly}(n+t)$ algorithm (see above).
% % \item No PK for $(k, d)$. 
% \end{compactitem}

\begin{prop}
\label{prop:hardness}
 \MLCE{} is \NP-complete even if $k=0$ and $\ell=3$.
\end{prop}
\begin{proof}
  To show the hardness, we reduce from the \NP-complete (2,2)-3-SAT problem~\cite[Lemma~1]{BulCheFalNieTal2015}.
  Let $I=(\mathcal{X}, \mathcal{C})$ be an instance of (2,2)-3-SAT, 
  where $\mathcal{X}=\{x_1,x_2,\ldots, x_n\}$ is a set of $n$ variables and $\mathcal{C}=\{C_1,C_2,\ldots,C_m\}$ 
  is a set of $m$ clauses of size two or three
  such that each variable appears exactly four times, twice as a positive literal and twice as a negative literal.

  We aim to construct an instance~$I'=(G_1=(V,E_1),G_2=(V,E_2),G_3=(V,E_3), k=0, d)$ for \MLCE{}.
  %Since we are not allowed to edit any edge ($k=0$)
  %the three input graphs that we are going to construct will be cluster graphs as otherwise we will always have a no-instance.
  The idea behind the reduction is to use two layers ($G_1$ and $G_2$) to construct a variable gadget for each variable, one layer for each truth value.
  Then, we use a third layer ($G_3$) to construct a satisfaction gadget for each clause that has ``connection'' to some vertices that correspond to the literals contained in the clause.
  This connection and the number~$d$ (defined below) ensure that we need to mark at least one vertex that corresponds to a literal in the clause.
  This means that at least one literal needs to be set to true in order to satisfy the clause.
  An example for the corresponding construction is shown in \cref{fig:hardness-instance}.
  
  Formally, vertex set~$V$ for $I'$ consists of two groups:
  \begin{compactitem}
    \item For each variable~$x_i\in \mathcal{X}$, create two pairs of \emph{variable vertices}, denoted as $x^1_i, y^1_i$, $x^2_i$, and~$y^2_i$.
    \item For each clause~$C_j \in \mathcal{C}$, create $|C_j|$ \emph{clause vertices}, denoted as $c^z_j$, $1\le z\le |C_j|$.
  \end{compactitem}
  Let $D^{\textsf{true}}_i=\{x^1_i, x^2_i\}$ and $D^{\textsf{false}}_i=\{y^1_i, y^2_i\}$.
  We will construct two layers so that for each variable~$x_i$, we need to mark either all vertices in $D^{\textsf{true}}_i$ or all vertices in $D^{\textsf{false}}_i$.
  Intuitively, marking~$D^{\textsf{true}}_i$ corresponds to setting the variable~$x_i$ to true while marking~$D^{\textsf{false}}_i$ corresponds to setting the variable~$x_i$ to false.
  
  \noindent The three layers are constructed as follows:
  \begin{description}
    \item[Layer 1.] For each $i\in \{1,2,\ldots, n\}$, add to $G_1$ the two disjoint edges~$\{x^1_i,y^1_i\}$ and $\{x^2_i, y^2_i\}$.
    For each $j\in \{1,2,\ldots, m\}$, add to $G_1$ a clique consisting of all the corresponding clause vertices~$c^z_j$, $1\le z\le |C_j|$.
    Note that $|C_j|$ has either two or three literals, and if $C_j$ has two literals, then the constructed clique is an edge; otherwise, it is a triangle.
    \item[Layer 2.] For each $i\in \{1,2,\ldots, n\}$,
    add to $G_2$ the two disjoint edges~$\{x^1_i, y^{2}_i\}$ and $\{x^2_i, y^1_i\}$.
    \item[Layer 3.] For each $j\in \{1,2,\dots, m\}$, let $\ell^{t}_j$ be the $t^{\text{th}}$ literal in $C_j$, $1\le t\le |C_j|$.
    If $\ell^t_j$ corresponds to a positive literal~$x_i$ for some $i\in \{1,2,\ldots,n\}$
    and it is the $z^{\text{th}}$ occurrence of literal~$x_i$ (note that $z\in \{1,2\}$),
    then add to $G_3$ the edge~$\{c^t_j, x^z_i\}$.
    If $\ell^t_j$ corresponds to a negative literal~$x_i$ for some $i\in \{1,2,\ldots,n\}$
    and it is the $z^{\text{th}}$ occurrence of literal~$\overline{x}_i$ (note that $z\in \{1,2\}$),
    then add to $G_3$ the edge~$\{c^t_j, y^z_i\}$.
  \end{description}
  Observe that the symmetric difference between the edge sets of the first two layers,
  restricted to the variable vertices,
  forms a vertex-disjoint union of cycles of length four each.
  This means that we need to mark at least two vertices in each cycle.

To complete the construction, we set $k=0$ and $d=2n+\sum_{j=1}^{m}(|C_j|-1)$.
Clearly, the construction can be conducted in polynomial time.
We move on to show the correctness, \emph{i.e.}, instance~$I$ admits a satisfying truth assignment if and only if there is a vertex subset~$D\subseteq V$ with $d$ vertices such that $G_1[V\setminus D]=G_2[V\setminus D]=G_3[V\setminus D]$.

\smallskip
\noindent $\boldsymbol{(\Rightarrow)}$ For the ``only if'' direction, assume that $\sigma\colon \mathcal{X}\to \{T,F\}$ is a satisfying truth assignment for $I$.
To prove that $I'$ is a yes-instance of \MLCE{} it suffices to show that after marking and removing the following vertices
we obtain only isolated vertices for all three layers.
\begin{compactenum}
  \item For each variable~$x_i\in \mathcal{X}$, if $\sigma(x_i)=T$, then mark the vertices in $D^{\textsf{true}}_i$; otherwise mark the vertices in $D^{\textsf{false}}_i$.
  \item For each clause~$C_j \in \mathcal{C}$, identify a literal~$\ell^{t}_j$ in $C_j$ such that $\sigma(\ell^{t}_j)$ makes $C_j$ satisfied.
  Let $u\in \{x^1_i,y^1_i,x^2_i, y^2_i\}$ (for some $i\in \{1,2,\ldots, n\}$) be the vertex that corresponds to the literal~$\ell^{t}_j$.
  Mark all $|C_j|-1$ clause vertices~$c^z_j$ ($1\le z\le |C_j|$) that are \emph{not} adjacent to $u$ in~$G_3$.
\end{compactenum}
Obviously, we have marked (and removed) in total $d$~vertices.
Moreover, in the first two layers, we have marked all but one vertex of each maximal clique. 
Thus, after removing all vertices marked according to the rules above, the first two layers each contain only isolated vertices.
To see why in the third layer only isolated vertices remain,
 we observe that $G_3$ consists of only disjoint edges and for each clause~$C_j$ we have marked exactly $|C_j|-1$ endpoints, each of a different edge.
Thus, exactly one clause vertex remains that is adjacent to a variable vertex~$u$ such that the corresponding literal, under the truth assignment of $\sigma$, makes $C_j$ satisfied.
This means that $u$ is removed. 
Thus, this remaining clause vertex is also isolated. 

\smallskip
\noindent $\boldsymbol{(\Leftarrow)}$ For the ``if'' direction, assume that there is a subset~$D$ of at most $d$~vertices, deleting which makes all three layers have the same (cluster) graph.
For each variable~$x_i \in \mathcal{X}$, the two pairs of variable vertices form a length-four cycle in the symmetric difference between the first two layers.
Thus, we need to mark at least two variable vertices for each variable,
implying that $D^{\textsf{true}}_i\subseteq D$ or $D^{\textsf{false}}_i\subseteq D$.
We claim that indeed $D \cap \{x^1_i,y^1_i,x^2_i, y^2_i\}$ equals either $D^{\textsf{true}}_i$ or $D^{\textsf{false}}_i$
by showing that for each variable at most two variable vertices belong to $D$.
To see this, observe that all $|C_j|$ clause vertices that correspond the same clause~$C_j \in \mathcal{C}$
form a maximal clique in the first layer and an independent set in the second layer.
Thus, we need to mark at least $|C_j|-1$ clause vertices.
In total, we have to mark at least $\sum_{j=1}^{m}(|C_j|-1)$ clause vertices.
Thus, by the definition of $d$, we can mark at most $2n$ variable vertices.
As already reasoned, for each variable, we have to mark at least two variable vertices.
All together, $D$ contains at most two variable vertices for each variable.

Now, we can assume that for each variable~$x_i\in \mathcal{X}$ the intersection $D \cap \{x^1_i,y^1_i,x^2_i, y^2_i\}$ equals either $D^{\textsf{true}}_i$ or $D^{\textsf{false}}_i$.
We claim that setting each variable~$x_i$ to true if $D^{\textsf{true}}_i\subseteq D$ and to false otherwise satisfies all clauses.
Suppose, towards a contradiction, 
that there is clause~$C_j \in \mathcal{C}$ that is not satisfied.
By our assignment, this means that for each literal~$\ell_j$ in $C_j$ 
we have that if it is a positive literal~$x_i$ (for some $i$), then $D^{\textsf{true}}_i \cap D=\emptyset$; otherwise $D^{\textsf{false}}_i\cap D=\emptyset$.
This means that no variable vertex that is adjacent to some clause vertex~$c_j^z$ (for some $z$) is marked.
However, by our reasoning above, exactly one clause vertex~$c_j^z$ is not marked.
This means that there is at least one edge remaining in the third layer after we remove all vertices in $D$ which does not exist in the first two layers---a contradiction.
\end{proof}

% \noindent An example for our hardness proof for \cref{prop:hardness} is depicted 
% % consider the following (2,2)-3-SAT-instance with three variables and four clauses: $\mathcal{X}= \{x_1,x_2.x_3\}$ and $\mathcal{C}= (C_1, C_2, C_3, C_4)$:
% % \begin{align*}
% %   C_1= (x_1\vee \overline{x}_2 \vee x_3),   C_2= (\overline{x}_1\vee \overline{x}_2 \vee \overline{x}_3), C_3 = ({x}_1 \vee x_2 \vee \overline{x}_3), C_4 = (\overline{x}_1 \vee x_2 \vee x_3).
% % \end{align*}
% % The instance that we construct for our problem could be as shown
% in \cref{fig:hardness-instance}.
 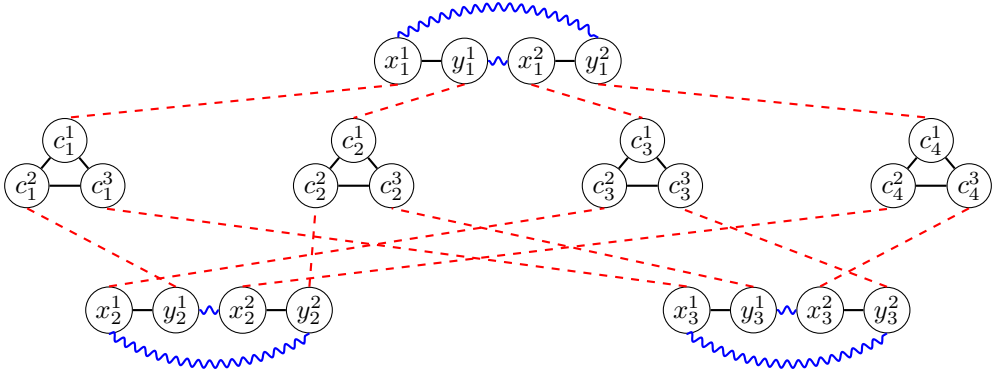
\begin{figure}
   \centering
   \begin{tikzpicture} 
     % The clause vertices
     \tikzstyle{nnode} = [draw, circle, inner sep=1pt]
     \def \xs {3.8}
     \def \xss {1}
     \tikzstyle{layer1}= [thick]
     \tikzstyle{layer2}= [decorate, decoration={snake,amplitude=.5mm,segment length=1.5mm}, blue, thick]
     \tikzstyle{layer3}= [dashed, red, thick]

     \foreach \i in {1,2,3,4} {%
       \foreach \j  in {2,3} {%
         \node[nnode] at (\i*\xs+\j*\xss, 0) (c\i\j) {$c^\j_\i$};
       }
       \node[nnode] at (\i*\xs+\xss*2.5, .6) (c\i1) {$c^1_\i$};
     }

     % The variable vertices
     \foreach \s / \d / \t / \z in {1/2/3/30, 2/1/2/-64, 3/3/4/-64} {%
       \gettikzxy{(c\d1)}{\onex}{\oney};
       \gettikzxy{(c\t1)}{\twox}{\twoy};
       \foreach \n / \j / \x / \y  in {x/1/-1.5/1, x/2/.5/1, y/1/-.5/1, y/2/1.5/1} {
           \node[nnode] (\s\n\j) at (\onex*.5+\twox*.5+\x*25, \oney+\z) {$\n^\j_\s$};
       }
     }
     
     % Layer 1
     \foreach \i in {1,2,3} {
       \foreach \j in {1,2} {
       \draw ({\i x\j}) edge[layer1] ({\i y\j});
     }
     }

     \foreach \i in {1,2,3,4} {
       \foreach \s / \t in {1/2, 2/3, 3/1} {
         \draw (c\i\s) edge[layer1] (c\i\t);
       }
     }

     % Layer 2
     \foreach \i / \t / \s in {1/bend left/north,2/bend right/south, 3/bend right/south} {
       \draw ({\i x1.\s}) edge[layer2,\t] ({\i y2.\s});
       \draw ({\i x2}) edge[layer2] ({\i y1});
     }

     % Layer 3
     \foreach \s / \t in {1x1.south/c11.north, 1x2.south/c31.north, 1y1.south/c21.north, 1y2.south/c41.north,%
       2x1.north/c32.south,2x2.north/c42.south,2y1.north/c12.south,2y2.north/c22.south, %
       3x1.north/c13.south,3x2.north/c43.south,3y1.north/c23.south,3y2.north/c33.south%
     } {
       \draw (\s) edge[layer3] (\t);
     }

   \end{tikzpicture}
   \caption{An instance produced for a (2,2)-3-SAT instance~$I$ in the proof for \cref{prop:hardness}, where $I=(\{x_1,x_2,x_3\}, C_1=(x_1\vee \overline{x}_2 \vee x_3),   C_2= (\overline{x}_1\vee \overline{x}_2 \vee \overline{x}_3), C_3 = ({x}_1 \vee x_2 \vee \overline{x}_3), C_4 = (\overline{x}_1 \vee x_2 \vee x_3))$. Each layer has a different line style: Layer 1 has solid black lines, layer 2 has blue zigzag lines, while layer 3 has red dashed lines.}
   \label{fig:hardness-instance}
 \end{figure}

\section{\MLCELong{} (\MLCE)}
\label{sec:fpt}
In this section, we present an \FPT\ algorithm for \MLCE{} with respect to the combined parameter $(k, d)$.

\begin{theorem}
\label{thm:fpt}
 \MLCE{} is \FPT{} with respect to the number~$k$ of edge modifications per layer and number~$d$ of marked vertices combined. It can be solved in $k^{O(k+d)}\cdot n^3\cdot \ell$ time.
\end{theorem}
\noindent We describe a recursive search-tree algorithm (see~\autoref{alg:mlce}) for the following input:
\begin{compactitem}
\item An instance $I$ of \MLCE{} consisting of a multi-layer graph $G_1, \ldots, G_\ell = (V, E_1), \ldots, (V, E_\ell)$ and two integers $k$ and $d$.
\item A \emph{constraint} $P = (D, (M_i)_{i \in [\ell]}, B)$, consisting of a set of \emph{marked} vertices $D \subseteq V$, edge modification sets $M_1, \ldots, M_\ell \subseteq \binom{V}{2}$, and a set $B\subseteq \binom{V\setminus D}{2}$ of \emph{permanent} vertex pairs.
\end{compactitem}
Moreover, we require that the constraint given to the recursive algorithm to be aligning. A constraint $P= (D, (M_i)_{i \in [\ell]}, B)$ is \emph{aligning} if $G'_i[V\setminus D]=G'_j[V\setminus D]$ for all $i,j\in[\ell]$, where $G'_i=(V, E_i\oplus M_i)$ for all $i\in[\ell]$.

Note that the algorithm expects some initial modification sets as input that, when applied to all layers, makes them equal up to marked vertices. These initial edge modification sets are computed greedily, hence the algorithm follows the greedy localization approach~\cite{hutchison_greedy_2004} in which we
make decisions greedily and possibly revert them through 
branching later on. The greedy decisions herein give us some structure
that we can exploit to keep the search-tree size bounded in~$k$ and~$d$. The edge
modification sets~$M_i$ represent both the greedy decisions and those
that we made through branching. The set~$B$ contains only those made
by branching. 

The initial modification sets are constructed according to the following rule that adds all edges that appear in at least half of all layers to all of the remaining layers and removes all other edges.
\begin{grule}
   \label{brule:greedy}
   Let $M_i = \emptyset$ for every $i \in [l]$.\\
  For every vertex pair $\{u, v\}\in \binom{V}{2}$ do the following:
  \begin{compactitem}
  \item 
  If $|\{ E_i \mid \{u, v\}\in E_i\}|\ge \frac{\ell}{2}$, then for all $ i\in[ \ell]$ set $M_i \leftarrow
  M_i \cup (\{\{u, v\}\}\setminus E_i)$. 
  \item If $|\{ E_i \mid \{u, v\}\in E_i\}|< \frac{\ell}{2}$, then for all $i \in[ \ell]$ set $M_i \leftarrow
  M_i \cup (\{\{u, v\}\}\cap E_i)$. 
  \end{compactitem}
\end{grule}

From now on, we assume that the input constraint of the algorithm contains edge modification sets produced by the \hyperref[brule:greedy]{Greedy Rule}, together with an empty set of marked vertices and an empty set of permanent vertex pairs. We call this constraint $P_\text{greedy}$. Note that $P_\text{greedy}$ is an aligning constraint. 

Throughout the algorithm, we try to maintain a \emph{good} aligning constraint which
intuitively means that the constraint can be turned into a solution
(if one exists). 
  \begin{defi}[Good Constraint]\label{defi:partialsolution}
  Let $I$ be an instance of \MLCE{}. A constraint $P=(D, M_1, \ldots, M_\ell, B)$ is \emph{good} for $I$ 
  if $I$ is a yes-instance and there is a solution $S=(M^\star_1, \ldots, M^\star_\ell, D^\star)$ for $I$
  such that
  \begin{compactenum}[(i)]
  \item $D \subseteq D^\star$,\label{cons:marked}
  \item there is no $\{u, v\}\in B$ such that $u\in D^\star$, and\label{cons:perm-not-marked}
  \item for all $i \in [\ell]$ we have $M_i \cap B = M^\star_i \cap B$.\label{cons:resp-perm}
  \end{compactenum}
  We also say that $S$ \emph{witnesses} that $P$ is good.
\end{defi}
\noindent If a constraint is not good, we call it \emph{bad}. It is easy to see that if we face a yes-instance, then any constraint containing an empty set of marked vertices and an empty set of permanent vertex pairs is good. We call such constraints \emph{trivial}.
  
\begin{obs}\label{obs:initial}
  For any yes-instance $I=(G_1, \ldots, G_\ell, d, k)$ of \MLCE{}, we have that the constraint $P = (D=\emptyset, M_1, \ldots, M_\ell, B=\emptyset)$ is a good constraint for $I$ for any sets $M_i\subseteq \binom{V}{2}$ with $i\in[\ell]$.
\end{obs}
\noindent It is obvious that $P_\text{greedy}$ is trivial. Hence if the input instance of \MLCE{} for our algorithm is a yes-instance, then the initial call is with a good constraint. The algorithm is supposed to return~\texttt{true} if the supplied constraint is good and \texttt{false} otherwise. 

Our algorithm uses various different branching rules to search for a solution to a \MLCE{} input instance. Formally, branching rules are defined as follows.
\begin{defi}[Branching Rule]
A \emph{branching rule} takes as input an instance~$I$ of \MLCE{} and an aligning constraint~$P$ and returns a set of aligning constraints $P^{(1)}, \ldots, P^{(x)}$. 
\end{defi}
When a branching rule is applied, the algorithm invokes a recursive call for each constraint returned by the branching rule and returns~\texttt{true} if at least one of the recursive calls returns~\texttt{true}; otherwise, it returns \texttt{false}. For that to be correct, whenever a branching rule is invoked with a good constraint, at least one of the constraints returned by the branching rule has to be a good constraint as well. Furthermore, if a branching rule is invoked with a bad constraint, none of the constraints returned by the branching rule should be good.
In this case we say that a branching rule is~\emph{safe}. 
\begin{defi}[Safeness of a Branching Rule]
   We say that a branching rule is \emph{safe} if the following holds: \begin{compactitem}
   \item If the branching rule is applied on an instance of \MLCE{} together with a good constraint for that instance, then at least one of the constraints returned by the branching rule is good.
   \item If the branching rule is applied on an instance of \MLCE{} together with a bad constraint for that instance, then none of the constraints returned by the branching rule is good.
   \end{compactitem}
\end{defi}

In the following, we introduce the branching rules used by the algorithm and prove that each of them is safe (in some cases under the condition that certain other rules are not applicable). This together with \autoref{obs:initial} will allow us to prove by induction that the algorithm eventually finds a solution for the input instance of \MLCE{} if it is a yes-instance. 

The following notion and observation will be useful for the safeness proofs.
\begin{defi}
 Let $I$ be an instance of \MLCE{} and $P=(D, M_1, \ldots, M_\ell, B)$ and $P'=(D', M'_1, \ldots, M'_\ell, B')$ two constraint. We say that $P'$ \emph{extends} $P$ if $D' \supseteq D$, $B' \supseteq B$, and  for each $i \in [\ell]$ we have $M'_i \cap B = M_i \cap B$.
\end{defi}

\begin{obs}\label{obs:extends}
 If $I$ be an instance of \MLCE{}, $P$ and $P'$ are two constraints such that $P'$ extends $P$ and $S$ is a solution witnessing that $P'$ is good, then $S$ also witnesses that $P$ is good.
\end{obs}

We start with a rule that checks obvious constraints and aborts the recursion if they are not fulfilled.
\begin{rulezero}\label{rule:zero}
If $|D|>d$ or there is an $i \in [\ell]$ such that $|M_i\cap B|>k$, then abort the current branch and return \texttt{false}.
\end{rulezero}
The correctness of this rule is obvious. With the next rule we edit the subgraphs induced by all non-marked vertices into cluster graphs. Similar to classical \CE{}, we branch on all edits that destroy induced $P_3$s. 
Additionally, we have to take into account that it may be necessary to mark vertices because otherwise they may force us to edit too many edges in some layer.
  \begin{brule}
  \label{brule:clustering}
  If there is an induced $P_3 = (\{u, v\},\{v, w\})$ in $G'_i[V\setminus D]$ for some $ i\in[ \ell]$, where $G'_i = (V, E_i\oplus M_i)$, then return the following up to six constraints:
  \begin{compactenum}
  \item If $\{u, v\}\notin B$: for all $ i\in[ \ell]$ put $M^{(1)}_i = M_i
  \oplus \{\{u, v\}\}$, $D^{(1)}=D$, and $B^{(1)}= B\cup \{\{u, v\}\}$. 
  \item If $\{v, w\}\notin B$: for all $ i\in[ \ell]$ put $M^{(2)}_i = M_i
  \oplus \{\{v, w\}\}$, $D^{(2)}=D$, and $B^{(2)}= B\cup \{\{v, w\}\}$. 
  \item If $\{u, w\}\notin B$: for all $i\in[ \ell]$ put $M^{(3)}_i = M_i
  \oplus \{\{u, w\}\}$, $D^{(3)}=D$, and $B^{(3)}= B\cup \{\{u, w\}\}$. 
  \item For each $x \in \{u, v, w\}$: If there is no $y\in V\setminus D$ such
  that $\{x, y\}\in B$, then return a constraint with $D^{(\cdot)}= D\cup\{x\}$, the rest stays the same. 
  \end{compactenum}
  If none of the above possibilities apply, then return \texttt{false}.\footnote{This technically does not fit the definition of a branching rule but we can achieve the same effect by returning trivially unsatisfiable constraints such as a constraint with $|D^{(\cdot)}|>d$ which is rejected by \hyperref[rule:zero]{Rule~0}.}
  \end{brule}
\begin{lem}\label{lem:clustering}
 \autoref{brule:clustering} is a safe branching rule.
\end{lem}
\begin{proof}
It is easy to check that \autoref{brule:clustering} is indeed a branching rule since it always modifies the pairs in the edge modifications sets of all layers, hence if the input constraint is aligning so are all output constraints. Since each output constraint extends the input constraint, by \autoref{obs:extends}, if any of the output constraints is good, then so is the input constraint. 

Now we show that if the input constraint is good, at least one output constraint is. Let the input constraint $P=(D, M_1, \ldots, M_\ell, B)$ be good and let $S=(M^\star_1, \ldots, M^\star_\ell, D^\star)$ be a solution for the input instance witnessing that $P$ is good. 
Since each $M_i^\star$ is a cluster editing set for $G_i$, it holds that, for all $i\in[\ell]$, graph~$G^\star_i[V\setminus D^\star]$ does not contain a $P_3$ as an induced subgraph, where $G^\star_i=(V, E_i\oplus M_i^\star)$. Hence, if there is some $i \in [\ell]$ and three vertices~$u, v, w$ that induce a $P_3$ in $G_i'[V\setminus D]$, where $G_i'=(V, E_i\oplus M_i)$, then there are two cases.

In the first case, one of $u, v, w$ is also in $D^\star$, say~$v \in D^\star$.
Note that, then, $v$ cannot be part of any permanent vertex pair, by the
definition of good constraints. Thus, the constraint that puts $v \in D$ output
in the fourth part of \autoref{brule:clustering} is good.

The second case is that $u, v, w \in V \setminus D^\star$. Then, since
$G^\star_i[V\setminus D^\star]$ is a cluster graph, at least one of the vertex
pairs formable from $u, v, w$ is modified by $S$, that is, in $M_i^\star$. Say
$\{u, v\} \in M_i^\star$. Since the solution is consistent, $\{u, v\}$ either appears in $G^\star_i$ for all $i \in [l]$ or in none of them.
settled. Note that $\{u, v\}$ cannot be permanent since otherwise we
already have that $\{u, v\} \in M_i$ by the definition of a good
constraint.
Thus the constraint which adds $\{u, v\}$ to $M_i$ and makes it permanent is good. Hence, the rule is safe.
\end{proof}
  
The next rule keeps the sets of edge modifications~$M_i$ free of marked vertices. Pairs in $M_i$ can become marked if vertices of vertex pairs processed by the \hyperref[brule:greedy]{Greedy Rule} are marked by other branching rules further down the search tree. 
We invoke this rule on the beginning of each recursive call to modify the constraint before applicability of other rules is tested.
  \begin{crule}
  \label{crule:cleanup}
  For each $i \in [ \ell]$ and each $\{u, v\} \in M_i$:
  If $\{u, v\} \cap D \neq \emptyset$, then remove $\{u, v\}$ from $M_i$.
  \end{crule}

To show the safeness of this rule, we can formally treat the Clean-up Rule as a special case of a branching rule, i.e., it produces one constraint.

\begin{lem}\label{lem:cleanup}
 The \hyperref[crule:cleanup]{Clean-up Rule} (in that sense) is a safe branching rule.
\end{lem}
\begin{proof}
It is easy to check that \hyperref[crule:cleanup]{Clean-up Rule} is indeed a branching rule since it only removes vertex pairs that contain marked vertices from the edge modification sets, hence if the input constraint is aligning so are all output constraints. 
Note that permanent vertex pairs cannot contain marked vertices by the definition of constraints. It follows that the \hyperref[crule:cleanup]{Clean-up Rule} does not add or remove permanent vertex pairs from any set~$M_i$. 
Furthermore, it does not change the sets $D$ and $B$. It follows that the input constraint cannot become bad if it was good or vice versa. Hence, the Clean-up Rule is safe.
\end{proof}

  The next rule tries to repair any budget violations that might occur. Since
  with the \hyperref[brule:greedy]{Greedy Rule} we greedily make decisions in the beginning 
  we expect that some
  of the choices were not correct. This rule will then revert these choices.
  Also, to have a correct estimate of the sizes of the current edge modification
  sets, this rule requires that the \hyperref[crule:cleanup]{Clean-up Rule} was applied. 
  \begin{brule}
  \label{brule:budget}
  If there is an $M_i$ for some $i \in [ \ell]$ with $|M_i|>k$, then take any set $M'_i\subseteq M_i \setminus B$ such that $|M'_i| + |B \cap M_i| = k+1$ and return the following constraints:
  \begin{compactenum}
  \item For each $\{u, v\}\in M'_i$ return a constraint in which for all $j \in [\ell]$ we put $M^{(\cdot)}_j = M_j \oplus \{\{u, v\}\}$, $D^{(\cdot)}=D$, and $B^{(\cdot)}= B\cup \{\{u, v\}\}$.
  \item For each $\{u, v\}\in M'_i$:
  \begin{compactitem}
   \item If there is no $x\in V\setminus D$ such that $\{u, x\}\in B$, then return a constraint with $D^{(\cdot)}= D\cup\{u\}$, $B^{(\cdot)}= B$, and for all $j \in [\ell]$ we put $M^{(\cdot)}_j = M_j \setminus \{\{u, v\}\}$. 
   \item If there is no $x\in V\setminus D$ such that $\{v, x\}\in B$, then return a constraint with $D^{(\cdot)}= D\cup\{v\}$, $B^{(\cdot)}= B$, and for all $j \in [\ell]$ we put $M^{(\cdot)}_j = M_j \setminus \{\{u, v\}\}$.
  \end{compactitem}
  \end{compactenum}
  \end{brule}
  \begin{lem}\label{lem:budget}
  If the
  \hyperref[crule:cleanup]{Clean-up Rule} was applied and \hyperref[rule:zero]{Rule~0} is not applicable, then
  \autoref{brule:budget} is a safe branching rule. 
\end{lem}
\begin{proof}
  It is easy to check that \autoref{brule:budget} is indeed a branching rule since it always modifies the pairs in the edge modifications sets of all layers, hence if the input constraint is aligning so are all output constraints. Since each output constraint extends the input constraint, by \autoref{obs:extends}, if any of the output constraints is good, then so is the input constraint. 

Now we show that if the input constraint is good, at least one output constraint is. Let $P=(D, M_1, \ldots, M_\ell, B)$ be the input constraint. Suppose that $P$ is good and let $S=(M^\star_1, \ldots, M^\star_\ell, D^\star)$ be a solution for the input instance witnessing that $P$ is good.
Since \hyperref[rule:zero]{Rule~0} is not applicable, we have $|M_i \cap B| \le k$ and, thus, 
$M_i \setminus B \neq \emptyset$. 

Since $|M'_i|+|M_i \cap B|=k+1$, $M_i \cap B \subseteq M^\star_i$, and $|M^\star_i| \le k$, we have $M'_i\setminus M^\star_i \neq \emptyset$, i.e., there is at least one vertex pair
$\{u, v\}\in M'_i$ such that $\{u, v\}\notin M^\star_i$.
The branching rule creates constraints for each possible
vertex pair in $M'_i$ to remove it from $M_i$. Thus, in particular, there is  
one output constraint where $\{u, v\}$ is removed from~$M_i$.

If $\{u, v\}\cap D^\star = \emptyset$, then, since the solution is consistent, either $\{u, v\}\in E_i\oplus M_i^\star$ for all $i \in[\ell]$ or $\{u, v\}\notin E_i\oplus M_i^\star$ for all $i \in[\ell]$. However, since $P$ is aligning, we also have that $\{u, v\}\in E_i\oplus M_i$ for all $i \in[\ell]$ or $\{u, v\}\notin E_i\oplus M_i$ for all $i \in[\ell]$ and furthermore, $\{u, v\}\in E_i\oplus M_i$ if and only if $\{u, v\}\notin E_i\oplus M_i^\star$. Since we have that $\{u, v\}\in E_i\oplus M_i$ if and only if $\{u, v\}\notin E_i\oplus M_i\oplus\{\{u, v\}\}$, one of the constraints in the first case is good.

Otherwise at least one of its endpoints is marked in~$S$ implying that one of the constraints in the second case is good.
\end{proof}
  
The last rule, \autoref{brule:budget2}, requires that all other rules are not applicable. In this case the non-marked vertices induce the same cluster graph in every layer. \autoref{brule:budget2} checks whether in every layer it is possible to turn the whole layer (including the marked vertices) into a cluster graph such that the cluster graph induced by the non-marked vertices stays the same and the edge modification budget is not violated in any layer. If this is not the case for a layer $i$, we will see that there are essentially two reasons for that. Either, (a), a modification in $M_i$ that was added greedily introduced many $P_3$'s containing marked vertices and the only way to remedy it is to roll back this modification. Or, (b), in order to make layer~$i$ a cluster graph including the marked vertices, we need to mark more vertices or make more edits outside of the marked vertices. Both cases will be treated by \autoref{brule:budget2} simultaneously. Since $M_i$ has bounded size, branching on the possibilities to roll back one of the edits (case~(a)) already results in a bounded number of branches. These possibilities are tested in Step~1 of \autoref{brule:budget2}. Case (b) is treated in Steps~2, 3, and 4. However, we need additional processing to bound the number of vertex markings or edge edits that we need to consider. To obtain the bound we introduce a modified version of a known kernelization algorithm~\cite{GrammGHN05} for classic \CE. We call it algorithm $K$ and it takes as input a tuple $(G, s, D, O)$. Herein, $G$ will represent the current, modified state of a layer, $D$ the currently marked vertices, $s$ the number of edits still allowed, and $O$ a set of vertex pairs that are \emph{obligatory}, meaning that they cannot be modified anymore. Algorithm~$K$ either outputs a distinct failure symbol or two sets $R$ and $C$, where $R$ contains all unmarked vertex pairs modified by~$K$ and $C$ contains all unmarked vertex pairs of the produced kernel which are not obligatory. (A vertex pair is unmarked if it does not contain a marked vertex.) In the following we give a formal description.

\newcommand{\modker}{\hyperref[kernel]{$K$}}
 \begin{kernel}
 \label{kernel}
 Given an input $(G,s,D,O)$. First, set all vertex pairs in~$O$ to \emph{obligatory} and exhaustively apply the following modified versions of standard data reduction rules for \CE. Let~$R=\emptyset$. Then, apply the following rules until none applies anymore.
\begin{compactenum}
\item[K1.] If $s < 0$ or there is an induced $P_3$ where all vertex pairs
  are obligatory, then abort and output a failure symbol.
\item[K2.] If a vertex pair $\{u, v\}$ is contained in the vertex set of
  $s + 1$ distinct induced $P_3$s of~$G$, then, if $\{u, v\}$ is obligatory,
  abort and output a failure symbol, otherwise modify $\{u, v\}$, set it to
  obligatory, and decrease~$s$ by one. If~$u\notin D$ and $v\notin D$,
  then add $\{u, v\}$ to~$R$.
\item[K3.] If there is an isolated clique, then remove it.
\end{compactenum}
Let $G^{(R)}$ be the resulting graph. If the number of vertices in $G^{(R)}$ is larger than~$s^2 + 2s$, then abort and output a failure symbol. Otherwise, let $C$ be the set of all unmarked vertex pairs  
in~$G^{(R)}$ which are not obligatory. 
Output $R$ and $C$. This concludes the description of~\hyperref[kernel]{$K$}.
\end{kernel}

In the description of the branching rule, we use the following notation.
  For all $1\le i\le \ell$ we use $\mathcal{M}_i$ to denote the set of all possible edge modifications where each edge is incident to at least one marked vertex, that turn $G'_i = (V, E_i\oplus M_i)$ into a cluster graph. More specifically, we have
  $$\mathcal{M}_i = \{M \subseteq \tbinom{V}{2} \mid \ \forall e\in M: e\cap D\neq\emptyset \ \wedge \ G''_i = (V, E_i\oplus(M_i\cup M)) \text{ is a cluster graph}\}.$$
Note that, since each $G'_i \setminus D$ is a cluster graph, each set $\mathcal{M}_i$ is non-empty.
    \begin{brule}
  \label{brule:budget2}
  If there is an $1\le i\le \ell$ such that $\min_{M\in \mathcal{M}_i}|M|>k-|M_i|$ then let $M'_i = M_i \setminus B$ and invoke the modified kernelization algorithm \modker\ on $(G'_i, k - |M_i|, D, M_i \cap B)$, where $G'_i = (V, E_i \oplus M_i)$. If \modker\ outputs a failure symbol and $M_i' = \emptyset$, then return \texttt{false}. If $M'\neq \emptyset$, then return the following constraints:
  \begin{compactenum}
  \item For each $\{u, v\}\in M'_i$: 
  \begin{compactitem}
   \item If there is no $x\in V\setminus D$ such that $\{u, x\}\in B$, then return a constraint with $D^{(\cdot)}= D\cup\{u\}$, $B^{(\cdot)}= B$, and for each $j \in [\ell]$ with $M^{(\cdot)}_j = M_j \setminus \{\{u, v\}\}$. 
   \item If there is no $x\in V\setminus D$ such that $\{v, x\}\in B$, then return a constraint with $D^{(\cdot)}= D\cup\{v\}$, $B^{(\cdot)}= B$, and for each $j \in [\ell]$ with $M^{(\cdot)}_j = M_j \setminus \{\{u, v\}\}$.
   \item Return a constraint in which for all $j\in[ \ell]$ we put $M^{(\cdot)}_j = M_j \oplus \{\{u, v\}\}$, $D^{(\cdot)}=D$, and $B^{(\cdot)}= B\cup \{\{u, v\}\}$.
  \end{compactitem}
  \end{compactenum}
If \modker\ does not output a failure symbol, then let $R$ and $C$ be the sets output by \modker\ and return the following constraints:  
\begin{compactenum}
\setcounter{enumi}{1}
  \item For each $\{u, v\}\in R$: 
  \begin{compactitem}
   \item If $u \notin D$ and there is no $x\in V\setminus D$ such that $\{u, x\}\in B$, then return a constraint with $D^{(\cdot)}= D\cup\{u\}$, $B^{(\cdot)}= B$, and for each $j \in [\ell]$ with $M^{(\cdot)}_j = M_j \setminus \{\{u, v\}\}$. 
   \item If $v \notin D$ and there is no $x\in V\setminus D$ such that $\{v, x\}\in B$, then return a constraint with $D^{(\cdot)}= D\cup\{v\}$, $B^{(\cdot)}= B$, and for each $j \in [\ell]$ with $M^{(\cdot)}_j = M_j \setminus \{\{u, v\}\}$.  
  \end{compactitem}
  \item If~$R \neq \emptyset$, then output a constraint with $D^{(\cdot)}= D$, $B^{(\cdot)}= B\cup M_i\cup R$, and $M^{(\cdot)}_j = M_j \oplus R$ for each $j \in [\ell]$.
    \item For each $\{u, v\}\in C$: 
    \begin{compactitem}
      \item If there is no $x\in V \setminus D$ such that $\{u, x\}\in B$, then return a constraint with $D^{(\cdot)}= D\cup\{u\}$, and the rest stays the same. 
      \item If there is no $x\in V\setminus D$ such that $\{v, x\}\in B$, then return a constraint with $D^{(\cdot)}= D\cup\{v\}$, and the rest stays the same.
      \item Return a constraint with $D^{(\cdot)}= D$, $B^{(\cdot)}= B
        \cup \{\{u,v\}\}$, and $M^{(\cdot)}_j = M_j \oplus \{\{u,
        v\}\}$ for each $j \in
        [\ell]$.
    \end{compactitem}
  \end{compactenum}
  \end{brule}
  
  \begin{lem}\label{lem:budget2}
 If the \hyperref[crule:cleanup]{Clean-up Rule} was applied and Branching Rules~\ref{brule:clustering} and~\ref{brule:budget} are not applicable, then \autoref{brule:budget2} is a safe branching rule.
\end{lem}
\begin{proof}  
  It is easy to check that \autoref{brule:budget2} is indeed a branching rule since it always modifies the edge modifications sets of all layers, hence if the input constraint is aligning so are all output constraints. Since each output constraint extends the input constraint, by \autoref{obs:extends}, if any of the output constraints is good, then so is the input constraint. 

Now we show that if the input constraint is good, at least one output constraint is. Let the input constraint $P=(D, M_1, \ldots, M_\ell, B)$ be good and let $S=(M^\star_1, \ldots, M^\star_\ell, D^\star)$ be a solution for the input instance witnessing that $P$ is good.  
For each layer $i$, \autoref{brule:budget2} checks the minimum number of edge modifications involving at least one marked vertex to turn $G_i'$ into a cluster graph. Since $G_i'[V\setminus D]$ is already a cluster graph, this number always exists. Since \autoref{brule:budget2} is applicable, there is a layer~$i \in [\ell]$ such that $\min_{M \in \M_i}|M| > k - |M_i|$. Fix this layer~$i$ in the following.

Suppose that there is a vertex pair $\{u, v\} \in M_i \setminus
M^\star_i$. Since $P$ is good, we have $M_i \cap B =
M^\star_i \cap B$, giving $\{u, v\} \in M_i' = M_i \setminus B$. Thus, $M'_i
\neq
\emptyset$ which means that the branch is not rejected after applying~\modker. In other words, there is one modification in
$M_i$ which is not in the solution witnessing that
$P$ is good, similar to \autoref{brule:budget}. It follows from an
analogous argumentation to the one in the proof of
\autoref{lem:budget} that \autoref{brule:budget2} produces a good
constraint in Step~1. That is, \autoref{brule:budget2} is safe in this
case. Thus, from now on we assume $M_i \subseteq M^\star_i$.

We claim that \modker\ does not produce a failure symbol. We in fact
now show the stronger statement that~\modker\ produces $R$ and $C$
such that $R \subseteq M^\star_i \setminus M_i$. To obtain this, we
show the following Invariant~(I) to hold before and after each
application of a rule of \modker. Invariant~(I) states that 
\begin{compactenum}[(i)]
\item \modker\ has not produced a failure symbol,
\item each edit made by \modker\ is in $M^\star_i$, and
\item $s = k - |M_i| - |L|$, where $L$ is the set of modifications
  made by \modker\ so far.
\end{compactenum}
Clearly, (I) holds in the beginning of $K$, before any application of a rule.
Since Rule~K2 is the only rule that makes modifications, and it clearly
maintains (I)~(iii), we will focus on (I)~(i) and~(ii).
Furthermore,~(I)~is clearly maintained by Rule~K3. It remains to treat
Rules~K1 and~K2. 

Consider Rule~K1. Let~$L$ be the set of modifications made by \modker\
so far. By (I)~(ii) we have $L \subseteq M^\star_i$. Observe that
$L \cap M_i = \emptyset$ since each pair in $M_i$ is obligatory. Hence, $(L \cup M_i)\subseteq M^\star_i$ which, since $M^\star_i$ is part of a
solution, implies that there are no induced $P_3$ where all three vertex pairs are obligatory.
Furthermore, we have that $|M^\star_i| \geq |M_i| + |L|$ and hence $k \geq |M^\star_i| \geq |M_i| + |L|$. By (I)~(iii)
we have~$s = k - |M_i| - |L|$. Thus,
$s \geq |M^\star_i| - |M_i| - |L| \geq 0$, meaning that no failure
symbol is produced by Rule~K1. Hence, Rule~K1 maintains Invariant~(I).

Now consider Rule~K2. Assume that the pair $\{u, v\}$ edited by Rule~K2
is not in $M_i^\star$. Since Rule~K2 applies, there are $s + 1$
distinct $P_3$s contained in the current
graph~$G^{(L)} \coloneqq (V, E_i \oplus (M_i \cup L))$, where~$L$ are
the modifications made by \modker\ so far. As $P$ is good,
$G^\star \coloneqq (V, E_i \oplus M_i^\star)$ is a cluster graph. To
compare $G^{(L)}$ and $G^\star$, recall that
$L \uplus M_i \subseteq M_i^\star$, where $\uplus$ denotes a disjoint
union: $M_i \subseteq M_i^\star$ by the considerations above,~$L \subseteq M_i$ by (I)~(ii), and $L \cap M_i = \emptyset$ because
each pair in $M_i$ is obligatory. Hence, for each of the induced~$P_3$s in $G^{(L)}$ there is at least one distinct vertex pair
in~$M^\star_i \setminus (L \cup M_i)$. Thus,~$|M^\star_i| \geq |L| + |M_i| + s + 1$. Since $s = k - |M_i| - |L|$,
we have $|M^\star_i| \geq k + 1$, a contradiction to the fact that
$M^\star_i$ is part of a solution. Thus, indeed
$\{u, v\} \in M_i^\star$. It follows that Invariant~(I) is maintained
by Rule~K2.

By Invariant~(I), after applying all rules in \modker\ we have
$L \subseteq M_i^\star \setminus M_i$, where $L=R$ is the set of
modifications made by \modker. We now bound the number of vertices in
$G^{(R)}$. Since $s = k - |M_i| - |L|$ which we obtain from
Invariant~(I)~(iii), we have
$|M_i^\star \setminus (M_i \cup L)| \leq s$ (recall that
$M_i \cap L = \emptyset$). Each vertex in $G^{(R)}$ is contained in an
induced~$P_3$. Each such $P_3$ contains a pair of
$M_i^\star \setminus (M_i \cup L)$. Each such pair is contained in at
most $s$ $P_3$s by inapplicability of Rule~K2. Thus,
graph~$G^{(R)}$ contains at most~$s^2 + 2s$ vertices. Thus,~\modker{}~does not produce a failure symbol. Furthermore, by Invariant~(I)~(ii),~$R \subseteq M^\star_i$ and, moreover, since no modification made by
\modker\ is in~$M_i$, $R \subseteq M^\star_i \setminus M_i$. Thus,
\modker\ produces the sets~$R \subseteq M^\star_i \setminus M_i$ and~$C$ as required.

Suppose that for one edge modification $\{u,v\}\in R$ we have that
$\{u,v\}\cap D^\star\neq \emptyset$, say $u \in D^\star$. Since~$P$ is
good, property~(ii) of being good gives that there is no pair
$\{u, w\} \in B$ for any $w \in V$. Thus, \autoref{brule:budget2}
outputs a good constraint in Step~2. Hence, we now assume that $R$
does not contain edge modifications containing vertices
from~$D^\star$.

Suppose that $R \neq \emptyset$. As argued above,
$R \subseteq M_i^\star \setminus M_i$. Since $D \subseteq D^\star$ and
no pair in~$R$ contains a vertex of $D^\star$, the constraint
produced in Step~3 is good. Thus, we now assume that $R = \emptyset$.

Suppose that $C$ contains a pair which contains a vertex in $D^\star$,
say~$u$. By property~(ii) of being good, there is no pair
$\{u, w\} \in B$ for any $w \in V$. Thus, one of the first group of
constraints produced in Step~4 is good. Thus, we now assume that no
pair in~$C$ contains a vertex in~$D^\star$ and hence also no
pair in~$C$ contains a vertex in~$D$.

Finally, we claim that $M_i^\star \cap C \neq \emptyset$. 
Suppose that $M_i^\star \cap C = \emptyset$. Since $R=\emptyset$, we have that 
$G^{(R)}$ is $G_i'$ with some isolated cliques removed.
Let $\widehat{M}_i$ be $M_i$ restricted to $G^{(R)}$ and, similarly,  
$\widehat{M}_i^\star$ be $M_i^\star$ restricted to $G^{(R)}$.
Note that, since $M_i \subseteq M_i^\star$ and $|M_i^\star| \le k$,
we have $|M_i^\star \setminus M_i| \le k - |M_i|$ and, thus, also 
$|\widehat{M}_i^\star \setminus \widehat{M}_i| \le k - |M_i|$.
Since $(V, E_i \oplus M_i^\star)$ is a cluster graph, also its 
subgraph induced by $V(G^{(R)})$ is a cluster graph, and, hence, 
also $(V, E_i \oplus (\widehat{M}_i^\star \cup M_i))$ is a cluster graph,
since the last two only differ in the isolated cliques.

Every pair of unmarked vertices in $G^{(R)}$ is in $\widehat{M}_i \cup C$
by the definition of $C$.
Hence, $\widehat{M}_i \subseteq \widehat{M}_i^\star$ and 
$\widehat{M}_i^\star \cap C = \emptyset$ implies 
$\widehat{M}_i^\star \setminus \widehat{M}_i \subseteq \binom{V}{2} 
\setminus \binom{V \setminus D}{2}$, and, therefore,
$(\widehat{M}_i^\star \setminus \widehat{M}_i) \in \mathcal{M}_i$.
As $|\widehat{M}_i^\star \setminus \widehat{M}_i| \le k - |M_i|$ and
$(V, E_i \oplus (\widehat{M}_i^\star \cup M_i))$ is a cluster graph,
this contradicts $\min_{M\in \mathcal{M}_i}|M|>k-|M_i|$.
Thus, indeed
$M_i^\star \cap C \neq \emptyset$. It follows that one of the last
group of constraints produced in Step~4 is good.
\end{proof}

  \begin{algorithm}[t]
  \DontPrintSemicolon

  \KwIn{\begin{compactitem}
    \item A set of graphs $G_1, \ldots, G_\ell = (V, E_1), \ldots, (V, E_\ell)$
    two integers $k$ and $d$.
    \item A set of marked vertices~$D$, edge modification sets $M_1, \ldots,
    M_\ell$.
    \item A set $B\subseteq \binom{V\setminus D}{2}$ of permanent vertex pairs.
  \end{compactitem}}
  \BlankLine
 	{Apply the first applicable rule in the following ordered list:
 	\begin{compactitem}
 	\item \hyperref[rule:zero]{Rule~0}
 	\item \hyperref[crule:cleanup]{Clean-up Rule}
 	\item \autoref{brule:clustering}
 	\item \autoref{brule:budget}
 	\item \autoref{brule:budget2}
 	\end{compactitem}} 
   If none of the rules applies, then \KwRet{\texttt{true}}
  \caption{\MLCE{}}
  \label{alg:mlce}
\end{algorithm}
With \autoref{brule:budget2} we can present the complete algorithm---see \autoref{alg:mlce}.
To prove correctness of the algorithm, we first argue that, whenever the algorithm outputs~\texttt{true}, then the input instance of \MLCE{} was indeed a yes-instance. This follows in a straightforward manner from the fact that, if the algorithm outputs \texttt{true}, then none of the branching rules is applicable.
  \begin{lem}
  \label{lem:correctness1}
  Given an instance $I$ of \MLCE{}, if \autoref{alg:mlce} outputs \texttt{true} on input $I$ and the constraint $P_\text{greedy}$, then $I$ is a yes-instance.
  \end{lem}
  
  \begin{proof}
  Let $I$ be the input instance of \MLCE{}. If the algorithm outputs \texttt{true}, then there is an aligning constraint $P=(D, M_1, \ldots, M_\ell, B)$ such that for all $e \in M_i$ we have that $e\cap D=\emptyset$, and none of the branching rules are applicable. 
  Let $D^\star = D$ and for every $i \in [l]$ let $M_i' = \arg\min_{M\in \mathcal{M}_i}|M|$ and $M^\star_i = M_i \cup M'_i$, where is $\mathcal{M}_i$ is as defined for \autoref{brule:budget2}.
  In the following we show that $S=(M^\star_1, \ldots, M^\star_\ell, D^\star)$ is a solution for $I$ (witnessing that $P$ is good). 
  
  Since \autoref{brule:budget2} is not applicable, we know that $|M_i'|\le k - |M_i|$ and hence $|M_i \cup M_i'|\le k$. 
  Also, since \hyperref[rule:zero]{Rule~0} is not applicable, we know that $|D^\star| = |D| \le d$.
  Let $G^\star_i = (v, E_i\oplus M^\star_i)$ for all $i \in [\ell]$.  
  For all $i, j \in [\ell]$ we have that $G^\star_i[V\setminus D] = G^\star_j[V\setminus D]$ since the constraint $P$ is aligning, and $M'_i$ contains no unmarked pairs. Furthermore, for all $i \in[ \ell]$ we have that $G^\star_i$ is a cluster graph by the definition of $\mathcal{M}_i$.
  \end{proof}
  
It remains to show that, whenever the input instance $I$ of the
algorithm is a yes-instance, then the algorithm outputs \texttt{true}.
To this end, we define the \emph{quality} of a good constraint and
show that the algorithm increases the quality until it eventually
finds a solution or determines that there is none.
  \begin{defi}[Quality of a constraint]
  Let $I =(G_1, \ldots, G_\ell, k, d)$ be an instance of \MLCE{}. The \emph{quality}~$\gamma_I(P)$ of a constraint $P = (D, M_1, \ldots, M_\ell, B)$ for $I$ is $\gamma_I(P) = |D| + |B|$. 
  \end{defi}
  \begin{lem}
  \label{lem:partialitydecrease}
    Let $P$ be a good constraint for a yes-instance of \MLCE{}. If applicable, each of 
 Branching Rules~\ref{brule:clustering},~\ref{brule:budget}, and~\ref{brule:budget2} returns a good constraint with strictly increased quality in comparison to~$P$.
\end{lem}
  \begin{proof}
  We show the claim individually for each of the rules. We consider each of the possible returned constraints~$P'$ and show that, assuming that $P'$ is good, then the quality of $P'$ is strictly larger than~$P$.

Consider \autoref{brule:clustering}. 
In the first three cases, the branching rule increases~$|B|$. In the remaining cases, the branching rule increases $|D|$. \autoref{brule:budget} increases $|B|$ in the first case and $|D|$ in the second case. \autoref{brule:budget2} also increases~$|B|$ or~$|D|$ in each of the four steps.
\end{proof}
 We can now show the correctness of \autoref{alg:mlce}. \autoref{lem:correctness1} ensures that we only output \texttt{true} if the input is actually a yes-instance and Lemma~\ref{lem:partialitydecrease}  together with the safeness of all branching rules ensures that if the input is a yes-instance, the algorithm outputs \texttt{true}.
  \begin{lemma}[Correctness of \autoref{alg:mlce}]
  \label{lem:correctness}
  Given an \MLCE{} instance $I$, \autoref{alg:mlce} outputs \texttt{true} on input $I$ and the initial constraint $P_\text{greedy}$ if and only if $I$ is a yes-instance.
  \end{lemma}
  \begin{proof}
  By \autoref{lem:correctness1}, if \autoref{alg:mlce} outputs \texttt{true} on input $I$ and the initial constraint $P_\text{greedy}$, then $I$ is a yes-instance. It remains to show the other direction.
  
  Let $I$ be a yes-instance of \MLCE{}. By \autoref{obs:initial} we have that $P_\text{greedy}$ is a good constraint. Note that the order in which rules are applied (see \autoref{alg:mlce}) ensures safeness for all branching rules (Lemmata~\ref{lem:clustering},~\ref{lem:cleanup},~\ref{lem:budget}, and~\ref{lem:budget2}). Furthermore, by \autoref{lem:partialitydecrease} we have that all branching rules except the \hyperref[crule:cleanup]{Clean-up Rule} strictly increase the quality of a good constraint. It is easy to see that the \hyperref[crule:cleanup]{Clean-up Rule} does not decrease the quality of a good constraint and it is applied only once before either one of the other rules apply or the algorithm terminates. 
  Let $P_\text{max}$ be a good constraint with the highest quality produced during the run of the algorithm. 
  Since the quality is an integer bounded above by $|V|+\binom{|V|}{2}$, there must be such a constraint.
  As the \hyperref[crule:cleanup]{Clean-up Rule} does not decrease the quality, we can also assume that it was exhaustively applied to $P_\text{max}$.
  If any of the branching rules would apply to $P_\text{max}$, then, by safeness of the branching rules and \autoref{lem:partialitydecrease}, it would produce a good constraint of strictly higher quality, contradicting the maximality of $P_\text{max}$. Hence the algorithm run on $P_\text{max}$ returned \texttt{true} and, therefore, the whole algorithm returned \texttt{true}.
  \end{proof}
  
  It remains to show that \autoref{alg:mlce} has the claimed running time upper-bound. We can check that all branching rules create at most $O(k^4)$ recursive calls. The preprocessing by the \hyperref[brule:greedy]{Greedy Rule} and the alignment of the constraints ensures that the edge modification sets in sufficiently many layers increase for the search tree to have depth of at most $O(k+d)$. The time needed to apply a branching rule is dominated by \autoref{brule:budget2}, where we essentially have to solve classical \CE{} in every layer.
   \begin{lem}
   \label{lem:runningtime}
   The running time of \autoref{alg:mlce} is $k^{O(k+d)}\cdot O(n^3\cdot \ell)$.
   \end{lem}
  \begin{proof}
    We bound the running time of \autoref{alg:mlce} by the following straightforward approach.
    Note that the recursive calls
    of \autoref{alg:mlce} define a tree in which each node corresponds
    to a call of \autoref{alg:mlce} and two nodes are connected by an
    edge if one of the corresponding calls of the algorithm is a
    recursive call of the other. The tree is rooted at the node
    corresponding to the initial constraint~$P_\text{greedy}$. Note that $P_\text{greedy}$ can be computed in $O(n^2\cdot\ell)$ time. We first bound the
    size of the search tree, and then the computation spent in each
    node of the search tree. Note that we apply \hyperref[crule:cleanup]{Clean-up Rule} at the beginning of each recursive call, that is, without creating further recursive calls. Hence we call this rule \emph{degenerate}.
 
    To bound the depth of the search tree, the length of a path from the root to the farthest leaf, we show that each (nondegenerate) branching rule increases either $|D|$ by at least one or it increases $\sum_{1\le i\le\ell}|M_i\cap B|$ by at least~$\frac{\ell}{2}$. If $|D|>d$ or $\sum_{1\le i\le\ell}|M_i\cap B| > \ell\cdot k$, then the algorithm terminates (\hyperref[rule:zero]{Rule~0}). 
    Before we show that, we prove an invariant that for every constraint produced by the algorithm and for every $\{u,v\} \in M_i \setminus B$ for some $i$ with $u,v \notin D$ we have $|\{j\mid \{u,v\} \in M_j\}| \le \frac{\ell}{2}$. 
    To this end, note that is initially fulfilled when $P_\text{greedy}$ is computed by the \hyperref[brule:greedy]{Greedy Rule} and whenever any other rule 
    touches a pair $\{u,v\}$ then in the produced constraint we have that $\{u,v\} \in B$, $u \in D$, or $v \in D$. That is, these rules cannot break the invariant. 
    
    Consider \autoref{brule:clustering}. In the first three cases $\sum_{1\le i\le\ell}|M_i\cap B|$ increases by at least $\frac{\ell}{2}$, since the pair was not in $B$ before the application of the rule and thus could appear in $M_i$ for at most $\frac{\ell}{2}$ different layers $i$ by the above proven invariant and, hence, after the application of the rule it appears in  $|M_i \cup B|$ for at least $\frac{\ell}{2}$ different layers~$i$. By a similar argument, also Branching Rules~\ref{brule:budget} and~\ref{brule:budget2} increase either~$|D|$ by one or $\sum_{1\le i\le\ell}|M_i\cap B|$ by at least $\frac{\ell}{2}$. Hence, we can upper-bound the depth of the search tree with~$2k+d$.
  
    Observe that the number of children of each node in the search tree where Branching Rule~\ref{brule:clustering}, or~\ref{brule:budget} is applied is upper-bounded by~$3k+3$. \autoref{brule:budget2} creates at most $3k$ recursive calls in the first step. In the second step it creates at most $2|R|$ recursive calls, at most one in the third step and in the fourth step the number of recursive calls created is at most $3|C|$. By the description of the modified kernelization algorithm~\modker\ we have $|R| \leq s$ and $|C| \leq (s^2 + 2s)^2$. By the way \modker\ is invoked by \autoref{brule:budget2} we have $s \leq k$ and, thus, the number of recursive calls is~$O(k^4)$. It follows that the size of the whole search tree is in $k^{O(k+d)}$.
  
  The \hyperref[crule:cleanup]{Clean-up Rule} plays a special role. The \hyperref[crule:cleanup]{Clean-up Rule} can be exhaustively applied in $O(|\binom{V}{2}|\cdot\ell) = O(n^2\cdot\ell)$ time on the beginning of each recursive call.
  
  Lastly, we analyze for each rule, how much time is needed to check whether the rule is applicable and if so to compute the constraints it outputs. 
  To check the applicability of \autoref{brule:clustering}, the algorithm needs to check whether there is a layer containing an induced~$P_3$. This can be done in $O(n + m)$ time, where $m$ is the maximum number of edges in a layer.\footnote{This can be done using breadth-first search (BFS) roughly in the following way: Start BFS at a vertex $v$ that is non-universal for its component (if such vertex does not exist, the component is already a clique). As soon as BFS reaches a vertex of distance two to $v$, you have found an induced $P_3$.} 
  Hence, overall we need $O((n + m)\cdot\ell)$ time to check whether \autoref{brule:clustering} is applicable and in this time we can also compute the output constraints.
   In the case of \autoref{brule:budget}, we need $O(n^2\cdot\ell)$ time to check whether it is applicable and to output the constraints. For the last rule, \autoref{brule:budget2}, we need to check whether the set families~$\M_i$ are nonempty. To do this we essentially need to solve \textsc{Cluster Editing} on each layer to check whether the rule is applicable.
  This can be done in $O(3^k \cdot (n + m))$ time, similarly to the straightforward algorithm for \textsc{Cluster Editing}. That is, recursively find a $P_3$ and branch into the at most three cases of modifying vertex pairs in the $P_3$ which contain at least one marked vertex. 
  The time to compute the constraints is dominated by the application of the modified kernelization algorithm~\modker. The data reduction rules of \modker\ can be applied exhaustively in  $O(n^3)$ time~\cite{GrammGHN05}. Hence, overall, the algorithm has running time~$k^{O(k+d)}\cdot O(n^3\cdot \ell)$.
  \end{proof}
This concludes the proof of \Cref{thm:fpt}.
  \begin{remark}
  It is not difficult to see that \MLCE{} can also be solved in $n^{O(n)}\cdot O(\ell)$ time, which is incomparable to the running time of \autoref{alg:mlce} since~$k$ might be as large as $\Omega(n^2)$: First guess the marked vertices. Then guess how many clusters (i.e.\ disjoint cliques) there are in the modified graph induced by the non-marked vertices, and for every non-marked vertex, guess to which cluster it belongs. Now for every layer, independently guess how many additional clusters there are consisting only of marked vertices, and for every marked vertex, guess to which cluster it belongs. Finally check, whether such a solution can be obtained by at most $k$ modifications per layer.
  \end{remark}

\section{\TCELong{} (\TCE)}
\label{sec:temporal}
\appendixsection{sec:temporal} In this section we turn to the temporal
version of the cluster editing problem. We provide an algorithm for
\TCE{} with a running time~$n^{O(k)}\ell$. We also show that the running time cannot be
substantially improved unless the Exponential Time
Hypothesis (ETH) fails.

\begin{theorem}
\label{thm:tempxpk}
  \TCE\ can be solved in $O(\ell\cdot n^{4k + 2} \log{n})$ time.
\end{theorem}
\begin{proof}
  Given an instance $((G_i)_{i \in [\ell]}, k, d)$ of \TCE, build an
  $\ell$-partite graph~$\cal G$ as follows:
  For each possible cluster editing set of~$G_i$ of
  size at most~$k$, add a vertex to the $i^{\text{th}}$ part of $\cal G$.
  %there is a vertex in the
 % $i$th part, $i \in [\ell]$, for each cluster editing set of~$G_i$ of
  %size at most~$k$. 
  Note that $\cal G$ contains $O(n^{2k}\ell)$
  vertices since each part contains $O(n^{2k})$ vertices. 
  For each~$i$, $1\le i \le n-1$, and each pair of vertices~$u, v$ in $\cal G$
  such that $u$ is in part~$i$ and $v$ is in part $(i+1)$
  add to $G$ the edge~$\{u, v\}$ if the algorithm of \cref{prop:mlce-two-layers} accepts on
  input of the following instance of \MLCE. Let $M_u, M_v$ be the
  cluster editing sets corresponding to~$u$ and~$v$, respectively. %, and
  %let $i, i + 1$ be the parts of $u$ and~$v$ in $\cal G$ (that is, the
 % indices of the corresponding layers). 
  The \MLCE\ instance consists
  of a multi-layer graph with the two layers $(V, E_i \oplus M_u)$
  and $(V, E_i \oplus M_v)$, edit budget equal to zero, and marking
  budget equal to~$d$. For each pair of vertices~$u,v$, 
  constructing the corresponding \MLCE\
  instance and solving it takes $O(n^2\log{n})$~time, amounting to
  overall~$O(\ell\cdot n^{4k + 2}\log{n})$ time, because there are at most
  $n^{4k}\ell$ pairs of vertices to consider. Finally, we test whether
  there is a path of length~$\ell$ from a vertex in the first part to a vertex in the
  last part in~$\cal G$. As there are at most $n^{4k} \ell$ edges in
  $\cal G$, this takes $O(n^{4k}\ell)$ time. Hence, overall the
  running time is $O(\ell \cdot n^{4k + 2}\log{n})$.

  If the algorithm above accepts, then the
  input instance is a yes-instance, because the path in $\cal G$
  corresponds to a sequence of cluster editing sets of size at
  most~$k$ for which any two consecutive cluster editing sets yield
  cluster graphs which differ in at most~$d$ vertices. To see that the
  algorithm above finds a solution $D_1, \ldots, D_\ell$,
  $M_1, \ldots, M_\ell$ if there is one, observe that each $M_i$
  corresponds to a vertex in~$\cal G$ and the existence of $D_i$
  witnesses that $M_i$ and $M_{i + 1}$ are adjacent in
  $\cal G$.
\end{proof}
\noindent\cref{thm:tempxpk} implies that \TCE\ is fixed-parameter tractable when parameterized by the number $n$ of vertices. 
% Furthermore, we can observe that the AND-cross-composition from \autoref{thm:nopolykernel-n} also works for \TCE. %Hence, we get the following corollary.
% 
% \begin{cor}[\appref{proof:cor:tempn}]\label{cor:tempn}
% \TCE\ is in \FPT\ when parameterized by the number $n$ of vertices but does not admit a polynomial kernel
% with respect to that parameter, unless \NoKernelAssume.
% \end{cor}
% \appendixproof{cor:tempn}{
% \begin{proof}
% The fact that \TCE\ is in \FPT\ when parameterized by $n$ follows directly from \autoref{thm:tempxpk} since $k\le n^2$. The AND-cross-composition from \autoref{thm:nopolykernel-n} allows all vertices to be marked, hence the same construction also works for \TCE{} and the results carries over.
% \end{proof}
% }
%
%
At first glance, it seems wasteful to iterate over all possible
cluster editing sets for each layer. Rather, the interaction between
two consecutive layers seems to be limited by~$k$ and~$d$, since the
necessary edits are local to induced~$P_3$, and the necessary markings
are local to incongruent clusters (perhaps resulting from destroying
$P_3$s). However, to our surprise, when the number of layers grows,
this interaction spirals out of control. As the reduction of the following hardness
result implies, we have to take into account splitting up small
clusters in an early layer (even though locally they were already
cliques), so as to be able to form cluster graphs a large number of
layers later on. This behavior stands in stark contrast to \MLCE,
where the combinatorial explosion is limited to~$k$ and~$d$.

\begin{theorem}%[\appref{proof:thm:tempw1hk}]
\label{thm:tempw1hk}
  \TCE\ is \W{1}-hard with respect to~$k$, even if~$d = 3$.
  Moreover, it does not admit an $f(k)(n\ell)^{o(k)}$-time algorithm unless the
  ETH fails.
\end{theorem}
\begin{proof}
  We reduce from the \W{1}-hard \MIS\
  problem~\cite{fellows_parameterized_2009} in which we are given an
  $s$-partite graph~$G$ with partite vertex sets~$V_1, \ldots, V_s$
  and we want to decide whether~$G$ contains an independent set
  containing one vertex of each partite set. For convenience we also
  say that vertices in~$V_p$, $p \in [s]$, are of \emph{color}~$p$.
  Without loss of generality assume that $|V_1| = \ldots = |V_s| = n$.
  For each~$p \in [s]$ denote $V_p = \{a^p_1, \ldots, a^p_n\}$.

  % We
  % refer by $i\in[n]$ to the vertices in $V_p$ for every $p\in[s]$ as
  % long as it is clear from the context in which set $V_p$ vertex $i$
  % belongs.

  Given $G$ and $V_1, \ldots, V_s$, we construct an instance of \TCE\
  with $\ell = 2sn + 3m$ layers
  $G_1 = (V, E_1), \ldots, G_\ell = (V, E_\ell)$. Herein, $m$ is the
  number of edges in $G$.

  \newcommand{\TCEwIB}{\textsc{TCEwIB}}
  Instead of maximum budgets $k$ and $d$ for editing edges and marking
  vertices over all layers, we specify, for each layer $j \in [\ell]$,
  budgets $k_j$ and $d_j$, meaning that we require in every solution
  that $|M_j| \leq k_j$ and $|D_j| \leq d_j$. Call this more general problem \TCE{} \textsc{with Individual Budgets} (\TCEwIB).
  In the following, by $k$
  we refer to $\max_{j \in [\ell]}k_j$ in the finished construction
  and by $d$ we refer to $\max_{j \in [\ell]}d_j$ in the finished
  construction. By a \emph{solution} we mean a solution to the \TCEwIB\
  instance (that is, it respects the individual budgets). An individual budget
  $k_j$ (resp.\ $d_j$) is \emph{saturated} if $k_j = |M_j|$ (resp.\
  $d_j = |D_j|$). After showing hardness for \TCEwIB\ we give a reduction to the plain \TCE\
  problem. % Let $\ell' \in \mathbb{N}$ such that $\ell' \leq \ell$. By
  % an \emph{$\ell'$-solution} we mean a solution to the \TCE\ instance
  % resulting from removing all layers in~$[\ell] \setminus [\ell']$.

  % Below we will force some vertices, say vertex~$u$, to be marked in
  % some layers. This we do by introducing cliques containing
  % $k + d + 1$~vertices into all layers, and by making all of them
  % adjacent to $u$ in some layers. Due to the size of the clique, the
  % only way to achieve consistency in a solution is to mark~$u$ in the
  % corresponding layers. We call the vertices in such cliques
  % \emph{helper} vertices. Due to the size of the cliques, no solution
  % can edit any incident edge and without loss of generality we can
  % assume that no helper vertex is marked. We will hence sometimes
  % tacitly exclude helper vertices from the discussion below.

  % We will
  % simultaneously construct the instance and argue about the structure
  % of a solution (if any exists). For this purpose, we suppose that a
  % solution exists and denote it by $M_1, \ldots, M_\ell$,
  % $D_1, \ldots, D_{\ell - 1}$. We treat the other direction in the
  % end.

  The first $2sn$ layers in the temporal graph that we construct
  contain vertex-selection gadgets, selecting for each color in the
  \MIS\ instance one vertex. The remaining $3m$ layers contain
  verification gadgets, which ensure that no two adjacent vertices
  have been selected. Each vertex~$a^p_i \in V_p$, $p \in [s]$, will
  be represented by a pair of vertices $u^p_i, v^p_i \in V$ (introduced in the vertex-selection gadgets below).
  % Thus,
  % initially put $V = \{u^p_i, v^p_i \mid p \in [s] \wedge i \in [n]\}$
  % (more vertices will be added later). 
  In most layers, pair
  $u^p_i, v^p_i$ will be adjacent and the only possible solutions will
  remove or keep the corresponding edge in most layers consistently.
  If the edge~$\{u^p_i, v^p_i\}$ is removed from a layer, then this
  will signify that the corresponding vertex~$a^p_i \in V_p$ shall be in
  the independent set.

  % In the construction, the following operation to modify the temporal
  % graph~$(G_i)_{i \in [s]}$ will be useful. Let $v \in V$ and
  % $i \in [s]$ such that $v$ is isolated in $G_{i}$. When we say that
  % we \emph{force a mark on~$v$ in layer~$i$} we mean (1) to increase
  % the marking budget~$d_i$ by one, (2) to introduce a clique~$C$
  % consisting of $k + d + 1$ new vertices into all layers, and (3) to
  % make each vertex in $C$ adjacent to~$v$ in layer~$i$. The following
  % observation is straightforward and allows us to ignore the
  % clique~$C$ when arguing about solutions.
  % \begin{prop}
  %   If there are cluster editing sets $(M_i)_{i \in [\ell]}$ and sets~$ $ of marked vertices. 
    
  %   If we force a mark on~$v$ in layer~$i$ and if $D_i$ is the set of
  %   marked vertices in layer~$i$ of some solution, then $v \in D_i$.
  % \end{prop}

  \proofparagraph{Vertex Selection.}
  % The selection gadgets use no further vertices to the vertices
  % defined above, that is, the vertices introduced later on will not
  % have any incident edges in the layers of the selection gadgets. Thus
  % we may ignore these additional vertices in the following discussion.
  We construct the selection gadgets as follows.
  For each color $p \in [s]$ in the \MIS\ instance, there is one
  \emph{vertex-selection gadget for color~$p$} occupying layers
  $2(p - 1)n + 1, 2(p - 1)n + 2, \ldots, 2sn$. The first
  $2n$~layers of the gadget for color~$p$ comprise a tree-like
  construction in which, in the $2i$-th layer,
  $i \in [n - 1]$, we either have already selected one of the first~$i - 1$ vertices in~$V_p$, or we are forced to select either the
  $i$th vertex or one of the last $n - i$~vertices. The remaining
  layers $2pn + 1, 2pn + 2, \ldots, 2sn$ (if any) of
  the gadget simply transmit the choice to layer~$2sn + 1$.

  For each color $p \in [s]$ and each layer~$2i$,
  $i \in [n - 1]$, in the vertex-selection gadget for color~$p$ we use
  another gadget that either is inactive (when one of the first
  $i - 1$ vertices has been selected) or forces a choice between
  vertex~$i$ or one of the later $n - i$ vertices. Recall that the
  selection of a vertex~$a^p_i \in V_p$ into the independent set shall
  correspond to the deletion of the edge~$\{u^p_i, v^p_i\}$. Hence,
  the following construction will be useful, in which the deletion of
  a special input edge $\{u, v\}$ will force the deletion of either of
  two special output edges~$\{w, x\}$, $\{y, z\}$.

  \newcommand{\selcons}{\textsf{selection}}

  Let $I=((G_j)_{j \in [i]}, 
  (k_j)_{j \in [i]}, 
  (d_j)_{j \in [i-1]})$ 
  be an instance of \TCEwIB{} with $i>1$ layers. 
  Let $V=V(G_1)$ be the vertex set of the temporal graph in $I$, $u, v \in V$, $c, w, x, y, z \notin V$  
  such that $\{u, v\} \in E(G_i)$ and no further edges are 
  incident with $u$ and $v$ in $G_i$. 
  The function $\selcons(I, u, v, w, x, y, z)$ produces 
  a new instance $\widehat{I}= ((\widehat{G}_j)_{j \in [i+2]}, 
  (\widehat{k}_j)_{j \in [i+2]}, 
  (\widehat{d}_j)_{j \in [i+1]})$ of \TCEwIB{} 
  with $i+2$ layers as follows. 
  See \cref{fig:tempw1hk} for an illustration.

%   Let $V$ be a vertex set, $u, v \in V$, $c, w, x, y, z \notin V$,
%   $i \in \mathbb{N}$ such that $i > 1$, and
%   $(G_j)_{j \in [i + 2]}$ a temporal graph
%   with $i$~layers and vertex set~$V$ such that $\{u, v\} \in E(G_i)$
%   and no further edges are incident with $u$ and $v$ in layer~$i$. The
%   construction $\selcons(i, u, v, w, x, y, z)$ proceeds as
%   follows to modify $(G_j)_{j \in [i + 2]}$. See \cref{fig:tempw1hk} for an illustration.

  Start by taking $\widehat{G}_{\tilde{\imath}}=G_{\tilde{\imath}}$ for every
  $\tilde{\imath} \in [i]$ and $\widehat{G}_{i+1}=\widehat{G}_{i+2}=G_i$.
  Then introduce the vertices~$c, w, x, y, z$ into~$V$. Introduce
  four cliques~$C_c, C_v, C_x, C_z$ into all layers, each clique
  consisting of $2k + d + 1$ new vertices. Make each vertex in $C_c$
  adjacent with $c$ in all layers except for~$i + 1$, make each vertex
  in $C_v$ adjacent with $v$ in layer~$i + 2$ (and no other layer),
  make each vertex in $C_x$ adjacent with $x$ in all layers up to and
  including~$i$, and make each vertex in $C_z$ adjacent with $z$ in
  all layers up to and including~$i$. Introduce the edges
  $\{c, u\}, \{c, v\}$ in layer $i + 1$, introduce the edges
  $\{w, x\}, \{y, z\}$ in layer $i + 2$ and remove the edge $\{u, v\}$
  from layer $i + 2$. 
  
  Define the editing budgets as $\widehat{k}_{\tilde{\imath}}=k_{\tilde{\imath}}$ 
  for every $\tilde{\imath} \in [i]$, $\widehat{k}_{i + 1} = k_i + 2$, and $\widehat{k}_{i + 2} = k_i$.
  Similarly, let the marking budgets be $\widehat{d}_{\tilde{\imath}}=d_{\tilde{\imath}}$ 
  for every $\tilde{\imath} \in [i-1]$, $\widehat{d}_{i} = 3$, $\widehat{d}_{i + 1} = 2$.
  This concludes the construction.

  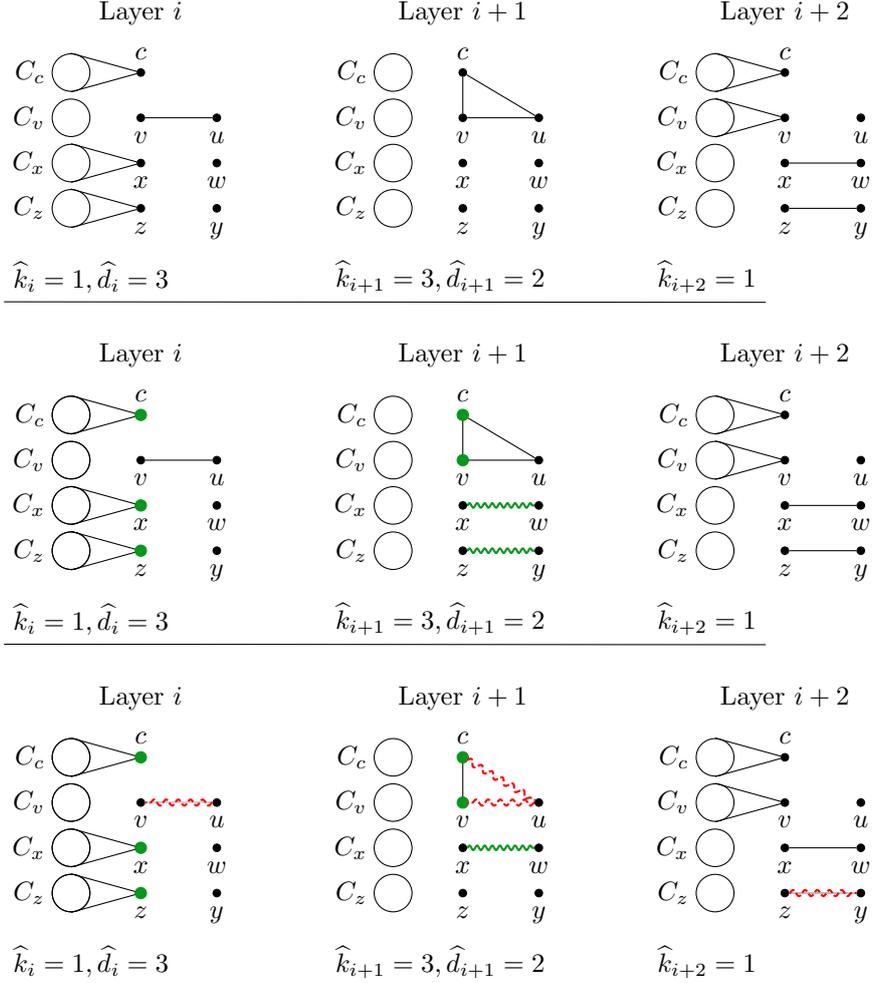
\begin{figure}[t!]
  \centering
  \begin{tikzpicture}
    % Layer 1
    \def \xs {28ex}
    \def \xss {56ex}
    \def \ys {-30ex}
    \def \yss {-60ex}
    \def \scc {.6}
    \def \aa {3ex}
    \def \bb {6ex}
  \begin{scope}
    \foreach \n / \x / \y / \s in  {c/0/1/above, u/1/0/below,v/0/0/below,w/1/-1/below,x/0/-1/below,y/1/-2/below,z/0/-2/below} {
      \node[nnode] at (\x,\y*\scc) (\n) {};
      \node[\s = 0pt of \n] {$\n$};
    }
    \foreach \s in {c,v,x,z} {
      \node[bnode, left = 4ex of \s] (C\s) {};
      \node[left = -1pt of C\s] (nC\s) {$C_\s$};
    }
    \foreach \s in {c,x,z} {
      \draw (\s) edge (C\s.south);
      \draw (\s) edge (C\s.north);
    }
    \foreach \s / \t in {u/v} {
      \draw (\s) edge (\t);
    }
    \node[above = \aa of c] {Layer $i$};
    % Parameters
    \node[below = \bb of nCz.west,anchor=west] (par11) {$\widehat{k}_i=1, \widehat{d}_i=3$};
  \end{scope}
    % Layer 2
  \begin{scope}[xshift=\xs]
    \foreach \n / \x / \y / \s in  {c/0/1/above, u/1/0/below,v/0/0/below,w/1/-1/below,x/0/-1/below,y/1/-2/below,z/0/-2/below} {
      \node[nnode] at (\x,\y*\scc) (\n) {};
      \node[\s = 0pt of \n] {$\n$};
    }
    \foreach \s in {c,v,x,z} {
      \node[bnode, left = 4ex of \s] (C\s) {};
      \node[left = -1pt of C\s] (nC\s) {$C_\s$};
    }

    \foreach \s / \t in {u/v,c/u,c/v} {
      \draw (\s) edge (\t);
    }
    \node[above = \aa of c] {Layer $i+1$};
    \node[below = \bb of nCz.west,anchor=west] {$\widehat{k}_{i+1}=3, \widehat{d}_{i+1}=2$};
  \end{scope}  
  % Layer 3
  \begin{scope}[xshift=\xss]
    \foreach \n / \x / \y / \s in  {c/0/1/above, u/1/0/below,v/0/0/below,w/1/-1/below,x/0/-1/below,y/1/-2/below,z/0/-2/below} {
      \node[nnode] at (\x,\y*\scc) (\n) {};
      \node[\s = 0pt of \n] {$\n$};
    }
    \foreach \s in {c,v,x,z} {
      \node[bnode, left = 4ex of \s] (C\s) {};
      \node[left = -1pt of C\s] (nC\s) {$C_\s$};
    }
    \foreach \s in {c,v} {
      \draw (\s) edge (C\s.south);
      \draw (\s) edge (C\s.north);
    }
    \foreach \s / \t in {w/x,z/y} {
      \draw (\s) edge (\t);
    }
    \node[above = \aa of c] {Layer $i+2$};
    \node[below = \bb of nCz.west,anchor=west] (par31) {$\widehat{k}_{i+2}=1$};
  \end{scope}
  \draw (par11.south west) edge (par31.south east);

  % Sol 1: Layer 1
   \begin{scope}[yshift=\ys]
    \foreach \n / \x / \y / \s in  {c/0/1/above, u/1/0/below,v/0/0/below,w/1/-1/below,x/0/-1/below,y/1/-2/below,z/0/-2/below} {
      \node[nnode] at (\x,\y*\scc) (\n) {};
      \node[\s = 0pt of \n] {$\n$};
    }
    \foreach \s in {c,v,x,z} {
      \node[bnode, left = 4ex of \s] (C\s) {};
      \node[left = -1pt of C\s] (nC\s) {$C_\s$};
    }
    \foreach \s in {c,x,z} {
      \draw (\s) edge (C\s.south);
      \draw (\s) edge (C\s.north);
    }
    \foreach \s / \t in {u/v} {
      \draw (\s) edge (\t);
    }
    \node[above = \aa of c] {Layer~$i$};
    % Parameters
    \node[below = \bb of nCz.west,anchor=west] (par12) {$\widehat{k}_i=1, \widehat{d}_i=3$};

    % Sols
    % Marked vertices D
    \foreach \s in {c,x,z} {
      \node[bnode, deletedv] at (\s) {};
    }  
    
    \foreach \s in {c,v,x,z} {
      \node[bnode] at (C\s) {};
    }
  \end{scope}
  % Sol 1: Layer 2
  \begin{scope}[xshift = \xs,yshift=\ys]
    \foreach \n / \x / \y / \s in  {c/0/1/above, u/1/0/below,v/0/0/below,w/1/-1/below,x/0/-1/below,y/1/-2/below,z/0/-2/below} {
      \node[nnode] at (\x,\y*\scc) (\n) {};
      \node[\s = 0pt of \n] {$\n$};
    }
    \foreach \s in {c,v,x,z} {
      \node[bnode, left = 4ex of \s] (C\s) {};
      \node[left = -1pt of C\s] (nC\s) {$C_\s$};
    }

    \foreach \s / \t in {u/v,c/u,c/v} {
      \draw (\s) edge (\t);
    }
    \node[above = \aa of c] {Layer $i+1$};
    \node[below = \bb of nCz.west,anchor=west] {$\widehat{k}_{i+1}=3, \widehat{d}_{i+1}=2$};   
    % Marked vertices D
    \foreach \s in {c,v} {
      \node[bnode, deletedv] at (\s) {};
    }  
        
    % Added edges 
    \foreach \s / \t in {x/w,z/y} {
      \draw  (\s) edge[added] (\t);
    }
  \end{scope}  
  % Sol 1: Layer 3
  \begin{scope}[xshift=\xss,yshift=\ys]
    \foreach \n / \x / \y / \s in  {c/0/1/above, u/1/0/below,v/0/0/below,w/1/-1/below,x/0/-1/below,y/1/-2/below,z/0/-2/below} {
      \node[nnode] at (\x,\y*\scc) (\n) {};
      \node[\s = 0pt of \n] {$\n$};
    }
    \foreach \s in {c,v,x,z} {
      \node[bnode, left = 4ex of \s] (C\s) {};
      \node[left = -1pt of C\s] (nC\s) {$C_\s$};
    }
    \foreach \s in {c,v} {
      \draw (\s) edge (C\s.south);
      \draw (\s) edge (C\s.north);
    }
    \foreach \s / \t in {w/x,z/y} {
      \draw (\s) edge (\t);
    }
    \node[above = \aa of c] {Layer $i+2$};
    \node[below = \bb of nCz.west,anchor=west] (par32) {$\widehat{k}_{i+2}=1$};
  \end{scope}  

  \draw (par12.south west) edge (par32.south east);

  % Sol 2: Layer 1
   \begin{scope}[yshift=\yss]
    \foreach \n / \x / \y / \s in  {c/0/1/above, u/1/0/below,v/0/0/below,w/1/-1/below,x/0/-1/below,y/1/-2/below,z/0/-2/below} {
      \node[nnode] at (\x,\y*\scc) (\n) {};
      \node[\s = 0pt of \n] {$\n$};
    }
    \foreach \s in {c,v,x,z} {
      \node[bnode, left = 4ex of \s] (C\s) {};
      \node[left = -1pt of C\s] (nC\s) {$C_\s$};
    }
    \foreach \s in {c,x,z} {
      \draw (\s) edge (C\s.south);
      \draw (\s) edge (C\s.north);
    }
    \foreach \s / \t in {u/v} {
      \draw (\s) edge (\t);
    }
    \node[above = \aa of c] {Layer~$i$};
    % Parameters
    \node[below = \bb of nCz.west,anchor=west] {$\widehat{k}_i=1, \widehat{d}_i=3$};

    % Sols
    % Marked vertices D
    \foreach \s in {c,x,z} {
      \node[bnode, deletedv] at (\s) {};
    }  
    
    \foreach \s in {c,v,x,z} {
      \node[bnode] at (C\s) {};
    }
    % Deleted edges 
    \foreach \s / \t in {u/v} {
      \draw  (\s) edge[white] (\t);
      \draw  (\s) edge[deletede] (\t);      
    }
  \end{scope}

  % Sol 2: Layer 2
  \begin{scope}[xshift = \xs,yshift=\yss]
    \foreach \n / \x / \y / \s in  {c/0/1/above, u/1/0/below,v/0/0/below,w/1/-1/below,x/0/-1/below,y/1/-2/below,z/0/-2/below} {
      \node[nnode] at (\x,\y*\scc) (\n) {};
      \node[\s = 0pt of \n] {$\n$};
    }
    \foreach \s in {c,v,x,z} {
      \node[bnode, left = 4ex of \s] (C\s) {};
      \node[left = -1pt of C\s] (nC\s) {$C_\s$};
    }

    \foreach \s / \t in {%u/v,c/u,
    c/v} {
      \draw (\s) edge (\t);
    }
    \node[above = \aa of c] {Layer $i+1$};
    \node[below = \bb of nCz.west,anchor=west] {$\widehat{k}_{i+1}=3, \widehat{d}_{i+1}=2$};   
    % Marked vertices D
    \foreach \s in {c,v} {
      \node[bnode, deletedv] at (\s) {};
    }  
        
    % Added edges 
    \foreach \s / \t in {x/w} {
      \draw  (\s) edge[added] (\t);
    }
    % Deleted edges 
    \foreach \s / \t in {c/u,u/v} {
      \draw  (\s) edge[white] (\t);
      \draw  (\s) edge[deletede] (\t);      
    }
  \end{scope}  
  % Sol 2: Layer 3
  \begin{scope}[xshift=\xss,yshift=\yss]
    \foreach \n / \x / \y / \s in  {c/0/1/above, u/1/0/below,v/0/0/below,w/1/-1/below,x/0/-1/below,y/1/-2/below,z/0/-2/below} {
      \node[nnode] at (\x,\y*\scc) (\n) {};
      \node[\s = 0pt of \n] {$\n$};
    }
    \foreach \s in {c,v,x,z} {
      \node[bnode, left = 4ex of \s] (C\s) {};
      \node[left = -1pt of C\s] (nC\s) {$C_\s$};
    }
    \foreach \s in {c,v} {
      \draw (\s) edge (C\s.south);
      \draw (\s) edge (C\s.north);
    }
    \foreach \s / \t in {w/x,z/y} {
      \draw (\s) edge (\t);
    }
    \node[above = \aa of c] {Layer $i+2$};
    \node[below = \bb of nCz.west,anchor=west] {$\widehat{k}_{i+2}=1$};
     % Deleted edges 
    \foreach \s / \t in {z/y} {
      \draw  (\s) edge[white] (\t);
      \draw  (\s) edge[deletede] (\t);      
    }
 \end{scope}  
\end{tikzpicture}
\label[figure]{fig:tempw1hk}
\caption{The crucial subgadget constructed by \selcons\ which is used
  in the vertex-selection gadgets in the proof of \cref{thm:tempw1hk}.
  The first row depicts the gadget (three layers). The second and the
  third row depict two possible selections. Budgets are given for the
  case that the input instance to \selcons\ has $k_i = 1$. A node
  drawn as a circle represents a clique while a node drawn as a solid
  dot represents a vertex. Solid lines represent edges. Zigzag lines
  represent edge modifications: red dashed zigzag lines mean edge
  deletion while green zigzag lines mean edge addition. A vertex in
  bold green means that it will be marked.}
\end{figure}

  % Let $a, b \in \mathbb{N}$ such that $a < b$ and $u, v \in V$. Let
  % $I$ be an instance of \TCE\ with individual budgets that contains a
  % temporal graph~$(G_{i})_{i \in [b]}$ such that $\{u, v\} \in E(G_a)$
  % and for each $j \in \{a + 1, \ldots, b\}$ we have
  % $\{u, v\} \notin E(G_j)$. A sequence $(D_i, M_i)_{i \in [a]}$ is
  % called a \emph{$(u, v, a)$-lenient} solution for $I$, if it is a
  % solution except that $G'_a$ and $G'_{a + 1}$ may differ in the
  % edge~$\{u, v\}$, that is, we allow that $\{u, v\} \in E(G'_a)$ and
  % $\{u, v\} \notin E(G'_{a + 1})$.

  In the following, we say that an instance~$I$ of \TCEwIB\ is \emph{sane} in layer~$i$ for vertices $u$
  and~$v$ if each solution for~$I$ has the property that the editing
  budget is saturated in layer~$i$ and $u$ and~$v$ either form
  singleton clusters~$\{u\}, \{v\}$ or a cluster~$\{u, v\}$ of size
  two in layer~$i$.
  
  \begin{lem}\label[lemma]{lem:selcons}
    Let $I=((G_j)_{j \in [i]}, 
  (k_j)_{j \in [i]}, 
  (d_j)_{j \in [i-1]})$ be an instance of \TCEwIB\, 
  $u,v,c,w,x,y,z$ as above, 
  and $\widehat{I}= ((\widehat{G}_j)_{j \in [i+2]}, 
  (\widehat{k}_j)_{j \in [i+2]}, 
  (\widehat{d}_j)_{j \in [i+1]})$
    the instance resulting from $\selcons(I, u, v, w, x, y, z)$. 
    Let $W$ be the set of new
    vertices introduced by \selcons. That is, $W$ contains
    $c, w, x, y, z$, and the new vertices in the
    cliques~$C_c, C_v, C_x, C_z$. The following properties hold.
    \begin{lemenum}
    \item If $I$ has a solution $((M_j)_{j \in
      [i]}, (D_j)_{j \in [i-1]})$ such that $\{u, v\}$ forms a cluster
      in $G_i \oplus M_i$, then $\widehat{I}$ has a solution $((\widehat{M}_j)_{j \in
      [i+2]}, (\widehat{D}_j)_{j \in [i+1]})$ such that both $\{w, x\}$
      and $\{y, z\}$ form clusters in $\widehat{G}_{i+2} \oplus \widehat{M}_{i+2}$,
      $(\widehat{G}_{i+2} \oplus \widehat{M}_{i+2}) - (W \cup \{u, v\}) = (G_i \oplus M_i) - \{u, v\}$,
      and there are no edges between $V(G_i)\setminus \{u, v\}$ and $(W \cup \{u, v\})$ in $\widehat{G}_{i+2} \oplus \widehat{M}_{i+2}$.
    \item If $I$ has a solution $((M_j)_{j \in
      [i]}, (D_j)_{j \in [i-1]})$ such that $\{u\}$ and $\{v\}$ are singleton
      clusters in $G_i \oplus M_i$, then $\widehat{I}$ has a solution $((\widehat{M}_j)_{j \in
      [i+2]}, (\widehat{D}_j)_{j \in [i+1]})$ such that in $\widehat{G}_{i+2} \oplus \widehat{M}_{i+2}$
      either (a)~$\{w, x\}$, $\{y\}$, and $\{z\}$ each form
      clusters or (b) $\{w\}$, $\{x\}$, and $\{y, z\}$ each form
      clusters. Moreover,
      $(\widehat{G}_{i+2} \oplus \widehat{M}_{i+2}) - (W \cup \{u, v\}) = (G_i \oplus M_i) - \{u, v\}$
      and there are no edges between $V(G_i)\setminus \{u, v\}$ and $(W \cup \{u, v\})$ in $\widehat{G}_{i+2} \oplus \widehat{M}_{i+2}$.
    % \item Each solution to $\widehat{I}$ saturates the marking budgets for
    %   layers $i$ and $i + 1$ and each marked vertex is
    %   in~$\{v\} \cup W$.
    \item Suppose that $I$ is sane in layer~$i$ for $u$ and~$w$. Then $\widehat{I}$ is
      sane in layer~$i + 2$ for $w$ and~$x$ as well as for~$y$
      and~$z$. Moreover, for each solution $((\widehat{M}_j)_{j \in
      [i+2]}, (\widehat{D}_j)_{j \in [i+1]})$ for $\widehat{I}$, we have 
      $(\widehat{G}_{i+2} \oplus \widehat{M}_{i+2}) - (W \cup \{u, v\}) = (\widehat{G}_i \oplus \widehat{M}_i) - \{u, v\}$
      and there are 
      no edges between $V(G_i)\setminus \{u, v\}$ and $(W \cup \{u, v\})$ 
      in $\widehat{G}_{i+2} \oplus \widehat{M}_{i+2}$.
      Furthermore, each solution $((\widehat{M}_j)_{j \in
      [i+2]}, (\widehat{D}_j)_{j \in [i+1]})$ for $\widehat{I}$ such
      that $\{u\}$ and $\{v\}$ form singleton clusters in $\widehat{G}_{i} \oplus \widehat{M}_{i}$ also has
      the property in $\widehat{G}_{i+2} \oplus \widehat{M}_{i+2}$ that $\{w\}$ and $\{x\}$ form singleton
      clusters or that $\{y\}$ and $\{z\}$ form singleton clusters.
    \end{lemenum}
  \end{lem}
  \begin{proof}
    Let $\widehat{V}=V(\widehat{G}_1)$ be the new vertex set.

    \proofparagraph{(i).} Let
    $((M_j)_{j \in
      [i]}, (D_j)_{j \in [i-1]})$
      be a solution to $I$ such that $\{u, v\}$ forms a cluster
    in $G_i \oplus M_i$. Modify this solution as follows to obtain a solution 
    $((\widehat{M}_j)_{j \in
      [i+2]}, (\widehat{D}_j)_{j \in [i+1]})$
    for~$\widehat{I}$. See the middle pane in \cref{fig:tempw1hk} for an
    illustration of the modifications. 
    Let $\widehat{M}_{j}=M_{j}$ for every $j \in [i]$ 
    and $\widehat{D}_{j}=D_{j}$ for every $j \in [i-1]$.
    Further, put $\widehat{M}_{i + 1} = M_i \cup \{\{w, x\}, \{y, z\}\}$, $\widehat{M}_{i + 2} = M_i$,
    $\widehat{D}_{i} = \{c, x, z\}$, and $\widehat{D}_{i + 1} = \{c, v \}$. Note that each individual budget is
    satisfied and, clearly, $\{w, x\}$ and $\{y, z\}$ form clusters in
    $\widehat{G}_{i+2} \oplus \widehat{M}_{i+2}$. Moreover, by definition, the equality condition on
    $(\widehat{G}_{i+2} \oplus \widehat{M}_{i+2}) - (W \cup \{u, v\})$ and 
    $(G_i \oplus M_i) - \{u, v\}$ as well as the non-existence of edges between $\widehat{V} \setminus (W \cup \{u,v\})$ and $(W \cup \{u,v\})$ are satisfied. Let $\widehat{G}'_j=\widehat{G}_j \oplus M_j$. It remains to
    show that for every $j \in [i+2]$ the graph $\widehat{G}'_j$ 
    is a cluster graph and each two consecutive $\widehat{G}'_j$'s are consistent.

    For every $j \in [i]$ let $G'_j= G_j \oplus M_j$.
    We first show that each $\widehat{G}'_j$, $j \in [i+2]$, is a cluster graph. Consider first the
    layers up to~$i$. Observe that, for each $j \in [i]$, $\widehat{G}'_j[W]$ is
    a cluster graph and that there are no edges in $\widehat{G}'_j$ between $W$
    and $\widehat{V} \setminus W$. Thus, it is enough to observe that
    $\widehat{G}'_j[\widehat{V} \setminus W]$ is a cluster graph. This is indeed the case
    because, for each $j \in [i]$, $\widehat{G}'_{j} - W = G'_j$ and $M_j$ is a cluster editing
    set for $G_j$. Thus, $\widehat{G}'_j$ is a cluster graph.

    Similarly, for both $j \in \{i + 1, i + 2\}$,
    $\widehat{G}'_j[W \cup \{u, v\}]$ is a cluster graph (for $j = i + 1$,
    observe that the only edits within $W \cup \{u, v\}$ join
    singleton clusters into clusters of size two) and there are no
    edges in $\widehat{G}'_j$ between $W \cup \{u, v\}$ and
    $\widehat{V} \setminus (W \cup \{u, v\})$. Since
    $\widehat{G}_i - (W \cup \{u, v\}) = \widehat{G}_{i + 1} - (W \cup \{u, v\}) = \widehat{G}_{i +
      2} - (W \cup \{u, v\})$ (as noted above), $\widehat{M}_{i + 1}$ and $\widehat{M}_{i + 2}$ 
      restricted to $\widehat{V} \setminus (W \cup \{u, v\})$ are cluster
    editing sets for $\widehat{G}_{i + 1} - (W \cup \{u, v\})$ and
    $\widehat{G}_{i + 2} - (W \cup \{u, v\})$. Thus, both $\widehat{G}'_{i + 1}$ and
    $\widehat{G}'_{i + 2}$ are cluster graphs and indeed all layers are.

    To see consistency, recall that, for each $j \in [i]$, there are
    no edges in $\widehat{G}'_j$ between $W$ and $\widehat{V} \setminus W$. For each
    $j \in [i - 1]$, $\widehat{G}'_j - W$ and $\widehat{G}'_{j + 1} - W$ are consistent by
    our choice of the editing and marking sets. In addition,
    $\widehat{G}'_j[W] = \widehat{G}'_{j + 1}[W]$ by construction and because no edits are
    made in these graphs, and thus $\widehat{G}'_j[W]$ is consistent
    with~$\widehat{G}'_{j + 1}[W]$. Thus, all pairs of consecutive
    layers~$j, j + 1$ up to layer~$j = i - 1$ are consistent.

    For the remaining layers~$i$, $i + 1$, and $i + 2$ we have
    $\widehat{G}'_i - (W \cup \{u, v\}) = \widehat{G}'_{i + 1} - (W \cup \{u, v\}) = \widehat{G}'_{i
      + 2} - (W \cup \{u, v\})$ and, thus, these graphs are
    consistent. % Let
    % $V_1 \coloneqq V \setminus (W \cup \{u, v\})$.
    Observe that there are no edges between $W \cup \{u, v\}$ and
    $\widehat{V} \setminus (W \cup \{u, v\})$ in $\widehat{G}'_i$, $\widehat{G}'_{i + 1}$, and
    $\widehat{G}'_{i + 2}$, since $\{u, v\}$ is a cluster in~$\widehat{G}'_i$. Thus,
    it is enough to show that, for both $j \in \{i, i + 1\}$,
    $\widehat{G}'_j[W \cup \{u, v\}]$ and $\widehat{G}'_{j + 1}[W \cup \{u, v\}]$ are
    consistent.
    % By choice of the editing and marking sets as solution for $I$,
    % for each $j \in [i + 1]$, $G'_j[V_1]$ and $G'_{j + 1}[V_1]$ are
    % consistent, and thus it is enough to show that
    % $G'_j[W \cup \{u, v\}]$ and $G'_{j + 1}[W \cup \{u, v\}]$ are
    % consistent. This is obvious for each $j \in [i - 1]$. 
    Consider the case that $j = i$. The edge set
    $E(\widehat{G}'_i[W \cup \{u, v\}]) \oplus E(\widehat{G}'_{i + 1}[W \cup \{u, v\}])$
    consists of $\{c, u\}, \{c, v\}, \{w, x\}, \{y, z\}$, the edges
    between $c$ and $C_c$, $x$ and $C_x$, and $z$ and $C_z$. Since
    $c, x, z \in \widehat{D}_i$, each of these edges has a marked endpoint and
    thus $\widehat{G}'_i[W]$ and $\widehat{G}'_{i + 1}[W]$ are consistent. Finally, in the
    case that $j = i + 1$, the edges in
    $E(\widehat{G}'_{i+1}[W \cup \{u, v\}]) \oplus E(\widehat{G}'_{i + 2}[W \cup \{u, v\}])$
    are $\{c, u\}, \{c, v\}, \{u, v\}$, the edges between $c$ and
    $C_c$, and the edges between $v$ and $C_v$. Since
    $c, v \in \widehat{D}_{i + 1}$, $\widehat{G}'_{i + 1}[W \cup \{u, v\}]$ and
    $\widehat{G}'_{i + 2}[W \cup \{u, v\}]$ are consistent. Thus, indeed, each
    two consecutive layers are consistent.

    \proofparagraph{(ii).} Let
    $((M_j)_{j \in
      [i]}, (D_j)_{j \in [i-1]})$ 
      be a solution to $I$ such that $\{u\}$ and $\{v\}$ are singleton
    clusters in $G_i \oplus M_i$. Modify this solution as follows to obtain a
    solution $((\widehat{M}_j)_{j \in
      [i+2]}, (\widehat{D}_j)_{j \in [i+1]})$
      for~$\widehat{I}$. See the lower pane in \cref{fig:tempw1hk} for an
    illustration of the modifications. 
    Initially, let $\widehat{M}_{j}=M_{j}$ for every $j \in [i]$ 
    and $\widehat{D}_{j}=D_{j}$ for every $j \in [i-1]$.
    Further, put $\widehat{D}_{i} = \{c, x, z\}$, $\widehat{D}_{i + 1} = \{c, v \}$. 
    For (a) let $\widehat{M}_{i + 1} = M_i \cup \{\{c, u\}, \{w, x\}\}$, 
    $\widehat{M}_{i + 2} = (M_i \setminus \{\{u,v\}\}) \cup \{\{y,z\}\}$ 
    and for (b) let $\widehat{M}_{i + 1} = M_i \cup \{\{c, u\}, \{y, z\}\}$, 
    $\widehat{M}_{i + 2} = (M_i \setminus \{\{u,v\}\}) \cup \{\{w,x\}\}$.
    This concludes the definition of the solution for $\widehat{I}$. Observe
    that all marking budgets are satisfied and, since we have replaced
    $\{u, v\}$ by some other edge in $\widehat{M}_{i + 2}$, also all editing
    budgets are satisfied. Furthermore, in solution~(a), $\{w, x\}$,
    $\{y\}$, and~$\{z\}$ form clusters in $\widehat{G}_{i+2} \oplus \widehat{M}_{i+2}$ and in
    solution~(b) $\{w\}$, $\{x\}$, and $\{y, z\}$ form clusters in $\widehat{G}_{i+2} \oplus \widehat{M}_{i+2}$, as required. 
    Moreover, by definition, the equality condition on
    $(\widehat{G}_{i+2} \oplus \widehat{M}_{i+2}) - (W \cup \{u, v\})$ and 
    $(G_i \oplus M_i) - \{u, v\}$ as well as the non-existence of edges between $\widehat{V} \setminus (W \cup \{u,v\})$ and $(W \cup \{u,v\})$ are satisfied. It remains to
    show that for every $j \in [i+2]$ the graph $\widehat{G}'_j=\widehat{G}_j \oplus M_j$ 
    is a cluster graph and each two consecutive $\widehat{G}'_j$'s are consistent.
    We only show this for solution (a) since
    the proof for solution (b) is analogous.
    
    For every $j \in [i]$ let $G'_j= G_j \oplus M_j$.
    We first show that each $\widehat{G}'_j$ is a cluster graph. By the same
    arguments as for statement~(i) it follows that this is the case for each $j \in [i]$.
    Consider layers $j \in \{i + 1, i + 2\}$.
    Observe that there are no edges in~$\widehat{G}'_j$ between
    $W \cup \{u, v\}$ and $\widehat{V} \setminus (W \cup \{u, v\})$.
    Furthermore, $\widehat{G}'_j - (W \cup \{u, v\})$ is a cluster graph,
    because $\widehat{G}_j - (W \cup \{u, v\}) = \widehat{G}_i - (W \cup \{u, v\})$ and
    $\widehat{M}_j$ restricted to pairs contained in $\widehat{V} \setminus (W \cup \{u, v\})$
    equals~$M_i$ (since $\{u\}$ and $\{v\}$ are singleton clusters in $G'_i$), and $M_i$ is a cluster editing set for $G_i$. 
    Thus it is enough to show that
    $\widehat{G}'_j[W \cup \{u, v\}]$ is a cluster graph. Since
    $\widehat{G}_j[W \cup \{u, v\}]$ is a cluster graph, we show that the edits
    do not destroy the cluster graph property. The only edits are contained in layers $i + 1$ and $i + 2$.

    Consider the case $j = i + 1$. Restricted to pairs in
    $W \cup \{u, v\}$, $\widehat{M}_{i + 1}$ contains only the pairs~$\{c, u\}$,
    $\{w, x\}$, and $\{u, v\}$. (Note that $\widehat{M}_{i + 1}$ inherits the
    pair $\{u, v\}$ from~$M_i$, because $\{u\}$ and $\{v\}$ are singleton
    clusters in~$\widehat{G}'_i$, and $\widehat{M}_{i + 1}$ inherits only this pair,
    because $M_i$ does not contain any pairs with a vertex in~$W$.)
    Since $\{w\}$ and $\{x\}$ are singleton clusters in
    $\widehat{G}_{i+1}[W \cup \{u, v\}]$, the edit $\{w, x\}$ does not destroy the
    cluster graph property of this graph. Since $\{c, u, v\}$ forms a
    cluster of size three in $\widehat{G}_{i+1}[W \cup \{u, v\}]$, also the edits
    $\{c, u\}$, $\{u, v\}$ do not destroy the cluster graph property.
    Thus, indeed, $\widehat{G}'_{i+1}[W \cup \{u, v\}]$ is a cluster graph.

    Consider the case $j = i + 2$. Since $\{u, v\}$ was removed
    from~$\widehat{M}_{i + 2}$, restricted to pairs in~$W \cup \{u, v\}$,
    $\widehat{M}_{i + 2}$ contains only $\{y, z\}$. Since $\{y, z\}$ forms a
    cluster of size two in~$\widehat{G}_{i+2}[W \cup \{u, v\}]$, edit $\{y, z\}$
    does not destroy the cluster graph property. Hence, for both
    $j \in \{i + 1, i + 2\}$ and indeed for all $j \in [i + 2]$ we
    have shown that $\widehat{G}'_j$ is a cluster graph.

    It remains to show that each two consecutive modified layers are
    consistent. By the same arguments as for statement~(i) we have
    that, for each $j \in [i - 1]$, $\widehat{G}'_j$ and $\widehat{G}'_{j + 1}$ are
    consistent. For the remaining layers~$i$, $i + 1$, and $i + 2$, we
    again have
    $\widehat{G}'_i - (W \cup \{u, v\}) = \widehat{G}'_{i + 1} - (W \cup \{u, v\}) = \widehat{G}'_{i
      + 2} - (W \cup \{u, v\})$ and, thus, these graphs are
    consistent. Since, for each $j \in \{i, i + 1, i + 2\}$, there are
    no edges between $\widehat{V} \setminus (W \cup \{u, v\})$ and
    $W \cup \{u, v\}$ in $\widehat{G}'_j$, it is enough to show for both
    $j \in \{i, i + 1\}$ that $\widehat{G}'_j[W \cup \{u, v\}]$ and
    $\widehat{G}'_{j + 1}[W \cup \{u, v\}]$ are consistent. The edge set
    $E(\widehat{G}'_i[W \cup \{u, v\}]) \oplus E(\widehat{G}'_{i + 1}[W \cup \{u, v\}])$
    consists of $\{c, v\}$ (note that both $\{u, v\} \in \widehat{M}_{i + 1}$
    and $\{c, u\} \in \widehat{M}_{i + 1}$), edge $\{w, x\}$, and the edges to the cliques. Since
    $c, x, z \in \widehat{D}_i$, graphs $\widehat{G}'_{i}$ and $\widehat{G}'_{i + 1}$ are
    consistent. The edge set
    $E(\widehat{G}'_{i+1}[W \cup \{u, v\}]) \oplus E(\widehat{G}'_{i + 2}[W \cup \{u, v\}])$
    contains only $\{c, v\}$ and the edges to the cliques and, since $c, v \in \widehat{D}_{i + 1}$, graphs
    $\widehat{G}'_{i + 1}$ and $\widehat{G}'_{i + 2}$ are consistent. Thus the consistency
    property holds.
 
    \proofparagraph{Helper statement.} Before proving (iii) we prove the following helper statement~(h). For each solution
    $((\widehat{M}_j)_{j \in
      [i+2]}, (\widehat{D}_j)_{j \in [i+1]})$
      to~$\widehat{I}$ it holds that
    $\widehat{D}_i = \{c, x, z\}$ and $\widehat{D}_{i + 1} = \{c, v\}$. For this, observe
    that, for each~$j \in [i + 2]$, there is no pair in $\widehat{M}_j$ which
    contains any vertex of~$C_c$, $C_v$, $C_x$, or $C_z$: Otherwise,
    because the minimum cut of any of these cliques contains more than
    $k$ edges, any edit could contain at most one vertex of such a
    clique. Since, however, there are at most $k$ edits incident with
    these cliques in a layer, at least one $P_3$ would remain. Thus,
    indeed, there are no edits incident with any of the cliques $C_c$,
    $C_v$, $C_x$, or $C_z$. For each $\alpha \in \{c, x, z\}$,
    $\alpha$ is adjacent to all vertices in $C_\alpha$ in layer~$i$
    but not adjacent to any vertex in $C_\alpha$ in layer~$i + 1$.
    Since not all vertices in $C_\alpha$ are marked and no edits are
    incident with~$C_\alpha$, because of consistency between
    layers~$i$ and~$i + 1$, we thus have $\alpha \in \widehat{D}_i$. By a
    similar argument, $c, v \in \widehat{D}_{i + 2}$: Observe that, for each
    $\alpha \in \{c, v\}$, $\alpha$ is not adjacent to any vertex in
    $C_\alpha$ in layer~$i + 1$ and adjacent to all vertices
    in~$C_\alpha$ in layer~$i + 2$. Hence, $c, v \in \widehat{D}_{i + 1}$. Now (h) follows since $\widehat{d}_{i} = 3$ and $\widehat{d}_{i + 1} = 2$.

    \proofparagraph{(iii).} Assume that $I$ is sane in layer~$i$ for
    $u$ and~$v$. Now take a solution
    $\widehat{S}=((\widehat{M}_j)_{j \in
      [i+2]}, (\widehat{D}_j)_{j \in [i+1]})$
      to~$\widehat{I}$ and let $\widehat{G}'_{j}=\widehat{G}_{j} 
      \oplus \widehat{M}_{j}$ for every $j \in [i+2]$. 
      We show that this solution fulfills all properties
    required for statement~(iii).

    Since an induced subgraph of a cluster graph is again a cluster
    graph, restricting
    $\widehat{S}$ to $\widehat{V} \setminus W$ and layers~$[i]$ yields a solution
    to~$I$. Since solutions to~$I$ saturate the editing budget in
    layer~$i$ and $\widehat{k}_i=k_i$, no pair in $\widehat{M}_i$ contains any vertex in~$W$.

    Furthermore, since $\widehat{D}_i \cap (\widehat{V} \setminus W) = \emptyset$
    by~(h), for consistency between layers~$i$ and $i + 1$, we have
    that $\widehat{M}_{i + 1} \cap \binom{\widehat{V}\setminus W}{2} = \widehat{M}_i$ and $\widehat{M}_i \subseteq \widehat{M}_{i + 1}$
    (recall that no pair in $\widehat{M}_i$ contains a vertex in~$W$). Moreover,
    since $E(\widehat{G}_{i + 1}) \oplus E(\widehat{G}_{i + 2})$ contains both $\{w, x\}$
    and $\{y, z\}$ and $\widehat{D}_{i + 1} \cap \{w, x, y, z\} = \emptyset$ by~(h), we have
    that $\widehat{M}_{i + 1} \oplus \widehat{M}_{i + 2}$ also contains both $\{w, x\}$
    and $\{y, z\}$.

    Having observed the above properties of the solution, we now
    distinguish whether $\{u, v\} \in \widehat{M}_i$.

    Suppose that $\{u, v\} \in \widehat{M}_i$. (Note that this includes the case
    where $\{u\}$ and $\{v\}$ are singleton clusters in $\widehat{G}'_i$.) Since
    $\widehat{M}_i \subseteq \widehat{M}_{i + 1}$, we have $\{u, v\} \in \widehat{M}_{i + 1}$. Since
    $\widehat{G}'_{i + 1}$ is a cluster graph, furthermore, either
    $\{c, v\} \in \widehat{M}_{i + 1}$ or $\{c, u\} \in \widehat{M}_{i + 1}$; otherwise,
    $u, c, v$ induce a $P_3$ in $\widehat{G}'_{i+1}$. Thus, since
    $|\widehat{M}_{i + 1}| \leq \widehat{k}_{i + 1} = \widehat{k}_{i} + 2$, there is at most one
    further vertex pair in~$\widehat{M}_{i + 1}$, that is $|\widehat{M}_{i + 1} \setminus (\widehat{M}_i \cup \{\{u, v\}, \{c, v\}, \{c, u\}\})| \leq 1$. Since $I$ is sane, and
    because $u$ and $v$ are incident only with each other in~$\widehat{G}_i$ by precondition of \selcons,
    $\{u, v\}$ is the only pair in~$\widehat{M}_i$ that contains~$u$ or~$v$.
    Thus, because $\widehat{D}_{i+1} \cap (\widehat{V} \setminus W) = \{v\}$ by~(h),
    by consistency between layers~$i$ and~$i + 1$ we have $\widehat{M}_{i} \setminus \{\{u, v\}\} \subseteq \widehat{M}_{i + 2}$.
    Therefore, since
    $|\widehat{M}_{i + 2}| \leq \widehat{k}_{i + 2} = \widehat{k}_{i}$, there is at most one
    further vertex pair in~$\widehat{M}_{i + 2}$. It follows that either $\{w, x\} \in \widehat{M}_{i+1}$
    and $\{y, z\} \in \widehat{M}_{i+2}$ or vice versa, that is, $\{w, x\} \in \widehat{M}_{i+2}$
    and $\{y, z\} \in \widehat{M}_{i+1}$. Hence, the solution 
    saturates the editing budget in layer~$i + 2$, we have $\widehat{G}'_{i+2} - (W \cup \{u, v\}) 
    = \widehat{G}'_i- \{u, v\}$, and there are no 
    edges between $V(G_i)\setminus \{u, v\}$ and $(W \cup \{u, v\})$ in $\widehat{G}'_{i+2}$.
    Moreover, either $\{w\}$ and $\{x\}$ form singleton clusters in~$\widehat{G}'_{i + 1}$ 
    and $\widehat{G}'_{i + 2}$  or $\{y\}$ and
    $\{z\}$ form singleton clusters in~$\widehat{G}'_{i + 1}$ and $\widehat{G}'_{i + 2}$. 
    That is, the last part of statement~(iii) holds and statements on the sanity of~$\widehat{I}$ also follow. 

    Finally, suppose $\{u, v\} \notin \widehat{M}_i$. Clearly, the last part
    of statement~(iii) holds. %It remains to show the statements on the sanity of~$\widehat{I}$. 
    As shown before, $\widehat{M}_i \subseteq \widehat{M}_{i + 1}$. Because
    $\{u, v\} \notin \widehat{M}_i$, by sanity of~$I$ in layer~$i$ for $u, v$,
    $\{u, v\}$ forms a cluster of size two. Thus, since $u$ and $v$ are
    only incident with each other in $\widehat{G}_i$ by precondition of \selcons, there are no pairs
    containing~$u$ or~$v$ in~$\widehat{M}_i$. Hence, $\widehat{D}_{i + 1}$ does not
    contain any vertices occurring in pairs of~$\widehat{M}_i$ by~(h). This implies
    $\widehat{M}_i \subseteq \widehat{M}_{i + 2}$. 
    As $|\widehat{M}_{i + 2}| \le \widehat{k}_{i + 2} =\widehat{k}_{i} =
    |\widehat{M}_{i}|$ by sanity of~$I$ in layer~$i$, we have $\widehat{M}_{i+2}= \widehat{M}_{i}$
    and the solution saturates the
    editing budget in layer~$i + 2$. 
    Thus, we have $\widehat{G}'_{i+2} - (W \cup \{u, v\}) 
    = \widehat{G}'_i- \{u, v\}$ and  in $\widehat{G}'_{i+2}$ there are no 
    edges between $V(G_i)\setminus \{u, v\}$ and $(W \cup \{u, v\})$ 
    and $\{w, x\}$ and $\{y, z\}$ form clusters.
%     As $\{w, x\}, \{y, z\} \notin \widehat{M}_{i + 2}$, we have 
%     $\{w, x\}, \{y, z\} \in \widehat{M}_{i + 1}$ and, as 
%     $|\widehat{M}_{i + 1}| \le \widehat{k}_{i + 1} =\widehat{k}_{i} +2  =
%     |\widehat{M}_{i}| + 2$, we have
%     $\widehat{M}_{i + 1} = \widehat{M}_{i + 1} \cup \{\{w, x\}, \{y, z\}\}$.
    Hence, $\widehat{I}$ is also sane in
    layer~$i + 2$ for $w, x$ and for~$y, z$.
  \end{proof}

  We now use function \selcons\ to construct the vertex-selection
  gadgets. To construct one vertex-selection gadget for color~$p$, we
  begin with an edge which has to be removed in any solution 
  (this can be enforced by consistency with previous layer, not containing the edge). This
  edge is the input edge of the first application of \selcons. Recall
  that \selcons\ has one input edge and two output edges and the
  removal of the input edge forces the removal of one of the output
  edges. One of the output edges corresponds to a vertex of color~$p$
  and the other will be the input edge to the next application of
  \selcons. In this application, in turn, one of the output edges
  corresponds to another vertex of color~$p$ and the other is the
  input edge to the next application and so forth. In the last
  application, both output edges correspond to distinct vertices. By
  the properties of \selcons, this construction will ensure that one
  of the output edges corresponding to vertices is removed, that is,
  one vertex is selected.

  Formally, we proceed as follows. Begin with an empty temporal graph
  (without vertices or edges) with only one layer, put $k_1 = 0$. For
  each $p \in [s]$ in order, proceed as follows to construct the
  \emph{vertex-selection gadget} for color~$p$. Introduce a new layer
  $2(p - 1)n + 2$ as an identical copy of $2(p - 1)n + 1$ (we will
  ensure that this layer exists), put
  $k_{2(p - 1)n + 2} = k_{2(p - 1)n + 1} + 1$, and
  $d_{2(p - 1)n + 1} = 0$. Introduce two new
  vertices $g_1, h_1$ and make them adjacent in layer~$2(p - 1)n + 2$
  and non-adjacent in all preceding layers.
  For each $i \in [n - 2]$ in order, apply
  $\selcons(2(p - 1)n + 2i, g_i, h_i, u_i^p, v_i^p, g_{i + 1}, h_{i +
    1})$. Finally, apply
  $\selcons(2(p - 1)n + 2n - 2, g_{n - 1}, h_{n - 1}, u_{n - 1}^p,
  v_{n - 1}^p, u_{n}^p, v_{n}^p)$ and introduce layer $2pn + 1$ as an
  identical copy of layer~$2pn$. Put $k_{2pn + 1} = k_{2pn}$, and
  $d_{2pn} = 0$. This concludes the construction of the
  vertex-selection gadgets. Note that, since the output edges in
  \selcons\ are isolated, the preconditions of \selcons\ are satisfied
  in each application. Furthermore, $k_{2sn + 1} = s$ and, for each
  layer~$i$ constructed so far, we have $k_i \leq s + 2$ and
  $d_i \leq 3$.

  \begin{lem}\label[lemma]{lem:wh-selection}
    Let $I=((G_i)_{i \in [2sn + 1]}, (k_i)_{i \in [2sn + 1]}, 
    (d_i)_{i \in [2sn]})$ be the instance of \TCEwIB\ 
    constructed above.
    \begin{lemenum}
    \item For each sequence $(i_p)_{p \in [s]}$ of integers in $[n]$,
      there is a solution $((M_i)_{i \in [2sn+1]}, (D_i)_{i \in [2sn]})$ 
      for~$I$ such that, for
      each $p \in [s]$, $\{u_{i_p}^p\}$ and $\{v_{i_p}^p\}$ form
      clusters in $G_{2sn + 1} \oplus M_{2sn + 1}$ and, for each $i \in [n] \setminus \{i_p\}$,
      $\{u_i^p, v_i^p\}$ forms a cluster in $G_{2sn + 1} \oplus M_{2sn + 1}$.
    \item In each solution $((M_i)_{i \in [2sn+1]}, (D_i)_{i \in [2sn]})$ 
      for~$I$, for each
      $p \in [s]$, there is at least one $i \in [n]$ such that
      $\{u_i^p\}$ and $\{v_i^p\}$ form singleton clusters in $G_{2sn + 1} \oplus M_{2sn + 1}$.
    \end{lemenum}
  \end{lem}
  \begin{proof}
    \proofparagraphm{(i).} We build the required solution successively by building a
    solution for the resulting instance after each application of
    \selcons. Before the first application, the only
    possible solution, and our initial solution, is to remove
    $\{g_1, h_1\}$ in layer~$2$. For each $p \in [s]$, proceed
    iteratively as follows. For each $i \in [i_p - 1]$, iteratively
    apply \cref{lem:selcons}~(ii), showing that there is a solution
    $((M_{\tilde{\imath}})_{\tilde{\imath} \in [2(p - 1)n + 2i + 2]}, (D_{\tilde{\imath}})_{\tilde{\imath} \in [2(p - 1)n + 2i + 1]})$
    such that, in $G_{2(p - 1)n + 2i + 2} \oplus M_{2(p - 1)n + 2i + 2}$, for each $j \in [i]$,
    $\{u_j^p, v_j^p\}$ forms a cluster and $\{g_{i+1}\}$ and $\{h_{i+1}\}$
    form singleton clusters. (Note that the equality property on the
    two layers and the non-existence of edges between the old and newly added vertices
    in \cref{lem:selcons}~(ii) ensures that clusters
    involving some $u_j^p, v_j^p$ remain the same when increasing $i$
    once they were formed.) If $i_p = n$, then indeed $\{u^p_n\}$ and $\{v^p_n\}$ form singleton clusters instead of $\{g_{i_p}\}$ and $\{h_{i_p}\}$, giving the required solution up to the last layer~$2pn + 1$ (for which we show below how to construct it) . Otherwise, if $i_p < n$, put $i = i_p$ and apply
    \cref{lem:selcons}~(ii), yielding that there is a solution
    $((M_{\tilde{\imath}})_{\tilde{\imath} \in [2(p - 1)n + 2i + 2]}, (D_{\tilde{\imath}})_{\tilde{\imath} \in [2(p - 1)n + 2i + 1]})$
    such that, in $G_{2(p - 1)n + 2i + 2} \oplus M_{2(p - 1)n + 2i + 2}$, for each $j \in [i - 1]$,
    $\{u_j^p, v_j^p\}$ forms a cluster, and $\{u_i^p\}$, $\{v_i^p\}$,
    and $\{g_{i+1}, h_{i+1}\}$ form clusters. If $i_p = n - 1$ then indeed $\{u^p_n, v^p_n\}$ forms a cluster, again giving the required solution already for layers up to~$2pn$. Otherwise, if $i_p < n - 1$, for each
    $i \in [n - 1] \setminus [i_p]$, iteratively apply
    \cref{lem:selcons}~(i), yielding that there is a solution
    $((M_{\tilde{\imath}})_{\tilde{\imath} \in [2pn]}, (D_{\tilde{\imath}})_{\tilde{\imath} \in [2pn - 1]})$
    such that, in $G_{2pn} \oplus M_{2pn}$, for each
    $j \in [i_p - 1]$, $\{u_j^p, v_j^p\}$ forms a cluster,
    $\{u_{i_p}^p\}$, $\{v_{i_p}^p\}$ form clusters, and, for each
    $j \in [n] \setminus [i_p]$, $\{u_j^p, v_j^p\}$ forms a cluster.
    Finally we let $D_{2np} = \emptyset$ and $M_{2np+1} = M_{2pn}$ to 
    obtain a solution for the instance before the first application of \selcons{}
    for the next color~$p + 1$.
    Hence, these same clusters occur in~$G_{2pn+1} \oplus M_{2pn+1}$, %for any solution $M_{2pn+1}$, 
    meaning that they carry over to the iteratively
    constructed solution for the next color~$p + 1$. Hence, after
    constructing the solution iteratively for all colors~$p$, we
    obtain the claimed solution for~$I$.
    
    \proofparagraph{(ii).} We first show that after each application of
    \selcons\ for some $p \in [s]$ and $i \in [n-1]$ the resulting
    instance~$I_{p, i}$ has the property
    \begin{compactenum}[(a)]
    \item of being sane in layer
      $2(p - 1)n + 2i + 2$ (the last layer constructed so far) for
      $g_{i + 1}$, $h_{i + 1}$ when $i \in [n-2]$ and for $u_n^p$, $u_n^p$
      when $i = n-1$, 
    \item that, for each solution 
      $((M_{\tilde{\imath}})_{\tilde{\imath} \in [2(p-1)n+2i+2]}, (D_{\tilde{\imath}})_{\tilde{\imath} \in [2(p-1)n+2i+1]})$,
      for each~$r \in [p - 1]$, there is
      an $i_r \in [n]$ such that $\{u_{i_r}^r\}$ and $\{v_{i_r}^r\}$
      form singleton clusters in~$G_{2(p - 1)n + 2i + 2}\oplus M_{2(p - 1)n + 2i + 2}$ and that
    \item either $\{g_{i + 1}\}$ and
      $\{h_{i + 1}\}$ form singleton clusters or there is a
      $i_p \in [i]$ such that $\{u_{i_p}^p\}$ and $\{v_{i_p}^p\}$ form
      singleton clusters.
  \end{compactenum}
  The proof is by induction on the number of
    applications of \selcons.

    For the first application of \selcons, the statement follows from
    the fact that $\{g_1, h_1\}$ is removed in the only possible
    solution (due to consistency between layer one and two) in
    combination with \cref{lem:selcons}~(iii). (Note that the instance
    before applying \selcons\ is sane in layer~$2$ for $g_1$, $h_1$.)

    Now assume that the statement holds for some number (at least one)
    of applications of \selcons\ and we prove that it holds after one
    more application. Let $I_{p, i}$ be the resulting instance. We
    distinguish whether $i > 1$ or not.

    Suppose $i = 1$. Since the current application of \selcons\ is not
    the first one, we have $p > 1$. In the last application of
    \selcons, the vertex-selection gadget for color~$p - 1$ was
    completed. Hence, by the induction hypothesis, and since there are
    no $g_{n}, h_{n}$, we have that for each solution 
    $((M'_{\tilde{\imath}})_{\tilde{\imath} \in [2(p-1)n]}, (D'_{\tilde{\imath}})_{\tilde{\imath} \in [2(p-1)n-1]})$
    for instance $I_{p - 1, n-1}$, for each
    $r \in [p - 1]$, there is an $i_r \in [n]$ such that
    $\{u_{i_r}^r\}$ and $\{v_{i_r}^r\}$ form singleton clusters in $G_{2(p - 1)n} \oplus M'_{2(p - 1)n}$.

    Fix an arbitrary solution~$S=((M_{\tilde{\imath}})_{\tilde{\imath} \in [2(p-1)n+4]}, (D_{\tilde{\imath}})_{\tilde{\imath} \in [2(p-1)n+3]})$ for
    $I_{p, 1}$. By
    the induction hypothesis, $I_{p - 1, n}$ is sane in layer
    $2(p - 1)n$ for $u_n^{p - 1}, v_n^{p - 1}$, that is, the editing
    budget is saturated in that layer by every solution. Since
    $S$ induces a solution for
    $I_{p - 1, n-1}$ and \selcons\ did not change the editing budgets
    for preexisting layers, $|M_{2(p - 1)n}| = k_{2(p - 1)n}$. Since
    $k_{2(p - 1)n + 1} = k_{2(p - 1)n}$, $d_{2(p - 1)n}=0$, and layer $2(p - 1)n + 1$ was
    taken to be an identical copy of layer~$2(p - 1)n$, we have
    $|M_{2(p - 1)n + 1}| = k_{2(p - 1)n + 1}$ as well. Since the only
    difference between layer $2(p - 1)n + 1$ and $2(p - 1)n + 2$
    before the application of \selcons\ is the introduction of the
    edge $\{g_1, h_1\}$ and an incremented editing budget, by
    consistency between layers $2(p - 1)n + 1$ and $2(p - 1)n + 2$ and
    because $d_{2(p - 1)n + 1} = 0$, we have
    $M_{2(p - 1)n + 2} = M_{2(p - 1)n + 1} \cup \{\{g_1, h_1\}\}$. Thus,
    the instance directly before applying \selcons\ is sane in layer
    $2(p - 1)n + 2$ for $g_1, h_1$. Thus, by \cref{lem:selcons}~(iii),
    property~(a) holds. 
    
    As each solution for $I_{p, 1}$ implies one
    for~$I_{p - 1, n-1}$, singleton clusters $\{u_{i_r}^r\}$
    and $\{v_{i_r}^r\}$ as above are present also in $G_{2(p - 1)n} \oplus M_{2(p - 1)n}$
    restricted to the parts present in the 
    instance $I_{p, 1}$. Since $d_{2(p - 1)n} = 0$ and
    $G_{2(p - 1)n} = G_{2(p - 1)n + 1}$, these singleton clusters
    occur also in $G_{2(p - 1)n + 1} \oplus M_{2(p - 1)n+1}$. Similarly, since
    $d_{2(p - 1)n + 1} = 0$, they also occur in $G_{2(p - 1)n + 2}\oplus M_{2(p - 1)n+2}$.
    Thus by \cref{lem:selcons}~(iii), $\{u_{i_r}^r\}$ and $\{v_{i_r}^r\}$ also
    form singleton clusters in $G_{2(p - 1)n + 4}\oplus M_{2(p - 1)n+4}$. This implies
    property~(b).

    Since $g_1, h_1$ form singleton clusters in
    any solution, by \cref{lem:selcons}~(iii), also the property~(c) holds for $i=1$.

    Suppose $i > 1$. Note that the instance directly before applying
    \selcons\ is $I_{p, i - 1}$. From property~(a) in the induction
    hypothesis, by \cref{lem:selcons}~(iii), property~(a) follows.
    Furthermore, from property~(b) in the induction hypothesis and
    since any solution  $((M_{\tilde{\imath}})_{\tilde{\imath} \in [2(p-1)n+2i+2]}, (D_{\tilde{\imath}})_{\tilde{\imath} \in [2(p-1)n+2i+1]})$ 
    for $I_{p, i}$ induces one for~$I_{p, i - 1}$
    as before, we know that for each~$r \in [p - 1]$, there is an
    $i_r \in [n]$ such that $\{u_{i_r}^r\}$ and $\{v_{i_r}^r\}$ form
    singleton clusters in~$G_{2(p - 1)n + 2i} \oplus M_{2(p - 1)n + 2i}$ in
    instance~$I_{p, i}$. Hence, again by \cref{lem:selcons}~(iii) $\{u_{i_r}^r\}$ and
    $\{v_{i_r}^r\}$ form singleton clusters
    in~$G_{2(p - 1)n + 2i + 2} \oplus M_{2(p - 1)n + 2i + 2}$. 
    Thus, property~(b) holds. As
    mentioned, $I_{p, i}$ is sane in layer $2(p - 1)n + 2i$
    for~$g_{i}, h_{i}$. If $\{g_{i}, h_{i}\}$ forms a cluster
    in~$G_{2(p - 1)n + 2i} \oplus M_{2(p - 1)n + 2i}$, 
    then, by property~(c) of the induction
    hypothesis, there is an $i_p \in [i - 1]$ such that
    $\{u_{i_p}^p\}$ and $\{v_{i_p}^p\}$ form clusters
    in~$G_{2(p - 1)n + 2i} \oplus M_{2(p - 1)n + 2i}$. By the same arguments as before, they
    also form clusters in $G_{2(p - 1)n + 2i+2} \oplus M_{2(p - 1)n + 2i+2}$, giving
    property~(c). If $\{g_{i}, h_{i}\}$ does not form a cluster in
    $G_{2(p - 1)n + 2i} \oplus M_{2(p - 1)n + 2i}$, then by sanity, $\{g_i\}$ and $\{h_i\}$ form
    clusters. By \cref{lem:selcons}~(iii) we thus obtain property~(c).
    
    Hence, applying the claim for $p=s$ and $i=n$, and since there are
    no $g_{n}, h_{n}$, we have that for each solution 
    $((M'_{\tilde{\imath}})_{\tilde{\imath} \in [2sn]}, (D'_{\tilde{\imath}})_{\tilde{\imath} \in [2sn-1]})$
    for instance $I_{s, n-1}$, in $G_{2sn} \oplus M'_{2sn}$, for each
    $r \in [s]$, there is an $i_r \in [n]$ such that
    $\{u_{i_r}^r\}$ and $\{v_{i_r}^r\}$ form singleton clusters.
    As each solution $((M'_{\tilde{\imath}})_{\tilde{\imath} \in [2sn+1]}, 
    (D'_{\tilde{\imath}})_{\tilde{\imath} \in [2sn]})$ for the resulting instance implies one
    for~$I_{s, n-1}$, the singleton clusters $\{u_{i_r}^r\}$
    and $\{v_{i_r}^r\}$ are present also in $G_{2sn} \oplus M_{2sn}$
    restricted to the parts present in the 
    instance $I_{s, n-1}$. Since $d_{2(p - 1)n} = 0$ and
    $G_{2sn} = G_{2sn + 1}$, these singleton clusters
    occur also in $G_{2sn + 1} \oplus M_{2sn+1}$, finishing the proof.    
  \end{proof}

  \cref{lem:wh-selection} implies that each solution of the instance
  constructed so far corresponds to some selection of one vertex in
  the graph~$G$ of the \MIS\ instance for each color in~$[s]$. 
  % Herein, we assume that $u$
 
  \proofparagraph{Verification.} To ensure that the selected vertices
  correspond to an independent set, we construct verification gadgets
  in the layers $2sn + 1, \ldots, 2sn + 3m$ as follows, where $m$
  is the number of edges in~$G$. Introduce layers
  $G_{2sn + 2}, \ldots, G_{2sn + 3m}$ as identical copies of
  $G_{2sn + 1}$. Enumerate the edges in $E(G)$ in an arbitrary order
  as $e_1, \ldots, e_m$. For each $j \in [m]$, let
  $e_j = \{a_{\alpha}^p, a_{\beta}^q\}$ where $a_{\alpha}^p \in V_p$
  and $a_{\beta}^q \in V_q$. Introduce the
  vertex~$w^{p,q}_{\alpha, \beta} = w^{q, p}_{\beta, \alpha}$ into all
  layers, and introduce into layer $2sn + 3j - 1$ the edges
  $\{w^{p, q}_{\alpha, \beta}, u^p_\alpha\}$,
  $\{w^{p, q}_{\alpha, \beta}, v^p_\alpha\}$,
  $\{w^{p, q}_{\alpha, \beta}, u^q_\beta\}$,
  and~$\{w^{p, q}_{\alpha, \beta}, v^q_\beta\}$. Furthermore,
  introduce a clique $C^{p, q}_{\alpha, \beta}$ with $2k + d + 1$ new
  vertices into all layers, and make $w^{p, q}_{\alpha, \beta}$
  adjacent to all vertices of $C^{p, q}_{\alpha, \beta}$ in each layer
  except for $2sn + 3j - 1$. Put
  $k_{2sn + 3j - 2} = s$, $k_{2sn + 3j - 1} = s + 2$,
  $k_{2sn + 3j} = s$, $d_{2sn + 3j-3} = 0$,
  $d_{2sn + 3j - 2} = 1$, and $d_{2sn + 3j - 1} = 1$ 
  (note that $k_{2sn+1}$ was already set to $s$ and $d_{2sn}$ to $0$ before). 
  This completes the construction of the
  verification gadgets.

  Intuitively, the budget of~$s$ edge edits corresponds to the
  vertices~$i_p$ chosen into the independent set, whose
  edges~$\{u^p_{i_p}, v^p_{i_p}\}$ we delete. Since we have only two
  additional edits in layer $2sn + 3j - 1$ and
  $w^{p, q}_{\alpha, \beta}$ is the center of a star with four leaves,
  we have to choose whether $w^{p, q}_{\alpha, \beta}$ forms a cluster
  with $u^p_{\alpha}, v^p_\alpha$ or with $u^q_\beta, v^q_\beta$. That
  is, one of these pairs must already be a cluster, meaning that the
  corresponding edge must not be deleted. Hence the corresponding
  vertex has not been chosen into the independent set.

  \begin{lem}\label[lemma]{lem:veri-cons}
    If there is a sequence~$(i_p)_{p \in [s]}$ of integers in $[n]$
    such that the corresponding vertices~$a_{i_p}^p \in V_p$ induce an
    independent set in~$G$, then there is a solution to the
    above-constructed instance of \TCEwIB.
  \end{lem}
  \begin{proof}
    Let $I'$ be the instance before constructing the verification
    gadgets and $I$ the instance afterwards. From
    \cref{lem:wh-selection}~(i) we know that there is a
    solution~$((M'_i)_{i \in [2sn+1]}, (D'_i)_{i \in [2sn]})$ for $I'$ such that, for each $p \in [s]$, $\{u_{i_p}^p\}$ and
    $\{v_{i_p}^p\}$ form clusters in
    $G_{2sn + 1}\oplus M'_{2sn+1}$ and, for each
    $i \in [n] \setminus \{i_p\}$, $\{u_i^p, v_i^p\}$ forms a cluster in
    $G_{2sn + 1}\oplus M'_{2sn+1}$.
    We modify this solution to obtain a solution 
    $((M_i)_{i \in [2sn+3m]}, (D_i)_{i \in [2sn+3m-1]})$ for~$I$ as follows.

    Initially, for each $i \in [2sn+1]$ put $M_i=M'_i$, 
    for each $i \in \{2sn + 2, \ldots, 2sn + 3m\}$ put $M_i = M_{2sn + 1}$, and 
    for each $i \in [2sn]$ put $D_i=D'_i$.
    Note that $d_{2sn}=0$ and hence $D_{2sn}=D'_{2sn}=\emptyset$.
    Observe that, for each
    $j \in [m]$, letting $e_j = \{a_{\alpha}^p, a_{\beta}^q\}$, either
    $\{u^p_{\alpha}, v^p_{\alpha}\} \notin M_i$ or
    $\{u^p_{\beta}, v^p_{\beta}\} \notin M_i$.

    For each~$j \in [m]$, let $i = 2sn + 3j - 1$, and, if
    $\{u^p_{\alpha}, v^p_{\alpha}\} \notin M_i$, i.e., 
    $\{u^p_{\alpha}, v^p_{\alpha}\}$ is an edge of $G_i \oplus M_i$, then put
    $\{w^{p, q}_{\alpha, \beta}, u^q_\beta\}, \{w^{p, q}_{\alpha,
      \beta}, v^q_\beta\} \in M_{i}$ and, otherwise, if
    $\{u^q_{\beta}, v^q_{\beta}\} \notin M_i$, put
    $\{w^{p, q}_{\alpha, \beta}, u^p_\alpha\}, \{w^{p, q}_{\alpha,
      \beta}, v^p_\alpha\} \in M_{i}$.

    For each~$j \in [m]$, put $D_{2sn + 3j-3}= \emptyset$ and
    $D_{2sn + 3j - 2}= D_{2sn + 3j - 1}= \{w^{p, q}_{\alpha, \beta}\}$. Clearly,
    the budget constraints are satisfied, each layer induces a
    cluster graph, and consecutive layers are consistent.
  \end{proof}

  Finally, we show that the verification gadgets work as they should.

  \begin{lem}\label[claim]{cl:veri-force}
    If the above-constructed instance of \TCEwIB\ admits a solution, then
    there is a sequence~$(i_p)_{p \in [s]}$ of integers in~$[n]$ such
    that the corresponding vertices $a_{i_p}^p \in V_p$ induce an
    independent set in~$G$.
    % If there are at least $s$ edges removed from
    % $G_{s\ell_1 + 3j - 3}$, then, for each $i \in \{0, 1, 2\}$, they
    % are also removed in~$G_{s\ell_1 + 3j - i}$. Furthermore, it is
    % only possible to satisfy the individual budgets for layers
    % $s\ell_1 + 3j - i$ if the removed edges do not include both
    % $\{u^p_a, v^p_a\}$ and $\{u^q_b, v^q_b\}$ for some vertices
    % $a \in V_p, b \in V_q$ which are adjacent in~$G$.
  \end{lem}
  \begin{proof}
    Let $I'$ be the instance of \TCEwIB\ before
    constructing the verification gadgets and $I$ the instance
    afterwards. By \cref{lem:wh-selection}~(ii), for each solution 
    $((M'_i)_{i \in [2sn+1]}, (D'_i)_{i \in [2sn]})$ for
    $I'$, there is a sequence~$(i_p)_{p \in [s]}$ of integers in~$[n]$
    such that $\{u_{i_p}^p\}$ and $\{v_{i_p}^p\}$ form singleton
    clusters in~$G_{2sn + 1} \oplus M'_{2sn+1}$. Since each solution for~$I$, when
    restricted to layers up to~$2sn + 1$, induces a solution for~$I'$,
    the same holds true for solutions for~$I$.

    Let $((M_i)_{i \in [2sn+3m]}, (D_i)_{i \in [2sn+3m-1]})$ be a solution to $I$.
    Observe that, for each $j \in [m]$, we have $D_{2sn+3j-3}=\emptyset$ as 
    $d_{2sn+3j-3}=0$ and, due to the connections between
    $w^{p, q}_{\alpha, \beta}$ and $C^{p, q}_{\alpha, \beta}$, we have
    $w^{p, q}_{\alpha, \beta} \in D_{2sn + 3j - 2}$ and
    $w^{p, q}_{\alpha, \beta} \in D_{2sn + 3j - 1}$. 
    Since $d_{2sn+3j-2}=d_{2sn+3j-1}=1$, we have $D_{2sn+3j-2}=D_{2sn+3j-1}= \{w^{p, q}_{\alpha, \beta}\}$.
    Hence, due to
    consistency, we have, for each
    $i \in \{2sn + 2, \ldots, 2sn + 3m\}$,
    $\{u_{i_p}^p, v_{i_p}^p\} \in M_i$.

    Now suppose that, for some $i_p$ and $i_r$, $p, r \in [s]$, we
    have that $a_{i_p}^p$ and $a_{i_r}^r$ are adjacent in~$G$, say via
    edge~$e_j$. Recall that $G_{2sn + 3j - 1} \oplus M_{2sn + 3j - 1}$ is a cluster
    graph and at least $s$ modifications in~$M_{2sn + 3j - 1}$
    are nonadjacent to $w^{p, r}_{i_p, i_r}$. 
    As there are at most two more edits allowed, at least two edges incident on $w^{p, r}_{i_p, i_r}$ remain,
    which form a $P_3$ as no two neighbors of $w^{p, r}_{i_p, i_r}$ are adjacent 
    due to the deletions $\{u_{i_p}^p, v_{i_p}^p\}, \{u_{i_r}^r, v_{i_r}^r\} \in M_i$.
    This is a contradiction, and hence, no two vertices~$a_{i_p}^p$,
    $a_{i_r}^r$, $p, r \in[s]$ are adjacent.
  \end{proof}
  
  \proofparagraph{Conclusion.} The above shows that the \MIS\ instance
  is a yes-instance if and only if the instance of \TCEwIB\ is a yes-instance.

  It remains to show how to realize the layer-individual budgets with
  global maximum budgets. Assume that we have a global maximum budget
  of~$k$ and $d$, respectively. To reduce the number of available
  marked vertices in layer~$j$, $j < 2sn +3m$, by one, introduce two cliques~$C_1, C_2$, each
  containing $2k + d + 1$ new vertices, and a vertex~$v$ into each
  layer. Make $v$ adjacent to all vertices in~$C_1$ in each layer
  $i \leq j$ and make $v$ adjacent to all vertices in~$C_2$ in each
  layer $i > j$.

  To reduce the edit budget in layer $i$ by one, introduce a clique
  containing $2k + d + 1$ new vertices into all layers and remove an
  edge from this clique in layer~$i$.

  Finally, observe that $d \leq 3$ and $k = s + 2$ in the resulting
  instance of \TCE. Thus, \TCE\ is \W{1}-hard with respect to~$k$ even
  when $d \leq 3$. Furthermore, the number of vertices is at most
  $(2k + d + 1)(4sn + k + m + (d+k)(2sn+3m)) \leq f(s)(n + m)$ and
  $\ell \leq 2sn + 3m$. Hence, as \MIS\ does not admit an
  $f(s)n^{o(s)}$-time algorithm unless the Exponential Time Hypothesis
  fails~\cite[Corollary 14.23]{CyganFKLMPPS15}, it follows that \TCE\
  does not admit an $f(k)(n\ell)^{o(k)}$-time algorithm unless the
  Exponential Time Hypothesis fails.
\end{proof}

\section{Kernelization for \MLCE{} and \TCE{}}
\label{sec:kernel}
\appendixsection{sec:kernel}

In this section we investigate the kernelizability of \MLCE{} and \TCE{} for different combinations of the four parameters as introduced in \cref{sec:intro}.
More specifically, we identify the parameter combinations for which \MLCE{} and \TCE{} admit polynomial kernels, and then we identify the parameter combinations for which no polynomial kernels exist, unless \NoKernelAssume.
% \subsection{A Polynomial Kernel for \MLCE{}}

We start with presenting a polynomial kernel for \MLCE{} for the parameter combination $(k, d, \ell)$ and then argue that essentially the same reduction rules also produce a polynomial kernel for \TCE{}. % Formally, we prove the following theorem.
\begin{theorem}
\label{thm:pk}
 % \MLCE{} and \TCE{} admit polynomial kernels with respect to the parameter combination $(k, d, \ell)$. % In particular, 
 \MLCE{} admits a kernel of size $O(\ell^3\cdot(k+d)^4)$ and \TCE{} admits a kernel of size $O(\ell^3\cdot(k+d\cdot\ell)^4)$. Both kernels can be computed in $O(\ell\cdot n^3)$ time.
\end{theorem}

%The rest of this section is devoted to the proof of this theorem.

\looseness=-1 We provide several reduction rules that subsequently modify the instance and we assume that if a particular rule is to be applied, then the instance is reduced with respect to all previous rules, that is, all previous rules were already exhaustively applied. 
% For each rule we immediately prove its correctness, that is, the produced instance is a yes-instance if and only if the original instance is. However, we leave the analysis of the running time of testing whether particular reduction rule applies and of applying the rule until all the rules are presented.
%
% To keep track of the budget in the individual layers we introduce the following intermediate problem.
% \newcommand{\MLCES}{\textsc{Multi-Layer Cluster Editing with Separate Budgets}}
% \decprob{\MLCES}{$\ell$ graphs $G_1=(V,E_1), \ldots, G_\ell=(V,E_\ell)$ and $\ell+1$ integers $k_1, \ldots k_\ell, d$.}{Is there a vertex subset $D\subseteq V$ with $|D|\le d$ and $\ell$ edge modification sets $M_1, M_2, \ldots, M_\ell\subseteq \binom{V}{2}$ such that \begin{compactenum}
%   \item for each $1\le i \le \ell$ we have that $|M_i|\le k_i$,
%   \item for each $1\le i \le \ell$ the graph $G_i'=(V, E_i\oplus M_i)$ is a cluster graph,
%   \item and for all $1\le i, j \le \ell$ we have that $G_i'[V\setminus D] = G_j'[V\setminus D]$?
% \end{compactenum}}
\newcommand{\MLCESlong}{\textsc{Multi-Layer Cluster Editing with Separate Budgets}}%
\newcommand{\MLCES}{\textsc{MLCEwSB}}%
To keep track of the budget in the individual layers we
introduce \MLCESlong\ (\MLCES) which differs from \MLCE\ only in that, instead of a global upper bound~$k$ on the number of edits, we receive $\ell$ individual budgets~$k_i$, $i \in [\ell]$, and we require that $|M_i| \leq k_i$.

We first transform the input instance of \MLCE{} to an equivalent instance of \MLCES{} by letting~$k_i=k$ for every $i \in [\ell]$. Then we apply all our reduction rules to \MLCES{}. Finally, we show how to transform the resulting instance of \MLCES{} to an equivalent instance of \MLCE{} with just a small increase of the vertex set. 

Through the presentation, let $(G_1=(V,E_1), \ldots, G_\ell = (V, E_\ell),k_1,\ldots,k_\ell,d)$ be the current instance. We let $k=\max\{k_i\mid i \in [\ell]\}$.

% First, we apply slightly modified versions of well known rules for classical \textsc{Cluster Editing}~\cite{GrammGHN05} and apply them on each layer individually. These rules are known to produce a kernel of size $k^2+2k$. 
% %For details on these reduction rules, we refer to the appendix (\appref{sec:easyrules}). 
% 
% Notably, we leave out a rule that removes isolated cliques. Hence, after the application of these rules we either conclude that we face a no-instance or every layer $i$ consists of a set~$R_i \subseteq V$, that contains the vertices $v$ that appear in some induced $P_3$ in $G_i$, and a number of isolated cliques. Furthermore, let $R=\bigcup_{i=1}^\ell R_i$. 
% %
% \toappendix{
% \subsection{Modified Versions of well-known \CE\ Reduction Rules}\label{sec:easyrules}
The following rules represent well known rules for classical \textsc{Cluster Editing}~\cite{GrammGHN05} applied to the individual layers of the multi-layer graph. The first rule formalizes the obvious constraint on the solvability of the instance. We omit a proof of correctness for this rule.

\begin{rrule}\label{rul:negat}
 If there is a layer $i \in [\ell]$ such that $k_i<0$, then answer NO.
\end{rrule}

% \begin{obs}
%  \autoref{rul:negat} is correct.
% \end{obs}

% \begin{rrule}[\appref{proof:rul:heavy_edge}]\label{rul:heavy_edge}
\begin{rrule}\label{rul:heavy_edge}
 If there is a layer $i \in [\ell]$ and an edge $\{u,v\}\in E_i$ in layer~$i$
 such that $G_i$ contains at least $k_i+1$ different induced $P_3$s 
 each of which contains the edge~$\{u,v\}$,
 then remove $\{u,v\}$ from $E_i$ and decrease $k_i$ by one.
 % If there is a layer $i \in \{1, \ldots, \ell\}$ and an edge $\{u,v\} \in E_i$ such that the vertex pair $u,v$ appears in at least $k_i+1$ different induced subgraphs of $G_i$ isomorphic to~$P_3$, then remove $\{u,v\}$ from $E_i$ and decrease $k_i$ by one.
\end{rrule}
% \appendixcorrectnessproof{rul:heavy_edge}{
\begin{lem}\label{lem:heavy_edge}
 \autoref{rul:heavy_edge} is correct.
\end{lem}
\begin{proof}
Let $I=(G_1, \ldots, G_\ell,k_1,\ldots, k_\ell,d)$ be the original instance and $\widehat{I}=(G_1, \ldots, \widehat{G_i}, \ldots G_\ell,$ $k_1,\ldots$ %hyphen
$\ldots, \widehat{k_i},\ldots, k_\ell,d)$ be the instance after the application of the rule, where $G_i=(V,E_i)$, $\widehat{G_i}=(V,\widehat{E_i})$, $\widehat{E_i}=E_i \setminus \{\{u,v\}\}$ and $\widehat{k_i}=k_i-1$ . 
If $\widehat{I}$ is a yes-instance, 
then $I$ is also a yes-instance with the same solution as the one for $\widehat{I}$ and the pair $\{u,v\}$ added.
 
For the converse, assume that $S=(D, M_1, \ldots M_\ell)$ is a solution for $I$ and let $G'_i=(V,E_i\oplus M_i)$. 
We claim that there is also a solution $S'$ for $\widehat{I}$.
Since the input multi-layer graphs in $I$ and $\widehat{I}$ only differ by one edge~$\{u,v\}$,
suppose towards a contradiction that $G'_i$ still contains~$\{u,v\}$, meaning that $\{u,v\}\notin M_i$. 
By the assumptions of the rule we know that there are $k_i+1$ vertices~$w_1, \ldots, w_{k_i+1}$ such that 
for each $i\in[k_i+1]$ the induced subgraph~$G_i[\{u,v,w_j\}]$ is a $P_3$,
which has to be destroy to obtain a cluster graph.
Since $\{u,v\}\in E_i \oplus M_i$,
in order to destroy all $P_3$s, for each $j\in[k_i+1]$ 
we have to either add the absent edge to or delete an existing edge~$e$ (with $e\neq \{u,v\}$) from the induced subgraph~$G[\{u,v,w_j\}]$.
However, since for two different indices~$j_1,j_2\in [k_i+1]$ the pair~$\{u,v\}$ is the only pair of vertices shared between $\{u,v,w_{j_1}\}$ and $\{u,v,w_{j_2}\}$,
we have to modify at least $k_i+1$ edges, a contradiction to $|M_i|\le k$,
%$\{u,v\} \in E_i\oplus M_i$, and $|M_i|\le k_i$, it follows that there is some $j_0 \in \{1,\ldots,k_i+1\}$ such that $G'_i[\{u,v,w_{j_0}\}]$ is isomorphic to~$P_3$, a contradiction.
Hence $G'_i$ does not contain $\{u,v\}$, i.e., $\{u,v\} \in M_i$ and $S'$ obtained from $S$ by replacing $M_i$ with $M_i \setminus \{\{u,v\}\}$ is a solution to~$\widehat{I}$.
\end{proof}
% }
% \begin{rrule}[\appref{proof:rul:heavy_non}]\label{rul:heavy_non}
\begin{rrule}\label{rul:heavy_non}
 If there is a layer $i \in [\ell]$ and a pair~$\{u,v\} \in V$ of vertices
 with $\{u,v\} \notin E_i$ (a non-edge) in layer~$i$ 
 such that $G_i$ contains at least $k_i+1$ different induced $P_3$s 
 each of which involves both $u$ and $v$,
 then add $\{u,v\}$ to $E_i$ and decrease $k_i$ by one.
\end{rrule}
% \appendixcorrectnessproof{rul:heavy_non}{
\begin{lem}
\label{lem:heavy_non}
 \autoref{rul:heavy_non} is correct.
\end{lem}
\begin{proof}
The proof is almost the same as for \autoref{lem:heavy_edge}, the obvious difference is that we assume $\widehat{E_i}=E_i \cup \{\{u,v\}\}$. 
Also in the second implication, supposing that $M_i$ does not contain $\{u,v\}$ leads to a contradiction.
% Let $(G_1, \ldots, G_\ell = (V, E_1), \ldots, (V, E_\ell),k_1,\ldots,k_i, \ldots, k_\ell,d)$ be the original instance and $(G_1, \ldots, \widehat{G_i}, \ldots G_\ell = (V, E_1), \ldots, (V,\widehat{E_i}),\ldots, (V, E_\ell),k_1,\ldots,\widehat{k_i},\ldots, k_\ell,d)$, where $\widehat{E_i}=E_i \cup \{\{u,v\}\}$ and $\widehat{k_i}=k_i-1$ be the instance after the application of the rule. If $(G_1, \ldots, \widehat{G_i}, \ldots G_\ell, k_1,\ldots,\widehat{k_i},\ldots, k_\ell,d)$ is a yes-instance, then also $(G_1, \ldots, G_\ell,k_1, \ldots, k_\ell,d)$ is a yes-instance.
%  
% For the converse, assume that $D, E_1, \ldots E_\ell$ is a solution to $(G_1, \ldots, G_\ell,k_1,\ldots, k_\ell,d)$ and let $G'_i=(V,E_i')$. 
% Assume for contradiction that $E_i$ does not contain $\{u,v\}$. By the assumptions of the rule we know that there are $w_1, \ldots, w_{k_i+1}$ such that $G_i[\{u,v,w_j\}]$ is isomorphic to $P_3$ for every $j \in \{1,\ldots,k_i+1\}$. Since $\{u,v\}$ is the only pair of vertices shared between $\{u,v,w_{j_1}\}$ and $\{u,v,w_{j_2}\}$ for two different $j_1,j_2 \in \{1,\ldots,k_i+1\}$, $\{u,v\} \notin E_i \cup E'_i$ and $|\Delta(E_i, E_i')|\le k_i$, it follows that there is some $j_0 \in \{1,\ldots,k_i+1\}$ such that $G'_i[\{u,v,w_{j_0}\}]$ is isomorphic to $P_3$, a contradiction.
% Hence $E_i$ contains $\{u,v\}$ and $D, E_1, \ldots E_\ell$ is also a solution to $(G_1, \ldots, \widehat{G_i}, \ldots G_\ell, k_1,\ldots,\widehat{k_i},\ldots, k_\ell,d)$.
\end{proof}%}

As with the classical \textsc{Cluster Editing} we can bound the number of vertices involved in a $P_3$ in each layer.
Let $R_i \subseteq V$ be the set of the vertices $v$ that appear in some induced $P_3$ in $G_i$.
%and let $R=\bigcup_{i=1}^\ell R_i$. 

% \begin{rrule}[\appref{proof:rul:dirty}]\label{rul:dirty}
\begin{rrule}\label{rul:dirty}
 If there is a layer $i \in [\ell]$ such that $|R_i| > k_i^2+2k_i$, then answer~NO.
\end{rrule}

% \appendixcorrectnessproof{rul:dirty}{
\begin{lem}\label{lem:dirty}
 \autoref{rul:dirty} is correct. 
\end{lem}
\begin{proof}
 Suppose towards a contradiction that $|R_i| > k_i^2+2k_i$ and $I=(G_1, \ldots, G_\ell,$ $k_1, \ldots, k_\ell,d)$ is a yes-instance. 
Let $(D, M_1, \ldots M_\ell)$ be a solution to $I$ and define $G'_i=(V,E(G_i)\oplus M_i)$.
 For each modified edge~$\{u,v\} \in M_i$ denote by $R_{uv}$ the set of vertices~$w$ such that the induced subgraph~$G_i[\{u,v,w\}]$ is a $P_3$. Since the instance is reduced with respect to Reduction Rules~\ref{rul:heavy_edge} and~\ref{rul:heavy_non},
for each modified edge~$\{u,v\}\in M_i$ we have $|R_{uv}| \le k_i$.
Since $G'_i$ is a cluster graph and, thus, does not contain $P_3$ as an induced subgraph, we know that $R_i \subseteq \bigcup_{\{u,v\} \in M_i} (\{u,v\} \cup R_{uv})$. It follows that $|R_i| \le k_i \cdot (2+k_i)=k_i^2+2k_i$---a contradiction.
\end{proof}%}
% }

As a major difference to \textsc{Cluster Editing} for a single layer, we cannot simply remove the vertices that are not involved in any $P_3$ since we require the cluster graphs in individual layers not to differ too much.
We show that the vertices in the clusters that do not change can be freely removed.

\begin{rrule}%[\appref{proof:rul:same}]
\label{rul:same}
 If there is a subset $A \subseteq V \setminus R$ such that 
 for each layer~$i\in [\ell]$, 
 the subset~$A$ is the vertex set of a connected component of $G_i$, then remove $A$ (and the corresponding edges) from every $G_i$.
\end{rrule}
\appendixcorrectnessproof{rul:same}{
\begin{lem}\label{lem:same}
 \autoref{rul:same} is correct.
\end{lem}
\begin{proof}
Let $I=(G_1, \ldots, G_\ell,k_1, \ldots, k_\ell,d)$ be the original instance and $\widehat{I}=(\widehat{G_1}, \ldots, \widehat{G_\ell}, k_1,$ $\ldots,k_\ell,d)$ be the instance after the application of the rule, where for each $i\in [\ell]$, we have $\widehat{G_i}= G_i[V(G_i)\setminus A]$. 
Since $A \cap R = \emptyset$ we have that $A$ induces a complete subgraph in each layer~$i\in[\ell]$.
Moreover, for each layer~$i\in[\ell]$,
the set~$A$ is the vertex set of a connected component of graph~$G_i$.
Thus,~$G_i[A]$ is a complete connected component in~$G_i$, meaning that $E(G_i)= E(\widehat{G_i})\cup \binom{A}{2}$.

%$N_{G_i}[v]=A$ for every $v \in A$ and every $i \in \{1, \ldots, \ell\}$. 

Let $(D, M_1, \ldots, M_\ell)$ be a solution to $I$ 
and let $\widehat{D}=D\setminus A$ and $\widehat{M_i}= M_i \cap \binom{V \setminus A}{2}$ for every $i \in [\ell]$. Then $(\widehat{D}, \widehat{M_1}, \ldots, \widehat{M_\ell})$ forms a solution to $\widehat{I}$.

Conversely, let $\widehat{S}=(\widehat{D}, \widehat{M_1}, \ldots, \widehat{M_\ell})$ be a solution to $\widehat{I}$.
We claim that $\widehat{S}$ is a also solution to $I$. 
Indeed, each $G'_i=(V, E(G_i)\oplus \widehat{M})$ is a cluster graph (note that $A$ is the vertex set of a complete connected component in $G_i$),
$|\widehat{M_i}|\le k_i$ and for all $1\le i, j \le \ell$ we have that 
$E(G_i)\oplus \widehat{M_i} \cap \binom{V\setminus D}{2} = 
 E(G_j)\oplus \widehat{M_j} \cap \binom{V\setminus D}{2}$ since 
 $E(\widehat{G_i}) \oplus \widehat{M_i} \cap \binom{V\setminus (A\cup D)}{2} = 
E(\widehat{G_j})\oplus \widehat{M_j} \cap \binom{V\setminus (A\cup D)}{2}$.
\end{proof}}

The next rule allows us to reduce vertices that appear in exactly the same clusters, if there are many.

\begin{rrule}%[\appref{proof:rul:core}]
\label{rul:core}
 If there is a set $A \subseteq V \setminus R$ with $|A| \ge k+d+3$ 
 such that for every layer~$i\in[\ell]$
 it holds that all vertices of $A$ are in the same connected component of $G_i$,
 then select an arbitrary~$v \in A$ and remove $v$ from every $G_i$.
\end{rrule}
%d in the size is necessary, think of the reduction from VC with k=0 and replace each vertex by true twins and double d. Then the rule would apply to each such twins, distorting the budget.
\appendixcorrectnessproof{rul:core}{
\begin{lem}\label{lem:core}
 \autoref{rul:core} is correct.
\end{lem}

For the proof of this and subsequent lemmas we find the following observation handy.
\begin{obs}\label{obs:cut}
  Let $G$ be a complete graph on at least $k+2$ vertices and let $H$ be a cluster graph 
  such that $V(G)=V(H)$.
  If $H$ is not complete, then $G$ and $H$ differ in at least $k+1$ edges, i.e., $|E(G)\oplus E(H)| \ge k+1$.
\end{obs}
\begin{proof}%[Proof of \autoref{obs:cut}]
The statement is obviously true for~$k \le 0$, let us assume that $k>0$.
Since $H$ is a cluster graph which is not complete, it must have several connected components. Let $X$ be the smallest of these connected components and define $Y = V(G) \setminus X$. The set $E(G) \oplus E(H)$ must contain at least all the edges between $X$ and $Y$, 
hence $|E(G)\oplus E(H) | \ge |X| \cdot |Y|$. 

If $|X| \ge \frac{k+2}{2}$, then also $|Y| \ge \frac{k+2}{2}$ since $X$ is the smallest component. But then $|X| \cdot |Y| \ge (\frac{k+2}{2})^2 =\frac{k^2+4k+4}{4} \ge k+1$. 

If $|X| < \frac{k+2}{2}$, then let us denote $x=|X|$ and we have $|Y| \ge k+2 - x$. We know that~$|X| \cdot |Y| \ge x \cdot (k+2-x)$. The function $f(x)=x \cdot (k+2-x)$ is increasing for $x < \frac{k+2}{2}$ with $f(1) =k+1$, hence $|X| \cdot |Y| \ge k+1$, finishing the proof.
\end{proof}}

\begin{proof}[Proof of \autoref{lem:core}]
 Let $I=(G_1, \ldots, G_\ell,k_1, \ldots, k_\ell,d)$ be the original instance and $\widehat{I}=(\widehat{G_1}, \ldots, \widehat{G_\ell}$, %hyphen
$k_1,\ldots,k_\ell,d)$, where $\widehat{G_i}= G_i[V(G_i)\setminus \{v\}]$ be the instance after the application of the rule. 
 %Note that  
 Let $(D, M_1, \ldots, M_\ell)$ be a solution to $I$,
 and 
 define $\widehat{D}=D\setminus \{v\}$ and for every~$i\in[\ell]$ $\widehat{M_i}= M_i \cap \binom{V \setminus\{v\}}{2}$. Then $(\widehat{D}, \widehat{M_1}, \ldots, \widehat{M_\ell})$ forms a solution to $\widehat{I}$.
 
Conversely, let $\widehat{S}=(\widehat{D}, \widehat{M_1}, \ldots, \widehat{M_\ell})$ be a solution to $\widehat{I}$.
Let $w$ be an arbitrary vertex of $A\setminus (\widehat{D} \cup \{v\})$ (note that since $|A| \ge k+d+3$, $|\widehat{D}|\le d$ and $k\ge 0$, the set~$A\setminus (\widehat{D} \cup \{v\})$ is not empty).
We will construct a solution for $I$ such that after applying the solution $v$ is a true twin of $w$ in every layer, i.e., we will put $v$ into the same clusters as $w$.
Formally, for each layer~$i\in[\ell]$, 
we define $\widehat{E_i}' = E(\widehat{G_i})\oplus \widehat{M_i}$, 
$E'_i= \widehat{E_i}' \cup \{\{x,v\}\mid \{x,w\} \in \widehat{E_i}'\}\cup\{\{v,w\}\}$ and $M_i = E_i\oplus E'_i$.  
We claim that $(\widehat{D}, M_1, \ldots, M_\ell)$ is a solution to $I$. 
 
 First, each $G'_i=(V, E'_i)$ is a cluster graph. If there are two layers $i,j \in\{1\ldots \ell\}$ such that $G'_i \setminus D \neq G'_j \setminus D$, then without loss of generality we can assume that there is some~$x\in V\setminus (D \cup \{v\})$ such that $\{v,x\} \in E'_i$ but $\{v,x\} \notin E'_j$. But then $\{w,x\} \in \widehat{E_i}'$ and~$\{w,x\} \notin \widehat{E_j}'$, a contradiction since neither $w$ nor $x$ is in $\widehat{D}$.
 
 Finally, let us show that for each layer~$i\in[\ell]$ we have that $|M_i| \le |\widehat{M_i}|\le k_i$. 
 To this end, we first observe that all vertices of $A\setminus \{v\}$ are in the same component of $(V\setminus\{v\}, \widehat{E_i}')$. 
 Since $A \cap R = \emptyset$ and all vertices of $A$ are in the same connected component of $G_i$ we have that $\widehat{G_i}[A\setminus\{v\}]$ is complete. 
Hence, if $A\setminus \{v\}$ does not induce a complete subgraph in $(V\setminus\{v\}, \widehat{E_i}')$, then by \autoref{obs:cut}, we can conclude that $\widehat{M_i}$ contains at least $k+1$ edges, as $|A \setminus \{v\}| \ge k+d+2 \ge k+2$.
 
 Now for every $x \in V \setminus A$ and every $i \in [\ell]$ we have that $\{v,x\} \in E(G_i)$ if and only if $\{u,x\} \in E(G_i)$ for every $u \in A$ as otherwise the induced subgraph~$G_i[\{v,u,x\}]$
would be a $P_3$, contradicting $A \cap R = \emptyset$.
 Similarly, for every $x \in V \setminus A$ and every $i \in [\ell]$ we have that $\{v,x\} \in E'_i$ if and only if $\{u,x\} \in E'_i$ for every $u \in A$, since $(V\setminus\{v\}, \widehat{E_i}')$ is a cluster graph and since $E'_i$ is constructed in this way.
 It follows that if $\{x,v\} \in M_i$ for some~$x \in V \setminus A$, then $\{x,u\} \in \widehat{M_i}$ for every $u \in A \setminus \{v\}$ and $|\widehat{M_i}| \ge k+1 \ge k_i+1$---a contradiction.
 Hence $\{x,v\}\notin M_i$ for every $x \in V \setminus \{v\}$ and thus~$M_i \subseteq \widehat{M_i}$.
\end{proof}

The next rule shows that the remaining clusters in a yes-instance cannot be too large.

\begin{rrule}%[\appref{proof:rul:big}]
\label{rul:big}
 If there is a layer $i \in [\ell]$ and a connected component $A$ of $G_i$ with $|A\setminus R| \ge k+2d+3$, then answer NO.
\end{rrule}
\appendixcorrectnessproof{rul:big}{
Let us first make a folklore observation.
\begin{obs}\label{obs:folklore}
 If a connected component~$C$ of a graph has at least three vertices and is not complete, 
 then every vertex of $C$ appears in some induced $P_3$.
\end{obs}
\begin{proof}
 Consider an arbitrary vertex~$u \in V(C)$.
 If $u$ is adjacent to~$v$ for every $v \in  V(C) \setminus \{u\}$, 
 then there must be some pair~$\{x,y\} \subseteq V(C) \setminus \{u\}$ of vertices such that $\{x,y\} \notin E(C)$ since the component is not complete. Then, $C[\{u,x,y\}]$ is a $P_3$.
 
 Otherwise $u$ is not adjacent to some vertex~$v \in  V(C) \setminus \{u\}$. Then let $P$ be the shortest path between $u$ and $v$. This path has at least three vertices and each three consecutive vertices of this path induce a subgraph which is a $P_3$.
\end{proof}

\begin{lem}
 \autoref{rul:big} is correct.
\end{lem}

\begin{proof}
Suppose towards a contradiction that $A$ is a connected component of $G_i$ for some layer~$i \in [\ell]$
with $|A\setminus R| \ge k+2d+3$, and $(G_1, \ldots, G_\ell,k_1, \ldots, k_\ell,d)$ is a yes-instance. 
Let $(D, M_1, \ldots M_\ell)$ be a solution to the instance.
For every layer~$j\in[\ell]$, define $E'_j = E(G_{j})\oplus M_j$ and $G'_{j}=(V,E'_j)$.
Let~$A' =A \setminus (D \cup R)$ and note that $|A'|\ge k+d+3$. 
We claim that for every $j\in [\ell]$, 
all vertices of $A'$ are in the same connected component of $G_{j}$,
contradicting the instance being reduced with respect to \autoref{rul:core}.

Since $A' \cap R = \emptyset$ and 
all vertices of $A$ are in the same connected component of $G_i$,
by \autoref{obs:folklore} we have that $G_i[A']$ is complete. 
Hence, if $G'_i[A']$ is not complete, then by \autoref{obs:cut},
$M_i$ contains at least $k+1 \ge k_i+1$ edges, as $|A'| \ge k+2$, a contradiction.
Therefore $G'_i[A']$ is complete. 
For every $j\in [\ell]$, since $G'_{j}[V\setminus D] = G'_{i}[V\setminus D]$,
the graph $G'_j[A']$ is complete. 
Now again, if $G_j[A']$ is not complete for some layer~$j \in [\ell]$, then again by \autoref{obs:cut} $M_j$ contains at least $k+1 \ge k_j+1$ edges, a contradiction. 
\end{proof}}

Now we introduce our final rule bounding the number of vertices in the instance.

\begin{rrule}%[\appref{proof:rul:many}]
\label{rul:many}
 If $|V| > \ell\cdot(k^2+2k+d\cdot(k+2d+2)+2k)$, then answer NO. 
\end{rrule}
\begin{lem}
 \autoref{rul:many} is correct.
\end{lem}

\begin{proof}
Suppose towards a contradiction that $|V| > \ell\cdot(k^2+2k+d\cdot(k+2d+2)+2k)$ and $(G_1, \ldots, G_\ell,k_1, \ldots, k_\ell,d)$ is a yes-instance. 
Let $(D, M_1, \ldots M_\ell)$ be a solution to the instance.
For each $i\in [\ell]$, define $E'_i=E(G_i)\oplus M_i$ and $G'_{i}=(V,E'_{i})$.
Let us denote by $S=\bigcup_{i=1}^\ell\bigcup_{\{u,v\} \in M_i} \{u,v\}$ the set of vertices adjacent to any modification. Obviously,~$|S| \le \ell\cdot 2k$.

For every layer~$i \in [\ell]$ and every $x \in D$ let us denote by $Q_x^i\subseteq V \setminus R$ the set of vertices from $V \setminus R$ in the same connected component of $G'_i$ as the vertex $x$ and $Q=\bigcup_{i=1}^\ell \bigcup_{x \in D} Q_x^i$.
Since the instance is reduced with respect to \autoref{rul:big}, we know that $|Q_x^i|\le k+2d+2$ and, thus $|Q| \le \ell\cdot d\cdot (k+2d+2)$.

Note also that $|R| \le \ell\cdot (k^2 +2k)$, since the instance is reduced with respect to \autoref{rul:dirty}. Now since $|V| > \ell\cdot(k^2+2k+d\cdot(k+2d+2)+2k)$, $|R| \le \ell\cdot (k^2 +2k)$, $|Q| \le \ell\cdot d\cdot (k+2d+2)$, and $|S| \le \ell\cdot 2k$, the set $V'=V \setminus (Q \cup R \cup S)$ is not empty. Let $u$ be an arbitrary vertex from~$V'$. Since the instance is reduced with respect to \autoref{rul:same}, we know that there are two distinct layers~$i,j \in [\ell]$  and a vertex $v$ such that $u$ and $v$ are in the same connected component of $G_i$ and in different connected components of $G_j$. Since $v$ is not in $S$, we know that the same holds for the graphs~$G'_i$ and $G'_j$. However, since $v$ is not in $Q$ neither in $R$, we have that neither $u$ nor $v$ is in $D$. But then $G'_i[V\setminus D]$ and $G'_j[V\setminus D]$ are different, a contradiction.
\end{proof}
\todo[inline]{ms: This section contains many ``a contradiction'' but it would be better to say what is contradicted.}

After bounding the size of the instance through \autoref{rul:many} % effectively bounds the size of the reduced instance to polynomial in $k$,~$d$, and $\ell$.
it remains to transform the resulting instance of \MLCES{} to an equivalent instance of \MLCE{}. To this end we introduce new vertex set $A$ of size exactly $2k+2$ to $V$ and to each $E_i$ introduce all edges from $\binom{A}{2}$. Then, for each $i \in \{1, \ldots, \ell\}$ we remove $k-k_i$ arbitrary edges between vertices of $A$ from $E_i$ and set $k_i=k$. 

If $\{u,v\}$ is an edge removed in this step, then $u$ and $v$ had $2k$ common neighbors in $A$ and by at most $k-1$ other edge removals they could loose at most $k-1$ of them. Hence, \autoref{rul:heavy_non} would apply to each pair of vertices from $A$ with an edge removed.
Applying \autoref{rul:heavy_non} exhaustively and then \autoref{rul:same} would revert all the changes made. Hence, the constructed instance is equivalent to the one obtained after exhaustive application of all the reduction rules. 

The constructed instance can be turned into an equivalent instance of \MLCE{} in an obvious way.
Since no rule increases $k$, $d$, or $\ell$, $|V|=O(\ell\cdot(k+d)^2)$, the resulting instance can be described using $O(\ell^3\cdot(k+d)^4)$ bits and it is equivalent to the original instance, it remains to show that the kernelization is computable in polynomial time.

\begin{lem}%[\appref{proof:lem:kernelruntime}]
\label{lem:kernelruntime}
 The kernelization can be done in $O(\ell\cdot n^3)$ time.
\end{lem}
\appendixproof{lem:kernelruntime}{
\begin{proof}
% Let us denote $n=|V|$. 
If $n < k^2$, then we can output the original instance as the kernel. Let us assume that $k^2 \le n$.

We can check whether \autoref{rul:negat} applies in $O(\ell)$ time on the beginning and in constant time whenever any later rule changes the budget. Applying the rule takes constant time.

For each layer $i \in \{1, \ldots, \ell\}$ in time $O(n^3)$ we can count for each pair of vertices in how many induced subgraphs isomorphic to $P_3$ it appears and classify the pairs according to that count. Then we apply Reduction Rules \ref{rul:heavy_edge} and \ref{rul:heavy_non} to the pairs which appear in many~$P_3$'s. Each application takes $O(n)$ time and at the same time we can update the counts for affected pairs. Hence, these reduction rules can be exhaustively applied to one layer in $O(n^3)$ time. Also in the same time we can determine the sets $R_i$ and eventually apply \autoref{rul:dirty}. Since the later rules only delete vertices or answer NO, no application of a later rule can create an opportunity to apply \autoref{rul:heavy_edge}, \ref{rul:heavy_non}, or \ref{rul:dirty}. Hence, these reduction rules can be exhaustively applied to the instance in $O(\ell \cdot n^3)$ time.

In $O(\ell \cdot n^2)$ time we can compute the graphs $G_{\cap}=(V,\bigcap_{i=1}^{\ell} E_i)$ and $G_{\cup}=(V,\bigcup_{i=1}^{\ell} E_i)$. Then \autoref{rul:same} applies to all connected components of $G_{\cup}$ not containing vertices of~$R$ that are also connected components of $G_{\cap}$. All of these applications can be recognized in $O(n^2)$ time and all of them together applied in $O(\ell\cdot n^2)$ time. No application of a later rule can create an opportunity to apply \autoref{rul:same}. 

\autoref{rul:core} applies to each connected component of $G_{\cap}$ which has the appropriate number of vertices not in $R$. All of these applications can be recognized in $O(n^2)$ time and all of them together applied in $O(\ell\cdot n^2)$ time. Since later rules only answer NO, no application of a later rule can create an opportunity to apply \autoref{rul:core}.

We can check whether the rule applies in $O(\ell \cdot n^2)$ for \autoref{rul:big} and in constant time for \autoref{rul:many} and apply any of them in constant time.

Hence the reduction rules can be exhaustively applied in $O(\ell\cdot n^3)$ time, the final reduction back to \MLCE{} takes $O(k^2)=O(n)$ time and the result follows.
\end{proof}}

Lastly, we argue that slightly modified reduction rules can be applied to \TCE{} (with individual edge-modification budgets, where the resulting instance can be transformed back). Intuitively, this follows from the following observations: The reduction rules do not mark vertices, and the union of all marked vertices of a \TCE{} solution together with the edge modification sets forms a solution for a \MLCE{} instance, where the maximal number of marked vertices is $d\cdot\ell$. Hence, replacing $d$ with $d\cdot\ell$ in the description of all reductions rules yields a set of rules that produce a kernel of size $O(\ell^3\cdot(k+d\cdot\ell)^4)$ for \TCE{}.%~(\appref{app:tcekernel}). 
%This concludes the proof of \cref{thm:pk}. 

% \subsection{Kernel for \TCE{}}\label{app:tcekernel}
\begin{lem}
If in the descriptions of \autoref{rul:negat}, \autoref{rul:heavy_edge}, \autoref{rul:heavy_non}, \autoref{rul:dirty}, \autoref{rul:same}, \autoref{rul:core}, \autoref{rul:big}, and \autoref{rul:many} ``$d$'' is replaced with ``$d\cdot\ell$'', then these rules yield a kernel of size $O(\ell^3\cdot(k+d\cdot\ell)^4)$ for \TCE\ that can be computed in $O(\ell\cdot n^3)$ time.
\end{lem}
\begin{proof}[Proof Sketch]
First, we observe that none of the original reduction rules decrease $d$, informally that means that no rule marks vertices. Second, note that none of the rules create trivial yes-instances. 

It is easy to see the following: If $I=(G_1=(V, E_1),\ldots, G_\ell=(V, E_\ell), k, d)$ is a yes-instance of \TCE, then $I'=(G_1=(V, E_1),\ldots, G_\ell=(V, E_\ell), k, d')$ with $d'=d\cdot\ell$ is a yes-instance of \MLCE{}. Given a solution $(M_1, \ldots, M_\ell, D_1, \ldots, D_{\ell-1})$ for $I$, it is easy to check that $(M_1, \ldots, M_\ell, D)$ with $D =\bigcup_{1\le i\le\ell-1}D_i$ is a solution for $I'$. By contraposition this means that if $I'$ is a no-instance for \MLCE{}, then $I$ is a no-instance for \TCE{}. If follows that all modified rules creating trivial no-instances are safe. The safeness of the remaining modified rules follows from straight forward adaptations of the original \todo{ms: Safeness, reduction/branching rules, etc. should go into prelims.}safeness proofs.
\end{proof}

Finally, we argue that for the parameter $n$ number of vertices (and all smaller parameters), \MLCE{} and \TCE{} do not admit polynomial kernels unless \NoKernelAssume{}. More specifically, we claim the following.

\begin{prop}%[\appref{proof:thm:nopolykernel-n}]
\label{thm:nopolykernel-n}
 \MLCE{} and \TCE{} do not admit polynomial kernels with respect to the number $n$ of vertices, unless \NoKernelAssume.
\end{prop}

% Notice that if the clustering is allowed to change completely in each layer,
% that is, $d = n$, then our task reduces to edit each layer into a
% cluster graph with $k$ modifications each. That is, each layer
% represents its own, independent instance of \CE. This naturally yields
% an AND-cross-composition~\cite{bodlaender2014kernelization, CyganFKLMPPS15} from \CE, proving \autoref{thm:nopolykernel-n}.

\appendixproof{thm:nopolykernel-n}{
  We need the following notation for the proof.
An equivalence
relation~$R$ on the instances of some problem~$L$ is a
\emph{polynomial equivalence relation} if
\begin{compactenum}[(i)]
 \item one can decide for each two instances in time polynomial in their sizes whether they belong to the same equivalence class, and
 \item for each finite set~$S$ of instances, $R$ partitions the set into at most~$(\max_{x \in S} |x|)^{O(1)}$ equivalence classes.  
\end{compactenum}

An \emph{AND-cross-composition} of a problem~$L\subseteq \Sigma^*$ into a
parameterized problem~$P$ (with respect to a polynomial equivalence
relation~$R$ on the instances of~\(L\)) is an algorithm that takes
$\ell$ $R$-equivalent instances~$x_1,\ldots,x_\ell$ of~$L$ and
constructs in time polynomial in $\sum_{i=1}^\ell |x_i|$ an instance
$(x,k)$ of~\(P\) such that
\begin{compactenum}[(i)]
\item $k$ is polynomially upper-bounded in $\max_{1\leq i\leq \ell}|x_i|+\log(\ell)$ and 
\item $(x,k)\in P$ if and only if $x_{\ell'}\in L$ for every $\ell'\in \{1,\ldots,\ell\}$. 
\end{compactenum}

If an \NP-hard problem~\(L\) AND-cross-composes into a parameterized
problem~$P$, then~$P$ does not admit a polynomial-size kernel, unless \NoKernelAssume~\cite{bodlaender2014kernelization, CyganFKLMPPS15},
which would cause a collapse of the polynomial-time hierarchy to the third
level.

\begin{proof}[Proof of \autoref{thm:nopolykernel-n}]
We provide an AND-cross-composition from classical \textsc{Cluster Editing}.
We define relation~$R$: Two instances $(G_1,k_1)$ and $(G_2,k_2)$ are equivalent under~$R$ if and only if $k_1=k_2$ and~$|V(G_1)|=|V(G_2)|$. Clearly, $R$ is a polynomial equivalence relation. 

Now let $(G_1, k_1),\ldots,(G_\ell,k_\ell)$ be $R$-equivalent instances of \textsc{Cluster Editing}. Then there is an integer $k \in \N$ such that $k=k_i$ for every $i\in \{1, \ldots,\ell\}$. Moreover, since the names of the vertices are not important for the problem and $|V(G_i)|=|V(G_j)|$ for every~$i, j\in \{1, \ldots,\ell\}$, we can assume without loss of generality that there is a set $V$ such that $V=V(G_i)$ for every~$i\in \{1, \ldots,\ell\}$. Hence, $(G_1, \ldots, G_\ell, k, d)$, where $d=|V|$, is a valid instance of \MLCE{} and \TCE{}.

This instance can be constructed in polynomial time and no extra vertices are added, hence~$|V|$ is upper-bounded by a maximum size of an input instance. Furthermore, as we are allowed to mark all vertices, it follows directly from the definition of \MLCE{} and \TCE{}, that $(G_1, \ldots, G_\ell, k, d)$ is a 
yes-instance if and only if for every $i\in \{1, \ldots,\ell\}$ it is possible to turn $G_i$ into a cluster graph by at most $k$ edge modifications.

Since \textsc{Cluster Editing} is \NP-hard~\cite{bansal2004correlation} and we AND-cross-composed it into \MLCE{} and into \TCE{} parameterized by $n=|V|$, the result follows.
\end{proof}
}
% \subsection{Reduction Rules}
% 
% \begin{compactitem}
%   \item Remove any connected component that is a clique and exists (as a connected component) in every time step.
%   \item If there is a vertex $v$ and two time steps $t_1$ and $t_2$ such that $|N_{t_1}(v) \setminus N_{t_2}(v)|>k+d$, then mark $v$.
%   \item If there is a vertex $v$ such that $|\{u \mid \forall t: \{u,v\} \in E_t\}|\ge 2k+d+1$ then there is no point in marking $v$.
%   %\item If there is a clique of size $k+d+2$ at any point in time, this clique
% \end{compactitem}

\section{Conclusion}
Our results highlight that \TCE\ and \MLCE\ are much richer in
structure than classical \textsc{Cluster Editing}. Techniques for the
classical problem seem to only carry over somewhat for kernelization
algorithms and otherwise new methods are necessary. In this regard, we
contribute a major step forwards with our fixed-parameter algorithm
for \MLCE\ with respect to the combination of~$k$ and~$d$. In contrast, the \W{1}-hardness
for \TCE\ with respect to~$k$ for $d = 3$ highlights the obstacles we
need to overcome. Perhaps we can break the problematic temporal
non-locality by bounding the number of allowed modifications at one
vertex in any given interval of layers of some fixed size.

\bibliography{bib}

\begin{thebibliography}{39}
\providecommand{\natexlab}[1]{#1}
\providecommand{\url}[1]{\texttt{#1}}
\expandafter\ifx\csname urlstyle\endcsname\relax
  \providecommand{\doi}[1]{doi: #1}\else
  \providecommand{\doi}{doi: \begingroup \urlstyle{rm}\Url}\fi

\bibitem[Agrawal et~al.(2016)Agrawal, Lokshtanov, Mouawad, and Saurabh]{ALMS15}
A.~Agrawal, D.~Lokshtanov, A.~E. Mouawad, and S.~Saurabh.
\newblock Simultaneous feedback vertex set: A parameterized perspective.
\newblock In \emph{Proceedings of the 33rd International Symposium on
  Theoretical Aspects of Computer Science ({STACS}~'16)}, volume~47 of
  \emph{LIPIcs}, pages 7:1--7:15. Schloss Dagstuhl - Leibniz-Zentrum fuer
  Informatik, 2016.

\bibitem[Akrida et~al.(2018)Akrida, Mertzios, Spirakis, and Zamaraev]{AMSZ18}
E.~C. Akrida, G.~B. Mertzios, P.~G. Spirakis, and V.~Zamaraev.
\newblock Temporal vertex cover with a sliding time window.
\newblock In \emph{Proceedings of the 45th International Colloquium on
  Automata, Languages, and Programming ({ICALP} 2018)}, volume 107 of
  \emph{LIPIcs}, pages 148:1--148:14. Schloss Dagstuhl - Leibniz-Zentrum fuer
  Informatik, 2018.

\bibitem[Bansal et~al.(2004)Bansal, Blum, and Chawla]{bansal2004correlation}
N.~Bansal, A.~Blum, and S.~Chawla.
\newblock Correlation clustering.
\newblock \emph{Machine Learning}, 56:\penalty0 89--113, 2004.

\bibitem[Barigozzi et~al.(2011)Barigozzi, Fagiolo, and
  Mangioni]{barigozzi_identifying_2011}
M.~Barigozzi, G.~Fagiolo, and G.~Mangioni.
\newblock Identifying the community structure of the international-trade
  multi-network.
\newblock \emph{Physica A: Statistical Mechanics and its Applications},
  390\penalty0 (11):\penalty0 2051--2066, 2011.

\bibitem[Bentert et~al.(2018)Bentert, Himmel, Molter, Morik, Niedermeier, and
  Saitenmacher]{BHMMNS16}
M.~Bentert, A.-S. Himmel, H.~Molter, M.~Morik, R.~Niedermeier, and
  R.~Saitenmacher.
\newblock Listing all maximal $k$-plexes in temporal graphs.
\newblock In \emph{Proceedings of the 2018 {IEEE/ACM} International Conference
  on Advances in Social Networks Analysis and Mining, ({ASONAM} '18)}, pages
  41--46. {IEEE} Computer Society, 2018.

\bibitem[Betzler et~al.(2011)Betzler, Guo, Komusiewicz, and
  Niedermeier]{betzler_average_2011}
N.~Betzler, J.~Guo, C.~Komusiewicz, and R.~Niedermeier.
\newblock Average parameterization and partial kernelization for computing
  medians.
\newblock \emph{Journal of Computer and System Sciences}, 77\penalty0
  (4):\penalty0 774--789, 2011.

\bibitem[B{\"o}cker and Baumbach(2013)]{bocker_cluster_2013}
S.~B{\"o}cker and J.~Baumbach.
\newblock Cluster {Editing}.
\newblock In \emph{Proceedings of the 9th Conference on Computability in Europe
  (CiE '13)}, volume 7921 of \emph{LNCS}, pages 33--44. {Springer}, 2013.

\bibitem[Bodlaender et~al.(2014)Bodlaender, Jansen, and
  Kratsch]{bodlaender2014kernelization}
H.~L. Bodlaender, B.~M. Jansen, and S.~Kratsch.
\newblock Kernelization lower bounds by cross-composition.
\newblock \emph{SIAM Journal on Discrete Mathematics}, 28\penalty0
  (1):\penalty0 277--305, 2014.

\bibitem[Bredereck et~al.(2017)Bredereck, Komusiewicz, Kratsch, Molter,
  Niedermeier, and Sorge]{bredereck2017assessing}
R.~Bredereck, C.~Komusiewicz, S.~Kratsch, H.~Molter, R.~Niedermeier, and
  M.~Sorge.
\newblock Assessing the computational complexity of multi-layer subgraph
  detection.
\newblock In \emph{Proceedings of the 10th International Conference on
  Algorithms and Complexity (CIAC '17)}, pages 128--139. Springer, 2017.

\bibitem[Bulteau et~al.(2015)Bulteau, Chen, Faliszewski, Niedermeier, and
  Talmon]{BulCheFalNieTal2015}
L.~Bulteau, J.~Chen, P.~Faliszewski, R.~Niedermeier, and N.~Talmon.
\newblock Combinatorial voter control in elections.
\newblock \emph{Theoretical Computer Science}, 589:\penalty0 99--120, 2015.

\bibitem[Cai and Ye(2014)]{CY14}
L.~Cai and J.~Ye.
\newblock Dual connectedness of edge-bicolored graphs and beyond.
\newblock In \emph{Proceedings of the 39th International Symposium on
  Mathematical Foundations of Computer Science ({MFCS}~'14)}, volume 8635 of
  \emph{LNCS}, pages 141--152. Springer, 2014.

\bibitem[Cao and Chen(2012)]{cao_cluster_2012}
Y.~Cao and J.~Chen.
\newblock Cluster {Editing}: {Kernelization} based on edge cuts.
\newblock \emph{Algorithmica}, 64\penalty0 (1):\penalty0 152--169, 2012.

\bibitem[Chen et~al.(2018)Chen, Molter, Sorge, and
  Such{\'{y}}]{chen2018parameterized}
J.~Chen, H.~Molter, M.~Sorge, and O.~Such{\'{y}}.
\newblock Cluster editing in multi-layer and temporal graphs.
\newblock In \emph{Proceedings of the 29th {I}nternational {S}ymposium on
  {A}lgorithms and {C}omputation ({ISAAC} '18)}, volume 123 of \emph{LIPIcs},
  pages 24:1--24:13. Schloss Dagstuhl - Leibniz-Zentrum fuer Informatik, 2018.

\bibitem[Cygan et~al.(2015)Cygan, Fomin, Kowalik, Lokshtanov, Marx, Pilipczuk,
  Pilipczuk, and Saurabh]{CyganFKLMPPS15}
M.~Cygan, F.~Fomin, {\L{}}.~Kowalik, D.~Lokshtanov, D.~Marx, M.~Pilipczuk,
  M.~Pilipczuk, and S.~Saurabh.
\newblock \emph{Parameterized Algorithms}.
\newblock Springer, 2015.

\bibitem[Dehne et~al.(2004)Dehne, Fellows, Rosamond, and
  Shaw]{hutchison_greedy_2004}
F.~Dehne, M.~Fellows, F.~Rosamond, and P.~Shaw.
\newblock Greedy {{Localization}}, {{Iterative Compression}}, and {{Modeled
  Crown Reductions}}: {{New FPT Techniques}}, an {{Improved Algorithm}} for
  {{Set Splitting}}, and a {{Novel}} $2k$ {{Kernelization}} for {{Vertex
  Cover}}.
\newblock In \emph{Proceedings of 1st International Workshop on Parameterized
  and Exact Computation ({IWPEC}~'04)}, volume 3162 of \emph{LNCS}, pages
  271--280. Springer, 2004.

\bibitem[Dey et~al.(2017)Dey, Rossi, and Sidiropoulos]{dey_temporal_2017}
T.~Dey, A.~Rossi, and A.~Sidiropoulos.
\newblock Temporal {{Clustering}}.
\newblock In \emph{Proceedings of the 25th Annual European Symposium on
  Algorithms ({ESA} '17)}, volume~87 of \emph{LIPIcs}, pages 34:1--34:14.
  {Schloss Dagstuhl}, 2017.

\bibitem[D\"{o}rnfelder et~al.(2014)D\"{o}rnfelder, Guo, Komusiewicz, and
  Weller]{dornfelder_parameterized_2014}
M.~D\"{o}rnfelder, J.~Guo, C.~Komusiewicz, and M.~Weller.
\newblock On the parameterized complexity of consensus clustering.
\newblock \emph{Theoretical Computer Science}, 542:\penalty0 71--82, 2014.

\bibitem[Fellows et~al.(2009)Fellows, Hermelin, Rosamond, and
  Vialette]{fellows_parameterized_2009}
M.~R. Fellows, D.~Hermelin, F.~Rosamond, and S.~Vialette.
\newblock On the parameterized complexity of multiple-interval graph problems.
\newblock \emph{Theoretical Computer Science}, 410\penalty0 (1):\penalty0
  53--61, 2009.

\bibitem[Fluschnik et~al.(2018)Fluschnik, Molter, Niedermeier, and
  Zschoche]{FMNZ18}
T.~Fluschnik, H.~Molter, R.~Niedermeier, and P.~Zschoche.
\newblock Temporal graph classes: {A} view through temporal separators.
\newblock In \emph{Proceedings of the 44th International Workshop on
  Graph-Theoretic Concepts in Computer Science ({WG} 2018)}, volume 11159 of
  \emph{LNCS}, pages 216--227. Springer, 2018.

\bibitem[Fomin et~al.(2014)Fomin, Kratsch, Pilipczuk, Pilipczuk, and
  Villanger]{fomin_tight_2014}
F.~V. Fomin, S.~Kratsch, M.~Pilipczuk, M.~Pilipczuk, and Y.~Villanger.
\newblock Tight bounds for parameterized complexity of {Cluster Editing} with a
  small number of clusters.
\newblock \emph{Journal of Computer and System Sciences}, 80\penalty0
  (7):\penalty0 1430--1447, 2014.

\bibitem[Gramm et~al.(2005)Gramm, Guo, H{\"{u}}ffner, and
  Niedermeier]{GrammGHN05}
J.~Gramm, J.~Guo, F.~H{\"{u}}ffner, and R.~Niedermeier.
\newblock Graph-modeled data clustering: Exact algorithms for clique
  generation.
\newblock \emph{Theory of Computing Systems}, 38\penalty0 (4):\penalty0
  373--392, 2005.

\bibitem[Himmel et~al.(2017)Himmel, Molter, Niedermeier, and
  Sorge]{himmel_adapting_2017}
A.-S. Himmel, H.~Molter, R.~Niedermeier, and M.~Sorge.
\newblock Adapting the {{Bron}}–{{Kerbosch}} algorithm for enumerating
  maximal cliques in temporal graphs.
\newblock \emph{Social Netwprk Analysis and Mining}, 7\penalty0 (1):\penalty0
  35, 2017.

\bibitem[Holme(2015)]{holme2015}
P.~Holme.
\newblock Modern temporal network theory: a colloquium.
\newblock \emph{The European Physical Journal B}, 88\penalty0 (9):\penalty0
  234, 2015.

\bibitem[Holme and Saram{\"a}ki(2012)]{holme2012temporal}
P.~Holme and J.~Saram{\"a}ki.
\newblock Temporal networks.
\newblock \emph{Physics Reports}, 519\penalty0 (3):\penalty0 97--125, 2012.

\bibitem[Kempe et~al.(2002)Kempe, Kleinberg, and Kumar]{kempe2002connectivity}
D.~Kempe, J.~Kleinberg, and A.~Kumar.
\newblock Connectivity and inference problems for temporal networks.
\newblock \emph{Journal of Computer and System Sciences}, 64\penalty0
  (4):\penalty0 820--842, 2002.

\bibitem[Kim and Lee(2015)]{kim_community_2015}
J.~Kim and J.-G. Lee.
\newblock Community detection in multi-layer graphs: {A} survey.
\newblock \emph{ACM SIGMOD Record}, 44\penalty0 (3):\penalty0 37--48, 2015.

\bibitem[Kivelä et~al.(2014)Kivelä, Arenas, Barthelemy, Gleeson, Moreno, and
  Porter]{Kiv+14}
M.~Kivelä, A.~Arenas, M.~Barthelemy, J.~P. Gleeson, Y.~Moreno, and M.~A.
  Porter.
\newblock Multilayer networks.
\newblock \emph{Journal of Complex Networks}, 2\penalty0 (3):\penalty0
  203--271, 2014.

\bibitem[Komusiewicz and Uhlmann(2012)]{komusiewicz_cluster_2012}
C.~Komusiewicz and J.~Uhlmann.
\newblock Cluster editing with locally bounded modifications.
\newblock \emph{Discrete Applied Mathematics}, 160:\penalty0 2259--2270, 2012.

\bibitem[Latapy et~al.(2018)Latapy, Viard, and Magnien]{latapy2017stream}
M.~Latapy, T.~Viard, and C.~Magnien.
\newblock Stream graphs and link streams for the modeling of interactions over
  time.
\newblock \emph{Social Network Analysis and Mining}, 8\penalty0 (1):\penalty0
  61:1--61:29, 2018.

\bibitem[Luo et~al.(2018)Luo, Molter, Nichterlein, and Niedermeier]{LMNN18}
J.~Luo, H.~Molter, A.~Nichterlein, and R.~Niedermeier.
\newblock Parameterized dynamic cluster editing.
\newblock In \emph{Proceedings of the 38th IARCS Annual Conference on
  Foundations of Software Technology and Theoretical Computer Science, FSTTCS
  '18}, volume 122 of \emph{LIPIcs}, pages 46:1--46:15, 2018.

\bibitem[Mertzios et~al.(2013)Mertzios, Michail, Chatzigiannakis, and
  Spirakis]{mertzios2013temporal}
G.~B. Mertzios, O.~Michail, I.~Chatzigiannakis, and P.~G. Spirakis.
\newblock Temporal network optimization subject to connectivity constraints.
\newblock In \emph{Proceedings of the 40th International Colloquium on
  Automata, Languages, and Programming (ICALP~'13)}, pages 657--668. Springer,
  2013.

\bibitem[Mertzios et~al.(2019)Mertzios, Molter, and Zamaraev]{MMZ2019}
G.~B. Mertzios, H.~Molter, and V.~Zamaraev.
\newblock Sliding window temporal graph coloring.
\newblock In \emph{Proceedings of the 33rd AAAI Conference on Artificial
  Intelligence (AAAI '19)}. {AAAI} Press, 2019.
\newblock To appear.

\bibitem[Michail(2016)]{michail2016introduction}
O.~Michail.
\newblock An introduction to temporal graphs: An algorithmic perspective.
\newblock \emph{Internet Mathematics}, 12\penalty0 (4):\penalty0 239--280,
  2016.

\bibitem[Nicosia and Latora(2015)]{nicosia_measuring_2015}
V.~Nicosia and V.~Latora.
\newblock Measuring and modeling correlations in multiplex networks.
\newblock \emph{Physical Review E}, 92\penalty0 (3):\penalty0 032805, 2015.

\bibitem[Tagarelli et~al.(2017)Tagarelli, Amelio, and
  Gullo]{tagarelli_ensemble-based_2017}
A.~Tagarelli, A.~Amelio, and F.~Gullo.
\newblock Ensemble-based community detection in multilayer networks.
\newblock \emph{Data Mining and Knowledge Discovery}, 31\penalty0 (5):\penalty0
  1506--1543, 2017.

\bibitem[Tang et~al.(2009)Tang, Lu, and Dhillon]{tang_clustering_2009}
W.~Tang, Z.~Lu, and I.~S. Dhillon.
\newblock Clustering with {{Multiple Graphs}}.
\newblock In \emph{Proceedings of the 9th IEEE International Conference on Data
  Mining (ICDM '09)}, pages 1016--1021. IEEE Computer Society, 2009.

\bibitem[Tantipathananandh and
  Berger-Wolf(2011)]{tantipathananandh_finding_2011}
C.~Tantipathananandh and T.~Y. Berger-Wolf.
\newblock Finding {{Communities}} in {{Dynamic Social Networks}}.
\newblock In \emph{Proceedings of the 11th IEEE International Conference on
  Data Mining (ICDM '11)}, pages 1236--1241. IEEE Computer Society, 2011.

\bibitem[Tantipathananandh et~al.(2007)Tantipathananandh, Berger-Wolf, and
  Kempe]{tantipathananandh_framework_2007}
C.~Tantipathananandh, T.~Berger-Wolf, and D.~Kempe.
\newblock A {{Framework}} for {{Community Identification}} in {{Dynamic Social
  Networks}}.
\newblock In \emph{Proceedings of the 13th ACM SIGKDD International Conference
  on Knowledge Discovery and Data Mining ({KDD} '07)}, pages 717--726. {ACM},
  2007.

\bibitem[Zschoche et~al.(2018)Zschoche, Fluschnik, Molter, and
  Niedermeier]{ZFMN18}
P.~Zschoche, T.~Fluschnik, H.~Molter, and R.~Niedermeier.
\newblock On efficiently finding small separators in temporal graphs.
\newblock In \emph{Proceedings of the 43rd International Symposium on
  Mathematical Foundations of Computer Science ({MFCS} 2018)}, volume 117 of
  \emph{LIPIcs}, pages 45:1--45:17. Schloss Dagstuhl - Leibniz-Zentrum fuer
  Informatik, 2018.

\end{thebibliography}
% \newpage
% \appendix
% \begin{bibunit}
% \section*{Appendix}
% \appendixProofText
% \putbib[bib]
% \end{bibunit}
\end{document}